\documentclass[11pt]{article}

\usepackage{amsfonts, amsthm}
\usepackage{thmtools}
\usepackage{thm-restate}
\usepackage{amssymb}
\usepackage{amsmath}  
\usepackage{graphicx, comment, xcolor}
\usepackage{authblk}
\usepackage{mathrsfs}
\usepackage[colorlinks=true, allcolors=blue]{hyperref}
\usepackage{cleveref}

\usepackage[sort&compress, numbers]{natbib}

\usepackage[T1]{fontenc}

\usepackage[shortlabels]{enumitem}
\usepackage{tikz}
\usetikzlibrary{positioning,calc}
\usetikzlibrary{intersections}
\tikzset{every picture/.style={line width=0.75pt}} %set default line width to 0.75pt     

\usepackage{fullpage}

\newtheorem{theorem}{Theorem}
\newtheorem{lemma}{Lemma}
\newtheorem{corollary}{Corollary}
\newtheorem{proposition}{Proposition}
\newtheorem{definition}{Definition}
\theoremstyle{remark}

\newtheorem{claim}{Claim}

\newcommand{\tr}{\mathrm{tr}}
\newcommand{\x}{\mathbf{x}}
\newcommand{\p}{\mathbf{p}}
\newcommand{\ac}{\mathbf{a}}

\newcommand{\R}{\mathbf{R}}

\newcommand{\dd}{\text{\rm{d}}}
\newcommand{\poly}{\text{\rm{poly}}}

\newcommand{\diag}{\text{{\rm diag}}}

\newcommand{\ket}[1]{|#1\rangle}
\newcommand{\bra}[1]{\langle#1|}
\newcommand{\ketbra}[2]{|#1\rangle\langle#2|}
\newcommand{\bfa}{\mathbf{a}}
\newcommand{\braket}[2]{\langle#1|#2\rangle }

\title{Optimal convergence rates in trace distance and relative entropy for the quantum central limit theorem}

\author{Salman Beigi$^{1}$, Milad M. Goodarzi$^{1,2}$, Hami Mehrabi$^3$}
\affil{\it \footnotesize $^1$School of Mathematics, Institute for Research in Fundamental Sciences (IPM), P.O. Box 19395-5746, Tehran, Iran \\
\it \footnotesize $^2$Centre for Quantum Technologies, National University of Singapore, Singapore 117543, Singapore \\
\it \footnotesize $^3$School of Electrical and Computer Engineering, Cornell University, Ithaca, New York 14850, USA\\}

\date{}

\begin{document}
\maketitle

\begin{abstract}
A quantum analogue of the Central Limit Theorem (CLT) for bosonic system, first introduced by Cushen and Hudson (1971), states that the $n$-fold convolution $\rho^{\boxplus n}$ of an $m$-mode quantum state $\rho$—with zero first moments and finite second moments—converges weakly, as $n$ increases, to a Gaussian state $\rho_G$ with the same first and second moments as those of $\rho$, called its Gaussification. Recently, this result has been extended with estimates of the convergence rate in various distance measures. In this paper, we establish optimal rates of convergence in both the trace distance and quantum relative entropy. Specifically, we show that for a centered $m$-mode quantum state with finite third-order moments, the trace distance between $\rho^{\boxplus n}$ and $\rho_G$ decays at the optimal rate of $\mathcal{O}(n^{-1/2})$. Furthermore, for states with finite fourth-order moments (order $4+\delta$ for an arbitrary small $\delta>0$ if $m>1$), we prove that the relative entropy between $\rho^{\boxplus n}$ and $\rho_G$ decays at the optimal rate of $\mathcal{O}(n^{-1})$. Both of these rates are proven to be optimal, even when assuming the finiteness of all moments of $\rho$. These results relax previous assumptions on higher-order moments, yielding convergence rates that match the best known results in the classical setting. By giving explicit examples we also show that our moment assumptions are essentially minimal. We show that for any $\theta>0$, there exists a quantum state $\rho$ with finite moments of order less than $3-\theta$, such that the convergence rate of $\rho^{\boxplus n}$ to $\rho_G$ in trace distance is not $\mathcal O(n^{-1/2})$. Similarly, we show that for any $\theta>0$, there exists a quantum state $\rho$ with finite moments of order less than $4-\theta$, such that the relative entropy between $\rho^{\boxplus n}$ to $\rho_G$ does not decay at the rate $\mathcal O(n^{-1})$. Our proofs draw on techniques from the classical literature, including Edgeworth-type expansions of quantum characteristic functions, adapted to the quantum context. A key technical step in the proof of our entropic CLT is establishing an upper bound on the relative entropy distance between a general quantum state and its Gaussification, which is of independent interest.
\end{abstract}

{\footnotesize
\tableofcontents
}

\section{Introduction}

\subsection{Classical central limit theorems}

As a landmark in probability theory, the Central Limit Theorem (CLT) states that the normalized sum $S_n = \frac{X_1 + \cdots + X_n}{\sqrt{n}}$ of independent and identically distributed random variables $X_1, X_2, X_3, \dots$ with zero mean, converges to a Gaussian random variable $Z$ with the same mean and variance as that of $X_1$~\cite{Feller}. Depending on the structure of the underlying distribution, this convergence can be characterized using various topologies or metrics. Two notable such CLTs are the Lindberg--Levy CLT~\cite{lindeberg1922}, which asserts convergence in distribution under the assumption of a finite second moment of $X_1$, and the Berry--Esseen CLT~\cite{berry_accuracy_1941}, which asserts the uniform convergence of cumulative distribution functions at a rate of $\mathcal O(1/\sqrt{n})$, assuming the finiteness of the third absolute moment.  

Subsequently, several authors have sought stronger forms of the CLT. When measured in terms of the total variation distance ($1$-norm), and assuming the finiteness of the third absolute moment along with the presence of an absolutely continuous part in the density of $X_1$, the convergence rate of $\mathcal O(1/\sqrt{n})$ has been established in~\cite{MaSi}, specifically expressed as
\begin{align}\label{eq:Best-Trace}
	\| S_n - Z\|_1 = \mathcal{O}\Big(\frac{1}{\sqrt{n}}\Big) \quad \text{as} \quad n\rightarrow\infty.
\end{align}
See~\cite{BhRa} and references therein for a generalization of this result to random vectors in $\mathbb{R}^d$.  

A more recent advancement in this area was made by Bobkov, Chistyakov, and G\"{o}tze~\cite{BCG}, who determined the optimal convergence rate for the entropic CLT. They demonstrated that if $X_1$ has a finite absolute moment of order $4$ in addition to a finite (differential) entropy, then 
\begin{equation}\label{eq:Best-Relative}
	D(S_n\| Z) = \mathcal{O}\Big(\frac{1}{n}\Big) \quad \text{as} \quad n\rightarrow\infty,
\end{equation}
where $D(\cdot\| \cdot)$ denotes the Kullback--Leibler divergence. In the multi-dimensional case, i.e., for random vectors in $\mathbb{R}^d$, the aforementioned result is implied by the findings of~\cite{bobkov2013rate}, assuming the finiteness of the absolute moment of order $4+\delta$ for any $\delta>0$.

\subsection{Cushen--Hudson quantum CLT}

Quantum analogues of the classical CLT have been developed in a wide variety of settings over the past decades. The earliest results are due to Cushen and Hudson~\cite{CH} and Hepp and Lieb~\cite{hepp1973phase,hepp1973superradiant}, with a later algebraic variant by Giri and von Waldenfels~\cite{giri1978algebraic}. Since then, several quantum CLTs have appeared in diverse contexts such as quantum statistical mechanics~\cite{goderis1989central,matsui2002bosonic,cramer2010quantum,jakvsic2009central,arous2013central}, quantum field theory~\cite{derezinski1985boson,streater1987entropy,michoel2004central}, von Neumann algebras~\cite{goderis1989non,jakvsic2010quantum}, free probability theory~\cite{voiculescu1992free}, non-commutative stochastic processes~\cite{accardi1994quantum}, and quantum information theory~\cite{hayashi2006quantum,hayashi2009quantum,campbell2013continuous}. These results differ in several structural aspects, including the notion of independence they rely on, the topology or mode of convergence in which the limit is formulated, and types of limiting objects. They are also tailored to fundamentally different applications. For more details on these formulations of quantum CLT we refer to \cite{jakvsic2010quantum,lenczewski1995quantum} and references therein.

In this work, we focus on the framework of Cushen and Hudson~\cite{CH}, where the central limit behavior for quantum bosonic systems is established. Their formulation is physically relevant particularly in quantum optics where beam splitter networks implement the averaging process, and the limiting Gaussian state is fully characterized by its covariance matrix. Moreover, in the Cushen--Hudson framework one can meaningfully ask for quantitative rates in operationally relevant metrics such as trace distance and relative entropy, which are the focus of this work. 

Cushen and Hudson demonstrated that if $\rho$ is a centered $m$-mode quantum state (i.e., its first-order moments vanish) with a finite covariance matrix, then the sequence of its $n$-fold symmetric convolutions $\rho^{\boxplus n}$ converges to a centered Gaussian state $\rho_G$ that has the same covariance matrix as $\rho$, which is referred to as the Gaussification of $\rho$. The definition of $\rho^{\boxplus n}$ will be presented in Section~\ref{secND}, but for now, we note that physically, the state $\rho^{\boxplus n}$ arises from the action of a sequence of beam splitters on $n$ identical copies of $\rho$, followed by tracing out all but the first $m$ output modes.

This quantum CLT was initially expressed in terms of convergence in the weak operator topology of the Canonical Commutation Relations (CCR) algebra, which corresponds to pointwise convergence of the associated quantum characteristic functions. For quantum states, this topology coincides with the one induced by the trace distance (see~\cite[Section 3]{CRLimitTheorem} and ~\cite[Lemma 4]{lami_all_2018}). However, this theorem does not provide any bounds on the rate of convergence.

It is worth mentioning that the setting of Cushen and Hudson fits into the abstract algebraic CLT of Giri and von Waldenfels~\cite{giri1978algebraic}, which establishes the convergence of moments for general non-commutative $*$-algebras. Here, we work in the more specific CCR algebra in the framework of Cushen and Hudson because their theorem concerns convergence of the states themselves, and therefore has a genuine analytic structure. This analytic structure is essential for obtaining quantitative bounds in trace distance and relative entropy, which is the focus of this work. We emphasize that such quantitative bounds are not implied solely by the convergence of moments that is established in~\cite{giri1978algebraic}. 
Therefore, the Cushen--Hudson framework is not only operationally relevant as discussed above, but also the appropriate setting for establishing quantitative convergence bounds.

\subsection{Main results: optimal convergence rates under minimal assumptions}

Very little was known about convergence rates in the Cushen--Hudson CLT until recently, when Becker, Datta, Lami, and Rouzé~\cite{CRLimitTheorem} established the asymptotic bound $\|\rho^{\boxplus n} - \rho_G\|_2 = \mathcal{O}(1/\sqrt{n})$ for the convergence rate of the quantum CLT in terms of the Hilbert–Schmidt norm $\|\cdot\|_2$, assuming that the third-order moments of $\rho$ are finite. Building on this result, they also obtained the asymptotic bounds $\|\rho^{\boxplus n} - \rho_G\|_1 = \mathcal{O}\big(n^{-1/2(m+1)}\big)$ and $D(\rho^{\boxplus n}\|\rho_G) = \mathcal{O}\big(n^{-1/2(m+1)} \ln(n)\big)$ for the trace norm $\|\cdot\|_1$ and the quantum relative entropy $D(\cdot\|\cdot)$, under the same moment assumption. In these bounds, and throughout the paper, $m$ stands for the number of modes.

Subsequently, two authors of the current paper derived the asymptotic bounds $\|\rho^{\boxplus n} - \rho_G\|_1 = \mathcal{O}(1/\sqrt{n})$ and $D(\rho^{\boxplus n} \|\rho_G) = \mathcal{O}(1/n)$ in~\cite{beigi2023towards}. While these bounds align with the corresponding convergence rates in the classical setting and are shown to be optimal in~\cite{beigi2023towards}, they are established under stronger assumptions. Specifically, the former bound is derived under the condition that the moments of order $\max\{3, 2m\}$ is finite, whereas the latter relies on a Poincaré-type inequality in the quantum context, which is sufficiently strong to ensure the finiteness of all moments. When compared to classical results (with the exception of the single-mode case $m=1$), these assumptions appear quite stringent. Therefore, it is both important and highly desirable to relax these conditions.

In this paper, we overcome the limitations of the previous contributions described above and prove optimal rates of convergence for the quantum CLT with minimal assumptions. Our first main result concerns the convergence rate in the trace distance.

\begin{restatable}{theorem}{MainTheoremTrace} 
\label{MainTheoremTrace}
Let $\rho$ be a centered $m$-mode quantum state with finite moments of order $3$. Let $\rho_{G}$ be the Gaussification of $\rho$. Then we have
	\begin{equation}\label{Result2}
		\big\| \rho^{\boxplus n} -  \rho_G \big\|_1 = \mathcal{O} \Big(\frac{1}{\sqrt n}\Big) \quad \text{as} \quad n\rightarrow\infty.
	\end{equation}
\end{restatable}
We note that this theorem for the single-mode case ($m=1$) has already been established in~\cite{beigi2023towards}. However, it is stronger than the corresponding result for the multi-mode case presented in that work. Furthermore, as illustrated by the example in~\cite{beigi2023towards}, the convergence rate of $\mathcal{O}(1/\sqrt{n})$ mentioned above is tight, even under the assumption that all moments of $\rho$ are finite. Additionally, the assumptions made here aligned with the minimal moment conditions required to demonstrate the convergence rate of $\mathcal{O}(1/\sqrt{n})$ in the classical CLT regarding total variation distance.

The proof of Theorem \ref{MainTheoremTrace} is based on the classical method of normal approximation, as detailed in~\cite{BhRa}, and incorporates ideas from~\cite{CRLimitTheorem, beigi2023towards}. To estimate ${\big\| \rho^{\boxplus n} -  \rho_G \big\|}_1$, we begin by projecting $\rho$ onto the subspace spanned by Fock states with a total particle number not exceeding $n$. This results in a truncated state of finite rank, where all moments remain finite. We then carry out a careful analysis of the moments and covariance matrix of this truncated state and find controls on their asymptotic behavior as $n$ increases. Another key aspect of the proof involves extending the concept of Edgeworth-type expansions from the classical literature to the quantum realm by analyzing an Edgeworth-type expansion for the quantum characteristic function. Ultimately, we leverage an estimate of the trace norm in terms of the Hilbert--Schmidt norm, along with the quantum Plancherel identity and Hölder's inequality, to utilize the previously mentioned expansion and derive the desired bound.

In our second main result, we consider a stronger mode of convergence, namely, convergence in relative entropy distance.

\begin{restatable}{theorem}{MainTheorem}
\label{MainTheorem}
Let $\rho$ be a centered $m$-mode quantum state with finite moments of order $4+\delta$, where $\delta=0$ if $m=1$, and $\delta>0$ otherwise. Let $\rho_{G}$ be the Gaussification of $\rho$. Then we have
	\begin{equation}\label{AsympBoundEntropy}
		D\big(\rho^{\boxplus n} \big\| \rho_{G}\big) = \mathcal{O}\Big(\frac{1}{n}\Big) \quad \text{as} \quad n\rightarrow\infty.
	\end{equation}
\end{restatable}
It is important to point out that, by virtue of the quantum Pinsker inequality, the bound in~\eqref{AsympBoundEntropy} implies the one in~\eqref{Result2}, though it is established under a stricter regularity assumption. Additionally, this bound aligns with the classical case in terms of both the convergence rate and conditions. Furthermore, based on the example in~\cite{beigi2023towards}, our convergence rate remains optimal, even when assuming the finiteness of all moments of $\rho$.

The proof of Theorem \ref{MainTheorem} draws inspiration from the approach developed by Bobkov, Chistyakov, and Götze in \cite{BCG}, where they link the relative entropy distance to a Sobolev norm of the characteristic functions. Their analysis involves examining higher-order derivatives of the characteristic function through Edgeworth-type expansions, yielding refined approximations of the norm. We follow a similar strategy. A key technical element in our proof is establishing an upper bound on the relative entropy distance between a general quantum state and its Gaussification, achieved at the operator level. We remark in passing that as an immediate corollary of this result, a general bound on the relative entropy of non-Gaussianity \cite{genoni2008quantifying,marian2013relative} follows, which is of independent interest. Finally, by leveraging the quantum Plancherel identity and employing the truncation technique discussed earlier, we are able to perform a semiclassical analysis of the quantum characteristic function to obtain the desired bound.

We also by giving explicit examples show that the moment assumptions in Theorems~\ref{MainTheoremTrace} and~\ref{MainTheorem} are essentially minimal and cannot be weakened. 

\begin{theorem}\label{thm:examples}
For any $0<\theta<1$ the following hold:
\begin{enumerate}[label=(\roman*)]
\item  There is a single-mode quantum state $\rho$ whose moments of order $\kappa$, for any $\kappa<3-\theta$, are finite, while
\[
\sup_{n} \sqrt n \big\| \rho^{\boxplus n } - \rho_G \big\|_1 = + \infty.
\]
\item  There is a single-mode quantum state $\rho$ whose moments of order $\kappa$, for any $\kappa<4-\theta$, are finite, while
\[
\sup_{n} n  D\big(\rho^{\boxplus n} \big\| \rho_G\big)= + \infty.
\]
\end{enumerate}
\end{theorem}

To construct these examples, we employ ideas from the classical literature, particularly techniques from~\cite{ibragimov1975independent} for the trace distance, and from~\cite{bobkov2013rate} for the relative entropy.

\subsection{Structure of the paper}

In Section~\ref{secND}, we establish notation for bosonic quantum systems and phase space formalism and introduce the concept of quantum convolution in Subsection~\ref{subsec:convolution}. In Section~\ref{Sec:Edge} we prove some basic results on moments of a quantum state and introduce the Edgeworth-type expansion for quantum states. Section~\ref{Sec:State-Truncation} focuses on the truncation technique discussed earlier. In Section~\ref{upperbound}, we establish an upper bound on the relative entropy between a general quantum state and a Gaussian state in terms of the Hilbert--Schmidt norm. Sections~\ref{Sec:Proof-Handle} and~\ref{proofMain} provide the proofs of Theorems~\ref{MainTheoremTrace} and~\ref{MainTheorem}, respectively, drawing on results from the preceding sections. The examples showing the essential minimality of our assumptions on the moments are presented in Section~\ref{secMinAssumption}.
Section~\ref{secCo} concludes with final remarks, and additional technical details are provided in the appendices.

%******************************************************************************************

\section{Preliminaries} \label{secND}

In this section, we review essential definitions and concepts regarding continuous-variable bosonic quantum systems, which will be used throughout the paper. Our primary reference for this topic is~\cite{Serafini}.

For two self-adjoint operators $A$ and $B$ acting on a Hilbert space, we write $A \geq B$ if the operator $A - B$ is positive. For an operator $T$ and $p \geq 1$, the \emph{Schatten $p$-norm} of $T$ is defined by
\[
\| T \|_{p} := \Big(\tr\big(|T|^{p}\big)\Big)^{\frac{1}{p}},
\]
where $|T| := \sqrt{T^\dagger T}$. In particular, for $p=1$ and $p=2$, we obtain the \emph{trace norm} $\|\cdot\|_1$ and the \emph{Hilbert--Schmidt norm} $\|\cdot\|_2$, respectively. If $\| T \|_{1}$ is finite, $T$ is called a \emph{trace-class} operator.

The Hilbert space of an $m$-mode bosonic quantum system is given by $\mathcal{H}_m = L^2(\mathbb{R}^m)$. The state space for this quantum system consists of positive trace-class operators acting on $\mathcal{H}_m$ with $\tr(\rho) = 1$. The quadrature operators $\x_j$ and $\p_j$ act on $\mathcal{H}_m$ for $j = 1, \dots, m$, and they satisfy the canonical commutation relation
\begin{equation}\label{eq:CCR-Quad}
	[\R , \R^{\top}] = i \Omega_m,
\end{equation}
where $\R := \big( \x_1, \p_1, \cdots, \x_m, \p_m \big)^{\top}$, and $\Omega_m$ is the symplectic form defined as
\begin{equation}
	\Omega_m := \bigoplus_{j=1}^{m} \Omega_1, \qquad \quad \Omega_1 = \begin{pmatrix}
		0 & 1 \\
		-1 & 0 
	\end{pmatrix}.
\end{equation}
Here, $[\R , \R^{\top}]$ represents the coordinate-wise commutation defined by $[A, B] = AB - BA$. The associated annihilation and creation operators $\ac_j$ and $\ac_j^\dagger$ for the $j$-th mode are given by
\[
\ac_j := \frac{1}{\sqrt 2} \big( \x_j + i \p_j \big), \qquad \ac_j^\dagger := \frac{1}{\sqrt 2} \big( \x_j - i \p_j \big),
\]
and the canonical commutation relations can be expressed in terms of $\ac_j$ and $\ac_j^\dagger$ as
\begin{equation}\label{eq:CCR-AnnCre}
	[\ac_j, \ac_k^\dagger] = \delta_{jk} \mathbb{I}, \qquad  [\ac_j, \ac_k] = 0.
\end{equation}
Here, $\mathbb{I}$ denotes the identity operator on $\mathcal{H}_m$, and $\delta$ is the Kronecker delta function.

The number operator is defined by
\[
N_m := \sum_{j=1}^m \ac_j^\dagger \ac_j.
\]
The eigenstates of $N_m$ are the Fock states, denoted by $\ket{k_1, \dots, k_m} := \ket{k_1} \otimes \cdots \otimes \ket{k_m}$, which form a complete orthonormal basis for $\mathcal{H}_m$. These Fock states, also known as number states, satisfy
\[
\ac_j \ket{k_j} = \sqrt{k_j} \ket{k_{j}-1}, \qquad \ac_j^\dagger \ket{k_j} = \sqrt{k_j+1} \ket{k_{j}+1},
\]
and moreover, $N_m \ket{k_1, \dots, k_m} = (k_1 + \cdots + k_m) \ket{k_1, \dots, k_m}$.

The von Neumann entropy of a quantum state $\rho$ is defined as
\[
S(\rho) = -\tr \big( \rho \ln \rho\big),
\]
while the quantum relative entropy between two quantum states $\rho$ and $\sigma$ is given by
\[
D\big(\rho \| \sigma \big) = \tr \big( \rho (\ln \rho - \ln \sigma) \big).
\]

%*********************
\subsection{Phase space formalism}\label{subsec:Phase space formalism}

For any $z = (z_1, z_2, \dots , z_m)^{\top} \in \mathbb{C}^m$, the $m$-mode displacement operator $D_z$ is defined by
\[
D_z := \bigotimes_{j=1}^m \exp \big( z_j \ac_j^\dagger - \bar{z}_j \ac_j \big).
\]
The term ``displacement'' highlights that conjugation by $D_z$ shifts the annihilation and creation operators. Specifically, for any $j \in \{ 1, \dots, m\}$ and $z \in \mathbb{C}^m$, the following relations hold:
\begin{align}\label{eq:Dz-displace} 
	D_z^\dagger \ac_j D_z = \ac_j + z_j, \qquad  D_z^\dagger \ac_j^\dagger D_z = \ac_j^\dagger + \bar{z}_j .
\end{align}
Using~\eqref{eq:CCR-AnnCre} and the Baker--Campbell--Hausdorff formula~\cite{Serafini}, for any $z , w \in \mathbb{C}^m$, we obtain
$$
D_z D_w = e^{\frac{1}{2} (z^{\top} \bar w - \bar z^{\top} w)} D_{z+w}.
$$
Additionally, we can express $D_z$ in the following alternative forms:
\begin{align}\label{eq:Dz-alt-form}
	D_z = \prod_{j=1}^m e^{- \frac 12|z_j|^{2}} e^{z_j \ac_j^{\dagger}} e^{-\bar z_j \ac_j} = \prod_{j=1}^m e^{\frac12 {|z_j|}^{2}} e^{-\bar z_j \ac_j} e^{z_j \ac_j^{\dagger}}.
\end{align}

The \emph{quantum characteristic function} of a trace-class operator $T$ acting on $\mathcal{H}_m$ is defined as
$$
\chi_T(z) := \tr \big( T D_z \big).
$$
Analogous to the classical case, the operator $T$ can be recovered from its characteristic function as
$$
T = \frac{1}{\pi^m} \int_{\mathbb{C}^m} \chi_T(z) D_{-z} \dd^{2m} z.
$$
The (complex) Fourier transform of $\chi_T(z)$, known as the Wigner function of $T$, is denoted by $W_{T}(z)$ and given by
\begin{equation}\label{eq:WignerFunc}
W_{T}(z) := \frac{1}{\pi^{2m}} \int \chi_{T}(w) e^{z^{\top} \bar w - \bar z^{\top} w} \dd^{2m} w.	
\end{equation}
Obviously, we can write the characteristic function using Wigner function as 
\begin{equation}\label{eq:InvWignerFunc}
    \chi_T(z) = \int W_T(u) e^{u^\dagger z - z^\dagger u} \dd^{2m} u
\end{equation}
It is straightforward to verify that $W_{\rho}(z) \in \mathbb{R}$ for any $m$-mode quantum state $\rho$, and
\begin{equation}\label{eq:WignerNormal}
\int W_{\rho}(z) \dd^{2m} w = \chi_{\rho}(0)= \tr(\rho) = 1.
\end{equation}
However, a key issue with $W_{\rho}(z)$ is that it can assume negative values, preventing it from serving as a genuine probability distribution function. Consequently, $W_{\rho}(z)$ is sometimes referred to as the \emph{Wigner quasiprobability distribution}. Another important property of the characteristic and Wigner functions is the quantum Plancherel identity:
\begin{equation} \label{eq:Plancherel}
	\| T\|^{2}_2 = \frac{1}{\pi^m} \int \big| \chi_{T}(z)\big|^{2} \dd^{2m}z=\pi^m \int W_{T}(z)^{2} \dd^{2m}z.
\end{equation}

We will use the following lemma to estimate the trace norm.
\begin{lemma} \emph{(\cite[Lemma 2]{beigi2023towards})} \label{lem:norm1-norm2}
Let $T$ be an operator such that $A^\dagger TA$ is trace-class, where $A= \bfa_1\cdots\bfa_m$. Then we have
$$\|T\|_1^2 \leq \Big(\frac{\pi^2}{6}\Big)^m \int |\chi_{A^\dagger TA}(z)|^2 \dd^{2m} z.$$
\end{lemma}

%*****************************
\subsection{Gaussian states}\label{subsec:Gaussian-states}

The vector of the \emph{first-order moments} of an $m$-mode quantum state $\rho$ is defined by
$$ \mathbf{d}(\rho) := \tr\big(\rho \R\big) := \big(\tr(\rho\x_1), \tr(\rho\p_1), \dots ,  \tr(\rho\x_m) , \tr(\rho \p_m)\big)^{\top} \in \mathbb R^{2m}.$$
A quantum state $\rho$ is said to be centered if $\mathbf{d}(\rho)$ is the zero vector. The \emph{covariance matrix} of $\rho$ is defined by
$$
\boldsymbol{\gamma}(\rho) := \tr\Big(\rho \big\{ \R - \mathbf{d}(\rho), {(\R - \mathbf{d}(\rho))}^{\top} \big\}\Big),
$$
where $\{\cdot,\cdot\}$ denotes the coordinate-wise anti-commutation $\{A, B\} = AB+BA$. As a consequence of~\eqref{eq:CCR-Quad}, the covariance matrix always satisfies the \emph{Robertson-Schr\"{o}dinger uncertainty relation}~\cite{Serafini}
\begin{equation}\label{eq:Uncertainty}
\boldsymbol{\gamma}(\rho) + i \Omega_m \geq 0.
\end{equation}

An $m$-mode quantum state is called Gaussian if its characteristic function (or equivalently, its Wigner function) is a Gaussian function. Similarly as in the classical setting, a Gaussian state is fully determined by its first moments and covariance matrix. Specifically, for a Gaussian state $\sigma$, we have
\begin{equation}\label{eq:Char-NonWilliam}
	\chi_{\sigma}(z) = \exp\bigg(-\frac 14 \hat{z}^\dagger \Lambda_m^\dagger \boldsymbol{\gamma}(\sigma) \Lambda_m  \hat{z} + i\mathbf{d}(\sigma)^\top \Lambda_m \hat z\bigg),
\end{equation}
where $\hat{z} = (z_1, \bar{z}_1, \cdots, z_m, \bar{z}_m)^\top$ and
\begin{equation}\label{eq:Lambda_m}
\Lambda_m = \bigoplus_{j=1}^{m} \Lambda_1, \qquad \quad	\Lambda_1 = \frac{1}{\sqrt 2} \begin{pmatrix}
	  -i & i\\
	  -1 & -1 
	  \end{pmatrix}.
\end{equation}

\begin{definition}
For any $m$-mode quantum state $\rho$ with finite first-order moments and covariance matrix, the Gaussification of $\rho$, denoted by $\rho_G$, is defined as the Gaussian quantum state with the same first-order moments and covariance matrix as those of $\rho$. 
\end{definition}

A unitary $U$ is said to be Gaussian if it maps Gaussian states to Gaussian states under conjugation. It is well known that for any $m$-mode quantum state $\rho$, there exists a Gaussian unitary $U$ that both centers $\rho$ and diagonalizes its covariance matrix in terms of the symplectic eigenvalues and  unsqueezes it~\cite{Serafini}. That is, there exists a Gaussian unitary $U$ such that
	\begin{equation}\label{eq:Moments-Williamson}
		\mathbf{d}(U \rho U^{\dagger}) = 0, \quad  \boldsymbol{\gamma}(U \rho U^{\dagger}) = \diag(\nu_1, \nu_1, \dots , \nu_m, \nu_m).
	\end{equation}
We refer to $U \rho U^{\dagger}$ as the Williamson form of $\rho$. When $\rho$ is in the Williamson form and its covariance matrix is diagonal as in~\eqref{eq:Moments-Williamson}, $\rho_G$ is a thermal state and can be expressed as
$$
\rho_G = \tau_1 \otimes \cdots \otimes \tau_m,
$$
where
\[
	\tau_j = \big( 1 - e^{-\beta_j}) \exp\big(-\beta_j \ac_j^\dagger \ac_j\big),
\]
and $\beta_j >0$ is given by
\[
	\nu_j = \frac{1 + e^{-\beta_j}}{1 - e^{-\beta_j}} \geq 1.
\]
The condition $\nu_j \geq 1$ follows from~\eqref{eq:Uncertainty}, since $\boldsymbol{\gamma}(\rho)$ is always positive definite with eigenvalues greater than or equal to 1. Notably, if $\nu_j = 1$ (equivalently, $\beta_j = +\infty$), then $\tau_j = \ketbra{0}{0}$ is the \emph{vacuum state}. Specifically, $\nu_j = 1$ implies that $\operatorname{tr}(\rho \ac_j^\dagger \ac_j) = 0$, so the support of the $j$-th mode of $\rho$ lies in the kernel of $\ac_j$, which is spanned by the vacuum state. Thus, if $\rho$ is in the Williamson form and $\nu_j = 1$, then the $j$-th mode is in the pure state $\ket{0}$ and is decoupled from the other modes.

The characteristic and Wigner functions of the thermal state $\rho_G = \tau_1\otimes \cdots \otimes \tau_m$ are given by
\begin{align}\label{eq:Gaussian-characteristic-Wigner}
\chi_{\rho_G}(z) = \prod_{j=1}^m e^{-\frac 12 \nu_j|z_j|^2}, \qquad \quad W_{\rho_G}(z)= \prod_{j=1}^m \frac{2}{\pi \nu_j} e^{-\frac{2}{\nu_j}|z_j|^2}.
\end{align}

%***********************
\subsection{Quantum convolution}\label{subsec:convolution}

We define an addition rule for two $m$-mode quantum states using a specific Gaussian unitary associated with a beam splitter (see~\cite{KS14} and references therein). Let $\rho$ and $\sigma$ be two $m$-mode quantum states, and let $0 \leq \eta \leq 1$. The quantum convolution of $\rho$ and $\sigma$ with parameter $\eta$ is defined as
\[
\rho \boxplus_\eta \sigma := \tr_2 \big( U_\eta \, (\rho \otimes \sigma) \, U_\eta^\dagger \big),
\]
where
\[
U_{\eta} := \exp\bigg(\arccos(\sqrt{\eta}) \sum_{j=1}^{m}\big(\ac_{j,1}^{\dagger} \ac_{j,2} - \ac_{j,1} \ac_{j,2}^{\dagger}\big)\bigg).
\]
Here, $U_\eta$ is a unitary operator acting on $2m$ modes, which we treat as two $m$-mode subsystems. The operators $\ac_{j,1}$ and  $\ac_{j,2}$ for $j = 1, \dots, m$ denote the annihilation operators on the first and second subsystems, respectively, and $\tr_2(\cdot)$ represents the partial trace over the second subsystem, consisting of $m$ modes. The Gaussian unitary operator $U_{\eta}$ is equivalently described by its action on the annihilation and creation operators:
\begin{align} 
	\begin{cases} \label{BeamSplitter-quadratures}
		U_{\eta} \ac_{j,1} U_{\eta}^{\dagger} = \sqrt{\eta} \ac_{j,1} - \sqrt{1 - \eta} \ac_{j,2},\\
		U_{\eta} \ac_{j,2} U_{\eta}^{\dagger} =  \sqrt{1 - \eta} \ac_{j,1} + \sqrt{\eta} \ac_{j,2}. 
	\end{cases}
\end{align}

A key property of quantum convolution is that it commutes with Gaussian unitaries~\cite{CRLimitTheorem}, meaning that for any Gaussian unitary $V$, we have
\begin{equation}\label{eq:Conv-Unit}
	V (\rho \boxplus_{\eta} \sigma) V^{\dagger} = (V \rho V^{\dagger}) \boxplus_{\eta} (V \sigma V^{\dagger}).
\end{equation}
Additionally, from~\eqref{BeamSplitter-quadratures}, it can be readily verified that
\begin{equation}\label{eq:characteristic-convolution}
	\chi_{\rho \boxplus_{\eta} \sigma}(z) = \chi_{\rho}(\sqrt{\eta} z) \, \chi_{\sigma}(\sqrt{1 - \eta} z).
\end{equation}
Thus, the characteristic function of the resulting quantum state retains a similar form to that seen in the classical case, where a new random variable is defined as the normalized sum of two independent random variables.

The quantum convolution for the parameter $\eta = 1/2$ is referred to as the symmetric convolution, and we denote it simply by $\rho \boxplus \sigma$. Extending this notion, we define the symmetric convolution of $\rho_1, \dots, \rho_n$ inductively by  
\[
\rho_1 \boxplus \cdots \boxplus \rho_n := \big( \rho_1 \boxplus \cdots \boxplus \rho_{n-1} \big) \boxplus_{1 - \frac{1}{n}} \rho_n.
\]
When $\rho_1 = \cdots = \rho_n = \rho$, we denote their symmetric convolution by $\rho^{\boxplus n}$, representing the $n$-fold convolution of $\rho$. From~\eqref{eq:characteristic-convolution}, it follows that
\begin{equation}\label{eq:char-symmetric}
	\chi_{\rho^{\boxplus n}}(z) = \chi_{\rho} \left( \frac{z}{\sqrt{n}} \right)^n.
\end{equation}

An important property of the symmetric convolution, shown in \cite[Lemma 16]{CRLimitTheorem}, is that for any two $m$-mode quantum states $\rho$ and $\sigma$, we have 
\begin{equation}\label{eq:positive-Wigner}
	W_{\rho \boxplus \sigma}(z) \geq 0, \quad \forall z \in \mathbb{C}^m,
\end{equation}
where $W_{\rho \boxplus \sigma}(z)$ denotes the Wigner function of $\rho \boxplus \sigma$. This property is essential for interpreting $W_{\rho \boxplus \sigma}(z)$ as a valid probability distribution, enabling the application of established results on the probability distribution $W_{\rho \boxplus \sigma}(z)$ and its corresponding characteristic function, which coincides with $\chi_{\rho \boxplus \sigma}(z)$.

We conclude this section by noting an interesting feature of the characteristic function of quantum states: for any nonzero $z$, the absolute value of $\chi_\rho(z)$ is strictly less than $1$.

\begin{proposition} \emph{(\cite[Proposition 15]{CRLimitTheorem})}
Let $\rho$ be a quantum state with finite covariance matrix. Then, for each $\epsilon > 0$, there exists $\delta > 0$ such that
\begin{equation}\label{CharSupOrigin}
	\sup_{z \in \mathbb{C}^m \setminus B(0, \epsilon)} |\chi_{\rho} (z)| \leq 1 - \delta,
\end{equation}
where $\delta$ depends on $\epsilon$ and the covariance matrix of $\rho$.
\label{prop:CharSupOrigin}
\end{proposition}

Interestingly, this property does not hold in such a general form for classical probability distributions. Using this property along with~\eqref{eq:char-symmetric}, we see that for $\|z\| > \epsilon \sqrt{n}$, $|\chi_{\rho^{\boxplus n}}(z)| \leq (1 - \delta)^n$ becomes exponentially small. This inequality is critical in our analysis of convergence rates in the quantum CLT.

%******************************************
\section{Edgeworth-type expansion}\label{Sec:Edge}

The Edgeworth expansion is a key tool in classical CLT analysis. In fact, it is often regarded as foundational to classical CLTs, underpinning early formulations such as the Lindeberg–Levy theorem~\cite{lindeberg1922} and the Berry–Esseen theorem~\cite{berry_accuracy_1941}. In this section, we extend the concept of Edgeworth expansion to the quantum setting and establish its properties. To proceed, however, we first need to define quantum state moments beyond the second order.

%*********************
\subsection{Moments of a quantum state}

Let $\rho$ be a quantum state and $\kappa > 0$. Following~\cite{CRLimitTheorem}, we say that $\rho$ has finite moments up to order $\kappa$ if\footnote{The moment of order $\kappa$ is sometimes defined as $\tr \Big( \rho N_m^{\frac{\kappa}{2}}\Big)$ in the literature. We note, however, that the finiteness of moments of order $\kappa$ according to these two definitions coincide as we have $\tr \big( \rho N_m^{\frac{\kappa}{2}}\big) \leq M_\kappa (\rho) \leq (2m)^{\frac{\kappa}{2}} \Big(\tr \big( \rho N_m^{\frac{\kappa}{2}}\big) +1\Big) $.}
$$
M_\kappa (\rho) := \tr \Big( \rho (N_m + m)^{\frac{\kappa}{2}}\Big)<\infty.
$$
It is straightforward to verify that if $M_2(\rho)<+\infty$, the covariance matrix $\boldsymbol{\gamma}(\rho)$ is well-defined. Furthermore, as shown in the following lemma, the finiteness of the $\kappa$-th moment of $\rho$ implies the finiteness of all lower-order moments.
\begin{lemma}\label{lem:moments-monotone}
	For any quantum state $\rho$ and $0<\kappa_1<\kappa_2$, we have
	$$M_{\kappa_1}(\rho) \leq M_{\kappa_2}(\rho)^{\kappa_1/\kappa_2}.$$
\end{lemma}

\begin{proof}
	Using H\"older's inequality, we find
	\begin{align*}
		M_{\kappa_1}(\rho) & = \tr\big(\rho (N_m+m)^{\kappa_1}\big) \\
		&\leq  \big \|\rho^{(\kappa_2-\kappa_1)/2\kappa_2}\big\|^2_{2\kappa_2/(\kappa_2-\kappa_1)}\cdot \big\|  \rho^{\kappa_1/2\kappa_2} (N_m+m)^{\kappa_1} \rho^{\kappa_1/2\kappa_2} \big \|_{\kappa_2/\kappa_1}\\
		& = \tr\Big(   \big(\rho^{\kappa_1/2\kappa_2} (N_m+m)^{\kappa_1} \rho^{\kappa_1/2\kappa_2} \big)^{\kappa_2/\kappa_1}  \Big)^{\kappa_1/\kappa_2}\\
		& \leq \tr\big( \rho (N_m+m)^{\kappa_2}\big)^{\kappa_1/\kappa_2}\\
		& = M_{\kappa_2}(\rho)^{\kappa_1/\kappa_2},
	\end{align*}
	where the second inequality follows from the Araki--Lieb--Thirring inequality~\cite[Theorem 1]{araki_inequality_1990}.
\end{proof}

The following lemma provides a sufficient condition to ensure that the operator $\rho A_1 \dots A_\kappa$ is trace-class, where $A_1, \dots, A_\kappa \in \big\{\bfa_1, \bfa_1^\dagger, \dots, \bfa_m, \bfa_m^\dagger, \mathbb{I} \big\}$.

\begin{lemma}\label{lem:tr-A-kappa-bounded}
If $M_\kappa(\rho)<+\infty$ for some integer $\kappa\geq 2$, then 
$$\big|\tr\big(\rho A_1 \cdots A_\kappa \big)\big|=\mathcal O(M_\kappa(\rho)),$$ 
for any $A_1, \dots, A_\kappa\in \big\{\bfa_1, \bfa_1^\dagger, \dots, \bfa_m, \bfa_m^\dagger,\mathbb I \big\}$. 
\end{lemma}

\begin{proof}
Inspecting the action of annihilation and creation operators on number states, we find that there exist integers $-\kappa\leq  s_1, s_2, \dots, s_m\leq  \kappa$ such that the action of $A_1 \cdots A_\kappa$ on any Fock state $\ket{k_1, \dots, k_m}$ can be written as
\[
    A_1 \cdots A_\kappa \ket{k_1, \dots, k_m} = C(k_1, \dots, k_m)\ket{k_1+s_1, \dots, k_m+s_m},
\]
where $ C(k_1, \dots, k_m)\geq 0$, and it is equal to zero whenever $k_j + s_j$ is negative for some $j$. Moreover, using a simple AM--GM inequality, it is straightforward to see that $ C(k_1, \dots, k_m) \leq \sum_{j=1}^m (k_j + \kappa)^{\kappa/2}$. Now, applying the Cauchy--Schwarz inequality we can write
{\allowdisplaybreaks\begin{align*}
\big|\tr\big(\rho A_1 \cdots A_\kappa \big)\big| & = \bigg| \sum_{k_1, \dots, k_m=0}^\infty  C(k_1, \dots, k_m)  \bra{k_1, \dots, k_m} \rho \ket{k_1+s_1, \dots , k_m+s_m}    \bigg|\\
& \leq  \frac 12  \sum_{k_1, \dots , k_m=0}^\infty C(k_1, \dots, k_m)  \bra{k_1, \dots, k_m} \rho \ket{k_1, \dots, k_m}  \\
&\quad   +  \frac 12  \sum_{k_1, \dots , k_m=0}^\infty    C(k_1, \dots, k_m)  \bra{k_1+s_1, \dots , k_m+s_m} \rho \ket{k_1+s_1, \dots , k_m+s_m}  \\
&\leq  \sum_{k_1, \dots , k_m=0}^\infty     \bigg(\sum_{j=1}^m (k_j + 2\kappa)^{\kappa/2} \bigg)  \bra{k_1, \dots, k_m} \rho \ket{k_1, \dots, k_m} \\
&\leq  (2\kappa)^{\kappa/2} \sum_{k_1, \dots , k_m=0}^\infty     \bigg(\sum_{j=1}^m (k_j + 1)^{\kappa/2} \bigg)  \bra{k_1, \dots, k_m} \rho \ket{k_1, \dots, k_m} \\
&\leq  (2\kappa)^{\kappa/2} \sum_{k_1, \dots , k_m=0}^\infty      \bigg(\sum_{j=1}^m k_j + m\bigg)^{\kappa/2}  \bra{k_1, \dots, k_m} \rho \ket{k_1, \dots, k_m} \\
& = (2\kappa)^{\kappa/2} M_\kappa(\rho),
\end{align*}
as desired.
}
\end{proof}

In our analysis, we sometimes use conjugation with a displacement operator to center a quantum state. It is therefore essential to examine how this conjugation affects the moments of a quantum state.
\begin{lemma}\label{lem:moment-displaced}
For any integer $\kappa\geq2$, there is a polynomial $g_{m, \kappa}(\cdot)$  such that 
$$D_z^\dagger (N_m+m)^{\kappa/2} D_z \leq g_{m, \kappa}(\|z\|) (N_m+m)^{\kappa/2},$$
where $\|z\|^2 = \sum_{i=1}^m |z_i|^2$.
As a result, for any $m$-mode quantum state $\rho$, we have
$$ M_{\kappa}\big(D_z \rho D_z^\dagger\big)\leq g_{m, \kappa}(\|z\|) M_\kappa(\rho).$$
\end{lemma}

\begin{proof}
Note that $N_m + m = \sum_{j=1}^m \bfa_j \bfa_j^\dagger$, and the operators $\bfa_j \bfa_j^\dagger$ commute. Moreover, using the convexity of $x\mapsto x^{\kappa/2}$, for any set of commuting positive operators $A_1, \dots, A_m$, we have
$$
\bigg( \sum_{j=1}^m A_j  \bigg)^{\kappa/2} \leq m^{\kappa/2-1}  \sum_{j=1}^m A_j^{\kappa/2}  \leq m^{\kappa/2-1}\bigg( \sum_{j=1}^m A_j  \bigg)^{\kappa/2}.
$$
Therefore, it is sufficient to prove the lemma for $m = 1$. In fact, we can set $g_{m, \kappa}(\cdot) = m^{\kappa/2 - 1} g_{1, \kappa}(\cdot)$. Furthermore, for $m = 1$, it suffices to establish the theorem for even values of $\kappa$ since, assuming the existence of a polynomial $g_{1, \kappa+1}(\|z\|)$ satisfying the lemma for $\kappa+1$, we have for $0 < \lambda = \frac{\kappa}{\kappa+1} < 1$ that
\begin{align*}
D_z^\dagger(\bfa\bfa^\dagger)^{\kappa/2}D_z & = D_z^\dagger(\bfa\bfa^\dagger)^{\lambda(\kappa+1)/2}D_z = \bigg(  D_z^\dagger(\bfa\bfa^\dagger)^{(\kappa+1)/2}D_z  \bigg)^{\lambda} \\
&\leq g_{1, \kappa+1}(\|z\|)^\lambda (\bfa\bfa^\dagger)^{\lambda (\kappa+1)/2} =   g_{1, \kappa+1}(\|z\|)^\lambda (\bfa\bfa^\dagger)^{\kappa/2},
\end{align*}
where the inequality follows from the operator monotonicity of $x \mapsto x^\lambda$~\cite[Theorem V.1.9]{bhatia_matrix_1997}. Hence, we proceed by assuming that $m = 1$ and $\kappa = 2s$ is an even integer.

Using~\eqref{eq:Dz-displace} and the canonical commutation relation~\eqref{eq:CCR-AnnCre}, there exist polynomials $h_{k, \ell}(\cdot)$ such that
\begin{align*}
D_z(\bfa\bfa^\dagger)^{s}D_z^\dagger = \big((\bfa+z)(\bfa^\dagger+\bar z)\big)^{s} = \sum_{k+\ell \leq s} h_{k, \ell}(z, \bar z) \big( \bfa^k(\bfa^\dagger)^\ell + \bfa^\ell (\bfa^\dagger )^k \big).
\end{align*}
Then, by applying the operator inequality $A^\dagger B+B^\dagger A \leq A^\dagger A + B^\dagger B$ which is equivalent to $(A-B)^\dagger(A-B)\geq 0$, we obtain
\[
D_z(\bfa\bfa^\dagger)^{s}D_z^\dagger \leq  \sum_{k+\ell \leq s} |h_{k, \ell}(z, \bar z)| \big( \bfa^k(\bfa^\dagger)^k + \bfa^\ell (\bfa^\dagger )^\ell \big).
\]  
An easy induction shows that $\bfa^k(\bfa^\dagger)^k = \bfa\bfa^\dagger (\bfa\bfa^\dagger+1)\cdots (\bfa\bfa^\dagger +k-1) \leq k! (\bfa\bfa^\dagger)^k\leq k! (\bfa\bfa^\dagger)^s$ if $k\leq s$. Thus, 
\[
D_z(\bfa\bfa^\dagger)^{s}D_z^\dagger \leq \bigg( \sum_{k+\ell \leq s}  (k!+\ell!)   |h_{k, \ell}(z, \bar z)| \bigg) (\bfa\bfa^\dagger)^s.
\]
\end{proof}

The characteristic function exhibits key properties, particularly around the origin. We previously noted such a property in~\eqref{CharSupOrigin}. Here, we state a bound on the derivative of the characteristic function around the origin. Using~\eqref{eq:Dz-alt-form} to  compute the derivatives of $D_z$, we obtain
\begin{equation}\label{eq:displacement-derivative}
	\partial_{z_j} D_z = \frac{1}{2} \bar z_{j} D_z + D_z \ac_j^{\dagger} = -\frac{1}{2} \bar z_{j} D_z + \ac_j^{\dagger} D_z = \frac 12\big(\ac_j^{\dagger} D_z+ D_z\ac_j^{\dagger}\big),
\end{equation}
and
\begin{equation}
\partial_{\bar z_{j}} D_z  = \frac{1}{2} z_j D_z - \ac_j D_z = -\frac{1}{2} z_j D_z - D_z \ac_j=-\frac 12\big( D_z \ac_j + \ac_j D_z\big).
\end{equation}
Hence, it follows from the definition of the characteristic function that
\begin{align}
	&\partial_{z_j} \chi_{T}(z) = \frac{1}{2} \bar z_{j} \chi_{T}(z) + \chi_{\ac_j^{\dagger}T}(z) = -\frac{1}{2} \bar z_{j} \chi_{T}(z) + \chi_{T\ac_j^{\dagger}}(z) = \chi_{\frac 12(\ac_j^{\dagger}T +T\ac_j^\dagger)}(z) ,\label{eq:characterisitc-derivative-0}\\
	& \partial_{\bar z_{j}} \chi_{T}(z) = \frac{1}{2} z_j \chi_{T}(z) - \chi_{T \ac_j}(z)= -\frac{1}{2} z_j \chi_{T}(z) - \chi_{\ac_jT}(z)=-\chi_{\frac 12(\ac_jT +T\ac_j)}(z).\label{eq:characterisitc-derivative-1}
\end{align}
These equations indicate that the derivatives of $\chi_\rho(z)$ around the origin are related to the moments of $\rho$.

\begin{proposition} \emph{(\cite[Proposition 25]{CRLimitTheorem})} For any positive integer $\kappa$ and $\epsilon >0$, there exists a constant $c(\kappa, \epsilon)$ such that for any $m$-mode quantum state $\rho$, and any $\alpha, \beta \in \mathbb{R}_{+}^m$ with $|\alpha| + | \beta| \leq \kappa$, we have
\begin{equation}\label{eq:BoundDerivative-Moments}
\sup_{z \in B(0, \epsilon)} \big|\partial_{z}^{\alpha} \partial_{\bar z}^{\beta} \chi_\rho(z)\big| \leq c(\kappa, \epsilon) M_{\kappa}(\rho).
\end{equation}
Here, $|\alpha|=\alpha_1+\cdots+\alpha_m$.
\end{proposition}

It would be beneficial 
to understand the proof steps of the above inequality from~\cite{CRLimitTheorem}. First recall from~\eqref{eq:positive-Wigner} that the Wigner function of the quantum state $\sigma = \rho \boxplus \ketbra 00$ is non-negative and can be viewed as a probability density, whose absolute moments are defined by
\begin{align}\label{eq:moment-W}
	M^W_\kappa (\sigma) := \int \|z\|^\kappa W_{\sigma}(z) \, \dd^{2m} z.
\end{align} 
By~\eqref{eq:characteristic-convolution}, the characteristic function of $\rho$ is expressed in terms of the characteristic function of $\sigma$. Then, using classical results that provide bounds on the derivatives of the characteristic function in terms of absolute moments, it is shown in~\cite{CRLimitTheorem} that for any positive integer $\kappa$ and $\epsilon > 0$, there exists a constant $c''(\kappa, \epsilon)$ such that for any $m$-mode quantum state $\rho$, and any $\alpha, \beta \in \mathbb{R}_{+}^m$ with $|\alpha| + | \beta| \leq \kappa$, we have
\begin{equation}\label{eq:BoundDerivative-Moments-W}
	\sup_{z \in B(0, \epsilon)} \big|\partial_{z}^{\alpha} \partial_{\bar z}^{\beta} \chi_\rho(z)\big| \leq c''(\kappa, \epsilon) M^W_{\kappa}\big(\rho \boxplus \ketbra 00\big).
\end{equation}
Next, it is shown in~\cite{CRLimitTheorem} that there exists a constant $c'(\kappa) > 0$ such that
\begin{align}\label{eq:M-W-c-W}
	M_\kappa^W(\sigma) \leq c'(\kappa) M_{\kappa}(\sigma).
\end{align}
Finally,~\cite[Eq.~(A3)]{CRLimitTheorem} shows that
\begin{align}\label{eq:moment-conv-rho-00}
	M_\kappa\big(\rho \boxplus \ketbra 00\big) \leq M_\kappa(\rho).
\end{align} 
Combining these results yields the desired bound in~\eqref{eq:BoundDerivative-Moments}.

%**************************
\subsection{Taylor expansion of the characteristic function}

Let $\rho$ be a centered $m$-mode quantum state with finite $\kappa$-th-order moments where $\kappa \geq 4$. By equation~\eqref{eq:char-symmetric}, we have 
\begin{align*}
	\ln  \chi_{\rho^{\boxplus n}}(z) = n \ln  \chi_{\rho}\Big(\frac {z}{\sqrt n}\Big).
\end{align*}
This expression indicates that, to analyze $\rho^{\boxplus n}$, it is useful to consider the Taylor expansion of $\ln  \chi_{\rho}(z)$ around the origin, as $\big|\frac{z}{\sqrt n}\big|$ becomes sufficiently small for any fixed $z$ as $n \to \infty$. Therefore, we aim to compute derivatives of $\chi_{\rho}(z)$ at the origin. By~\eqref{eq:BoundDerivative-Moments} and our assumption on $\rho$, all such derivatives up to order $\kappa$ exist. Thus, we can write the Taylor series for $\ln  \chi_{\rho}(z)$ around the origin:
\begin{equation} \label{TaylorChar}
	\ln  \chi_{\rho}(z) = \sum_{|\alpha| \leq \kappa} q_\alpha \frac{z^\alpha}{\alpha !} + o(\| z\|^{\kappa}),
\end{equation}
where $\alpha = (\alpha_1, \dots, \alpha_{2m})$ is a tuple of non-negative integers, and
\begin{align}
	|\alpha| = \sum_{i=1}^{2m} \alpha_i, \qquad z^\alpha = \prod_{i=1}^{m} z_i^{\alpha_{2i-1}} \bar{z_i}^{\alpha_{2i}}, \qquad \alpha ! = \alpha_1 ! \cdots \alpha_{2m} !.
\end{align}
Recall that the Wigner function is defined as the complex Fourier transform of the characteristic function. Accordingly, we can interpret the coefficients $q_\alpha$ in~\eqref{TaylorChar} as the classical cumulants of the Wigner function of $\rho$, with the understanding that the Wigner function may not represent a true probability distribution.

Now, let us compute the coefficients $q_\alpha$. Recall that by definition, $\chi_{\rho}(0)=\tr(\rho) =1$, which implies $q_0 = 0$. Furthermore, according to~\eqref{eq:characterisitc-derivative-0} and~\eqref{eq:characterisitc-derivative-1}, and the fact that $\rho$ is centered, all the first-order derivatives of $\chi_{\rho}(z)$ at the origin are equal to $0$. This leads to the conclusion that $q_\alpha=0$ for any $\alpha$ satisfying $|\alpha| = 1$. On the other hand, the coefficients $q_\alpha$ with $|\alpha|=2$ are computed in terms of the covariance matrix of $\rho$ and match those of $\rho_G$, the Gaussification of $\rho$, similar to the cumulant coefficients in the classical setting. Moreover, the characteristic function of $\rho_G$ is a Gaussian function, and its logarithm consists only of quadratic terms. Consequently, we have 
\[
\sum_{|\alpha| = 2} q_\alpha \frac{z^\alpha}{\alpha !} = \ln  \chi_{\rho_G}(z).
\]
In particular, if $\rho$ is in the Williamson form and its covariance matrix is diagonal, given by $\boldsymbol{\gamma}( \rho) = \diag(\nu_1, \nu_1, \dots , \nu_m, \nu_m)$, then we have 
\[
\sum_{|\alpha| = 2} q_\alpha \frac{z^\alpha}{\alpha !} = \ln  \chi_{\rho_G}(z) = \sum_{j=1}^m -\frac 12 \nu_j|z_j|^2.
\]
Putting these together we can write
\begin{equation}\label{eq:log-expansion-0}
	\ln  \chi_{\rho^{\boxplus n}}(z) = \ln  \chi_{\rho_G}(z) + \sum_{3 \leq |\alpha| \leq \kappa} q_\alpha \frac{z^\alpha}{\alpha !} n^{-\frac{|\alpha|-2}{2}} + n\cdot o\Big(\Big\| \frac{z}{\sqrt{n}}\Big\|^{\kappa}\Big).
\end{equation}
Moreover, as $n\to \infty$, we have
\begin{equation}\label{eq:expansion-ET}
	\chi_{\rho^{\boxplus n}}(z) =  \chi_{\rho_G}(z)\cdot \bigg(1 + \sum_{r=1}^{\kappa-2} n^{-\frac r2} E_{\rho, r} (z)\bigg) \cdot \Big(1 + o\big(n^{-\frac{\kappa-2}{2}} \cdot \| z\|^\kappa\big)\Big),
\end{equation}
where $E_{\rho, r}(z)$ is a polynomial of degree at most $r+2$ that depends only on the coefficients $q_\alpha$. In particular, $E_{\rho, 1}(z)$ is a polynomial  all of whose terms have degree $3$. We refer to~\eqref{eq:expansion-ET} as the quantum Edgeworth-type expansion.

Using the expansion~\eqref{eq:expansion-ET}, we can readily prove that as $n$ goes to infinity, the characteristic function of the $n$-fold convolution of $\rho$ converges pointwise to the characteristic of $\rho_G$. In the classical setting, the Lindberg--Levy version of the CLT can be derived using this argument and Levy's continuity theorem~\cite{durrett2019}, which states that pointwise convergence can lead to convergence in the distribution.
Even in the quantum setting, this expansion can give us the Cushen--Hudson version of the quantum CLT, using ~\cite[Lemma 4]{lami_all_2018}.

However, in order to prove stronger versions of the quantum CLT, we need to carefully analyze the quantum Edgeworth-type expansion~\eqref{eq:expansion-ET}. In the rest of this section, we will present two important lemmas that will aid in proving the optimal convergence rates for the quantum CLT under minimal assumptions.

\begin{lemma} \label{BRMethod}
Let $\rho$ be an $m$-mode quantum state that is centered and has finite $\kappa$-th-order moments for some integer $\kappa \geq 4$, and let $\nu_{\min} = \|\boldsymbol{\gamma}(\rho)^{-1}\|^{-1}$ and $\nu_{\max} = \|\boldsymbol{\gamma}(\rho)\|$ be the minimum and maximum eigenvalues of the covariance matrix of $\rho$, respectively. Then, there exist a constant $C=C(\nu_{\min}, \nu_{\max})>0$ and a polynomial $L(\cdot) = L_{\nu_{\min}, \nu_{\max}}(\cdot)$ whose degree depends only on $\kappa, \alpha, \beta$, and whose coefficients depend only on $\nu_{\min}, \nu_{\max}$, such that the following holds:
for any $z \in \mathbb{C}^{m}$ with $\|z\|\leq C M_{\kappa}(\rho)^{-1/(\kappa-2)}\sqrt n$, and any $\alpha, \beta\in \mathbb Z_+^m$ with $|\alpha| + |\beta| \leq \kappa $, we have
\begin{align}\label{eq:BRMethod}
	\bigg|\partial_{z}^\alpha  \partial_{\bar z}^\beta \bigg(\chi_{\rho^{\boxplus n}}(z) - \chi_{\rho_{G}}(z) \Big(1 + \sum_{r=1}^{\kappa-2} n^{-\frac r2} E_{\rho, r} (z)\Big) \bigg)\bigg| \leq  n^{-\frac{\kappa-2}{2}} M_{\kappa} L(\| z\|)  e^{ -\frac{\nu_{\min}-1}{4} \|z \|^2},
\end{align}
where $M_{\kappa}$ can be any of $M^W_{\kappa}\big(\rho\boxplus\ketbra{0}{0}\big)$, $M_{\kappa}\big(\rho\boxplus\ketbra{0}{0}\big)$ or $M_{\kappa}\big(\rho\big)$.
\end{lemma}

We remark that $M^W_{\kappa}\big(\rho\boxplus\ketbra{0}{0}\big)$ is defined in~\eqref{eq:moment-W}, and it can be seen to be finite by~\eqref{eq:M-W-c-W} and~\eqref{eq:moment-conv-rho-00}.

\begin{proof}
First suppose that $\rho$ has a non-negative Wigner function. Then, by~\cite[Theorem 9.10]{BhRa} for some constant $C_0>0$ and polynomial $L_0(\cdot)$ as in the statement of the lemma we have
\begin{align}\label{eq:BRMethod-positive-W}
	\bigg|\partial_{z}^\alpha  \partial_{\bar z}^\beta \bigg(\chi_{\rho^{\boxplus n}}(z) - \chi_{\rho_{G}}(z) \Big(1 + \sum_{r=1}^{\kappa-2} n^{-\frac r2} E_{\rho, r} (z)\Big) \bigg)\bigg| \leq  n^{-\frac{\kappa-2}{2}} M^W_{\kappa}(\rho) L_0(\| z\|) \chi_{\rho_G}(z)^{\frac 12}, 
\end{align}
for all $\|z\|\leq C_0 M_{\kappa}(\rho)^{-1/(\kappa-2)}\sqrt n$. Here, we use the fact that the covariance matrix of any density operator is positive definite and as discussed in the previous subsection, the moments  of $\rho$ bound the moments of its Wigner function as a probability density. 

Next, for an arbitrary $\rho$ consider its symmetric convolution with the vacuum state $\sigma = \rho\boxplus \ketbra{0}{0}$. Observing that the Wigner function of $\sigma$ is non-negative, we can apply the above inequality for $\sigma$. Using~\eqref{eq:characteristic-convolution}, we have
$$
\chi_{\sigma^{\boxplus n}} (z) = \chi_{\ketbra{0}{0}}(z/\sqrt 2)\chi_{\rho^{\boxplus n}} (z/\sqrt 2)  = e^{-\frac 14\|z\|^2}  \chi_{\rho^{\boxplus n}} (z/\sqrt 2).
$$
Then, the derivative  $\partial_{z}^\alpha  \partial_{\bar z}^\beta\chi_{\rho^{\boxplus n}} (z)$ can be expressed in terms of lower-order derivatives of $\chi_{\sigma^{\boxplus n}} (\sqrt 2 z)$ and those of $e^{\frac 12\|z\|^2} $. 
We similarly have $\chi_{\sigma_G}(z) = e^{-\frac 14\|z\|^2}  \chi_{\rho_G} (z/\sqrt 2)$. Comparing with~\eqref{eq:expansion-ET}, we realize that $E_{\sigma, r}(z) = E_{\rho, r}(z/\sqrt 2)$. Thus, the derivative 
$\partial_{z}^\alpha  \partial_{\bar z}^\beta \chi_{\rho_{G}}(z) \Big(1 + \sum_{r=1}^{\kappa-2} n^{-\frac r2} E_{\rho, r} (z)\Big)$ can also be expressed in terms of lower-order derivatives of $\chi_{\sigma_G} (\sqrt 2 z)\big(1 + \sum_{r=1}^{\kappa-2} n^{-\frac r2} E_{\sigma, r} (\sqrt 2 z)\big) $ and those of $e^{\frac 12\|z\|^2}$. Therefore, starting with the left hand side of~\eqref{eq:BRMethod}, writing down the derivative as explained, applying the triangle inequality and using~\eqref{eq:BRMethod-positive-W} for $\sigma$ to bound each term in the triangle inequality, we arrive at
\begin{align*}
	  \bigg|\partial_{z}^\alpha  \partial_{\bar z}^\beta \bigg(\chi_{\rho^{\boxplus n}}(z) - \chi_{\rho_{G}}(z) \Big(1 + \sum_{r=1}^{\kappa-2} n^{-\frac r2} E_{\rho, r} (z)\Big) \bigg)\bigg| & \leq  n^{-\frac{\kappa-2}{2}} M^W_{\kappa}\big(\rho\boxplus \ketbra00\big) L(\| z\|) \chi_{\sigma_G}(\sqrt 2 z)^{\frac 12} e^{\frac 12\|z\|^2}\\
& = n^{-\frac{\kappa-2}{2}} M_{\kappa}^W\big(\rho\boxplus \ketbra00\big) L(\| z\|) \chi_{\rho_G}( z)^{\frac 12} e^{\frac 14\|z\|^2},
\end{align*}
which holds for any $z\in \mathbb C^m$ satisfying $\|z\|\leq C\epsilon \sqrt n$, where $C=C_0/\sqrt 2$. Here, $L(\cdot)$ is a polynomial which can be computed in terms of $L_0(\cdot)$ and the polynomials appearing in the derivatives of $e^{\frac 12 \|z\|^2}$. The proof for $M_\kappa = M_\kappa^W\big(\rho\boxplus \ketbra00\big)$ concludes noting that  $\chi_{\rho_G}(z) \leq e^{-\frac 1 2 \nu_{\min}\|z\|^2}$.  
Also, the fact that $M^W_{\kappa}\big(\rho\boxplus\ketbra{0}{0}\big)$ can be replaced with $M_{\kappa}\big(\rho\boxplus\ketbra{0}{0}\big)$ or $M_{\kappa}\big(\rho\big)$ follows from~\eqref{eq:M-W-c-W} and~\eqref{eq:moment-conv-rho-00}. 

\end{proof}

We will also need the following bound on the derivatives of $E_{\rho, r} (z)$.

\begin{lemma}\label{EdgeWorth-BindDerivatives}
For any $\alpha, \beta\in \mathbb Z_+^m$ with $|\alpha| + |\beta| \leq \kappa $, and any $r\in \mathbb{N}$, there exists a polynomial $Q(\cdot) = Q_{\nu_{\max}}(\cdot)$ depending only on $\nu_{\max} = \|\boldsymbol{\gamma}(\rho)\|$, such that
\[
	\Big|\partial_{z}^\alpha  \partial_{\bar z}^\beta \big(  E_{\rho, r} (z) \big)\Big| \leq  Q(\|z\|) M_{r+2}, \qquad \quad \forall z\in \mathbb C^{m},
\]
where $M_{r+2}$ can be any of $M^W_{r+2}\big(\rho\boxplus \ketbra{0}{0}  \big)$, $M_{r+2}\big(\rho\boxplus \ketbra{0}{0}  \big)$ or $M_{r+2}\big(\rho  \big)$.
\end{lemma} 

\begin{proof}
The proof is similar to that of Lemma~\ref{BRMethod}. First, since $\sigma=\rho\boxplus \ketbra{0}{0}$ has a non-negative Wigner function by~\cite[Lemma 9.5]{BhRa} we have
$$
\Big|\partial_{z}^\alpha  \partial_{\bar z}^\beta \big(  E_{\sigma, r} (z) \big)\Big| \leq  Q_0(\|z\|) M^W_{r+2}\big(\rho\boxplus \ketbra{0}{0}  \big),
$$
for some polynomial $Q_0(\cdot)$. Then, as argued in the proof of Lemma~\ref{BRMethod}, we have $E_{\sigma, r}(z) = E_{\rho, r}(z/\sqrt 2)$, from which the lemma for arbitrary $\rho$ follows.
\end{proof}

%****************************************************************

\section{State truncation via Fock basis}\label{Sec:State-Truncation}

To prove Theorem~\ref{MainTheoremTrace} and Theorem~\ref{MainTheorem} under the assumption of finite moments of constant order independent of \( m \), the number of modes, we need to approximate the underlying quantum state \( \rho \) with states that possess finite higher-order moments. To achieve this, we introduce a method we refer to as \emph{state truncation via the Fock basis}, which constructs a sequence of finite-rank states \( \sigma_n \) that approximate \( \rho \), ensuring that all the moments of \( \sigma_n \) are finite. However, since our truncated states \( \sigma_n \) depend on \( n \), their moments will also vary with \( n \). Therefore, our analysis requires careful consideration of the moments of \( \sigma_n \) as well as its characteristic function.

For any integer \( n \in \mathbb{N} \), we define the projector \( \Pi_n \) as 
\begin{equation}\label{eq:Truncation_Projection}
	\Pi_n := \sum_{|k| \leq n} \ketbra{k}{k},
\end{equation}
where \( \ket{k} = \ket{k_1, \cdots, k_m} \) is a state in the Fock basis and \( |k| = \sum_{i=1}^m k_i \). We then define the truncation of \( \rho \) by
\begin{equation}\label{eq:def_sigmaN}
	\sigma_n = \frac{\rho_n}{\tr(\rho_n)}, \quad \text{where} \quad \rho_n = \Pi_n \rho \Pi_n.
\end{equation}
Note that all moments of \( \sigma_n \) are finite, yet they depend on \( n \). The following lemma helps to ensure that the trace distance between the \( n \)-fold convolution of \( \rho \) and the \( n \)-fold convolution of \( \sigma_n \) decays at a sufficient rate, allowing us to use them to prove Theorem~\ref{MainTheoremTrace} and Theorem~\ref{MainTheorem}.

\begin{lemma}\label{Aux1}
	Let $\rho$ be an $m$-mode quantum state with finite $s$-th order moments for some $s \geq 3$. With the above notation, we have
	$$
	\big\| \rho^{\boxplus n} - \sigma_n^{\boxplus n}\big\|_1=\mathcal{O}\Big(\frac{1}{n^{(s-2)/2}}\Big) \quad \text{as} \quad n\rightarrow\infty.
	$$
\end{lemma}
\begin{proof}
	Note first that
	\begin{align}
		1 - \tr(\rho_n) &= 1 - \sum_{|k| \leq n} \bra{k} \rho \ket{k} \nonumber \\
		&= \sum_{|k| > n} \bra{k} \rho \ket{k} = \sum_{|k| > n} |k|^{s/2} |k|^{-s/2} \bra{k} \rho \ket{k} \nonumber \\
		&\leq n^{-s/2} \sum_{|k| > n} |k|^{s/2} \bra{k} \rho \ket{k} \nonumber\\
		&= \mathcal{O}\Big(\frac{1}{n^{s/2}}\Big), \label{eq:bound-tr-rho-n}
	\end{align}
	where in the last line we use the fact that the $s$-th order moments of $\rho$ are finite. This in turn yields
	\begin{align*}
		\big\| \rho - \sigma_n\big\|_1 &= \tr\Big(\Big|\rho -  \frac{\rho_n}{\tr(\rho_n)}\Big|\Big) \\
		&= \sum_{|k| \leq n} \Big(\frac{1}{\tr(\rho_n)} - 1\Big) \bra{k} \rho \ket{k} + \sum_{|k| > n} \bra{k} \rho \ket{k} \\ 
		&= 2 \big(1- \tr(\rho_n)\big) = \mathcal{O}\Big(\frac{1}{n^{s/2}}\Big).
	\end{align*}
	Finally, we use this estimate to obtain
	\begin{align}
		\big\| \rho^{\boxplus n} - \sigma_n^{\boxplus n}\big\|_1 &= \bigg\| \sum_{i=0}^{n-1} \left( \rho^{\boxplus i} \boxplus_{\frac in} \sigma_n^{\boxplus n-i} - \rho^{\boxplus i+1} \boxplus_{\frac {i+1}{n}} \sigma_n^{\boxplus n-i-1}\right)\bigg\|_1 \nonumber \\
		&\leq \sum_{i=0}^{n-1} \Big\| \rho^{\boxplus i} \boxplus_{\frac in} \sigma_n^{\boxplus n-i} - \rho^{\boxplus i+1} \boxplus_{\frac {i+1}{n}} \sigma_n^{\boxplus n-i-1}\Big\|_1 \nonumber\\
		&= \sum_{i=0}^{n-1} \Big\| \rho^{\boxplus i} \boxplus_{\frac in} (\rho - \sigma_n) \boxplus_{\frac {n-i-1}{n}} \sigma_n^{\boxplus n-i-1}\Big\|_1 \nonumber\\
		&\leq n \| \rho - \sigma_n\|_1 = \mathcal{O}\Big(\frac{1}{n^{(s-2)/2}}\Big). \nonumber
	\end{align}
Note that in the first line we invoked the telescoping identity
\begin{equation*}
    X_0 - X_n = \sum_{i = 0}^{n - 1} (X_i - X_{i + 1})
\end{equation*}
with $X_0 = \sigma_n^{\boxplus n}$, $X_n = \rho^{\boxplus n}$, and
\begin{equation*}
    X_i = \rho^{\boxplus i} \boxplus_{\frac in} \sigma_n^{\boxplus n-i}, \qquad i = 1,\cdots,n-1.
\end{equation*}
In the second line we used the triangle inequality for the trace distance. The third line is a consequence of the associativity of the beam-splitter cascade and the linearity of $\boxplus_\eta$ in each of its arguments. The penultimate line follows from the data-processing inequality for the trace distance.
\end{proof}

It is important to note that the states \( \sigma_n \) are not centered in general. In our arguments, we apply displacement operators to center them, so as we will see, it is beneficial to have upper bounds on the first moments of \( \sigma_n \). The following lemma provides us with such an upper bound.
\begin{lemma}\label{bound-FirstNorm}
	Let $\rho$ be a centered $m$-mode quantum state with finite $s$-th order moments for some $s \geq 3$. Then, for any $j$ we have $|z_{n, j}|=\mathcal{O}\big(\frac{1}{n^{(s-1)/2}}\big)$, where $z_{n,j}=\tr(\sigma_n \ac_j)$ and $\sigma_n$ is defined in~\eqref{eq:def_sigmaN}.
\end{lemma}
\begin{proof}
	For sufficiently large $n$, we have
	\begin{align}
		|\tr(\sigma_n \ac_j)| 
		&\leq 2 |\tr(\rho_n \ac_j)| \nonumber\\
		&=  2 |\tr((\rho - \rho_n) \ac_j)| \nonumber\\
		&= \sum_{|k| >n} 2 \sqrt{k_j} \, |\bra{k} \rho \ket{k^{-j}}| \qquad\qquad \qquad \qquad \text{($k^{-j} = (k_1, \cdots, k_j-1, \cdots, k_m)$)} \nonumber\\
		&\leq \sum_{|k| >n}  \sqrt{k_j} \, \big( \bra{k} \rho \ket{k} + \bra{k^{-j}} \rho \ket{k^{-j}}\big) \nonumber\\
		&\leq 2\sum_{|k| \geq n}  \sqrt{|k|+1} |\bra{k} \rho \ket{k}| \nonumber \\
		&\leq \frac{2}{n^{(s-1)/2}} \sum_{|k| \geq n}  (|k|+1)^{s/2} |\bra{k} \rho \ket{k}| \nonumber = \mathcal{O}\Big(\frac{1}{n^{(s-1)/2}}\Big).
	\end{align}
	In the first line, we utilize the fact that, according to~\eqref{eq:bound-tr-rho-n}, for large \( n \), we have \( \tr(\rho_n) \geq \frac{1}{2} \). The second line leverages the fact that \( \rho \) is centered. The fourth line follows from the Cauchy–Schwarz and AM-GM inequalities, while, in the fifth line, we use the fact that \( |k| \geq k_j \). Finally, the last line follows from the finiteness of the \( s \)-th order moments of \( \rho \).
\end{proof}
Now, let \( z_n = (z_{n,1}, \dots, z_{n,m}) \) consist of the first moments of \( \sigma_n \), defined as 
$$ z_{n,j} := \tr(\sigma_n \bfa_j). $$
Additionally, let \( \tilde{\sigma}_n \) be the displacement of \( \sigma_n \) by \( D_{z_n} \):
$$ \tilde{\sigma}_n := D_{z_n}^\dagger \sigma_n D_{z_n}. $$
Moreover, denote the Gaussification of \( \tilde{\sigma}_n \) by \( \tilde{\sigma}_{n,G} \). Note that by definition, \( \tilde{\sigma}_{n,G} \) is centered and has the same covariance matrix as that of \( \sigma_n \). In the following lemma, we establish an asymptotic bound on the distance between the covariance matrices of \( \tilde{\sigma}_{n,G} \) and \( \rho_G \) in terms of the operator norm.

\begin{lemma}\label{OperatorNormCovarianceMatrix}
	Let $\rho$ be a centered $m$-mode quantum state with finite $s$-th order moments for some $s \geq 3$. With the above notation, we have
	\begin{equation} \label{SecondMomentTruncation}
		\big\| \boldsymbol{\gamma}(\rho) - \boldsymbol{\gamma}(\tilde \sigma_{n})\big\| =\big\| \boldsymbol{\gamma}(\rho_G) - \boldsymbol{\gamma}(\tilde \sigma_{n,G})\big\| = \mathcal{O}\Big(\frac{1}{n^{(s-2)/2}}\Big) \quad \text{as} \quad n\rightarrow\infty,
	\end{equation}
	where $\|\cdot\|$ denotes the operator norm.
\end{lemma}

\begin{proof}
	By the definition of the covariance matrix and using Lemma~\ref{bound-FirstNorm}, it suffices to show that \( \big| \tr\big((\rho - \sigma_n) AB\big) \big| = \mathcal{O}(1/n^{(s-2)/2}) \) for any \( A, B \in \{\bfa_1, \bfa_1^\dagger, \dots, \bfa_m, \bfa_m^\dagger\} \). We will only prove this for \( A = \bfa_j \) and \( B = \bfa_{j'} \) with \( j \neq j' \), as the proof for the other cases is similar. We compute
	\begin{align*}
		\big| \tr\big(\rho \bfa_j \bfa_{j'}\big) - \tr\big(\sigma_n \bfa_j \bfa_{j'}\big) \big| 
		& \leq \Big| \sum_k  \sqrt{k_j k_{j'}} \bra{k} \rho \ket{k^{-j,-{j'}}} - \frac{1}{\tr(\rho_n)}  \sum_{|k| \leq n}  \sqrt{k_j k_{j'}} \bra{k} \rho \ket{k^{-j,-{j'}}} \Big|  \\
		& \leq \frac{1 - \tr(\rho_n)}{\tr(\rho_n)} \sum_{|k| \leq n}  \sqrt{k_j k_{j'}} \big| \bra{k} \rho \ket{k^{-j,-{j'}}} \big| + \sum_{|k| > n}  \sqrt{k_j k_{j'}} \big| \bra{k} \rho \ket{k^{-j,-{j'}}} \big|. 
	\end{align*}
	For the first term, by~\eqref{eq:bound-tr-rho-n}, we have \( \frac{1 - \tr(\rho_n)}{\tr(\rho_n)} = \mathcal{O}(n^{-s/2}) \). Applying the Cauchy–Schwarz inequality yields \( \sum_{|k| \leq n}  \sqrt{k_j k_{j'}} \big| \bra{k} \rho \ket{k^{-j,-{j'}}} \big| < +\infty \) since the second-order moments of \( \rho \) are finite.  For the second term, by applying a similar chain of inequalities as in the proof of Lemma~\ref{bound-FirstNorm} and using the finiteness of the \( s \)-th order moments, we obtain
	\begin{align*}
		\sum_{|k| > n}  \sqrt{k_j k_{j'}} \big| \bra{k} \rho \ket{k^{-j,-{j'}}} \big| \leq  \frac{2}{(n-1)^{(s-2)/2}} \sum_{|k| \geq n-1}  (|k| + 2)^{s/2} \big| \bra{k} \rho \ket{k} \big| = \mathcal{O}\Big(\frac{1}{n^{(s-2)/2}}\Big).
	\end{align*}
	Combining these, the desired bound is established.
	
\end{proof}

Other moments of $\tilde \sigma_n$ can also be bounded in terms of the moments of $\rho$. 

\begin{lemma}\label{lem:BindMoments-Truncation}
	Let $\rho$ be a centered $m$-mode quantum state with finite $s$-th order moments for some $s \geq 3$. For any $\kappa\geq s$, we have 
	$M_\kappa(\tilde \sigma_n)= \mathcal{O}\big( M_s(\rho)n^{(\kappa-s)/2} \big)$. 
\end{lemma}

\begin{proof}
	Using Lemma~\ref{lem:moment-displaced} for any $\kappa \geq 3$, we have 
	\begin{align*}
		M_{\kappa}\big(\tilde{\sigma}_{n}\big)  = M_{\kappa}\big(  D^\dagger_{z_n} \sigma_{n} D_{z_n}\big)    \leq g_{m, \kappa}(\|z_n\|)  M_{\kappa}( \sigma_{n} ),
	\end{align*}
	for some polynomial $g_{m, \kappa}(\cdot)$.
	Additionally, we have
	\begin{align*}	
		M_\kappa(\sigma_n) &= \sum_{k} (|k|+m)^{\kappa/2} \bra k \sigma_n \ket k \\
		&= \sum_{|k|\leq n} (|k|+m)^{\kappa/2} \bra k \sigma_n \ket k\\
		&= \frac{1}{\tr(\rho_n)}\sum_{|k|\leq n} (|k|+m)^{\kappa/2} \bra k \rho_n \ket k\\
		&\leq 2  (n+m)^{(\kappa-s)/2}  \sum_{|k| \leq n} (|k|+m)^{s/2} \bra{k} \rho \ket{k}  \\
		&\leq  2  (n+m)^{(\kappa-s)/2} M_s(\rho),
	\end{align*}
	where in the fourth line we use~\eqref{eq:bound-tr-rho-n} to handle $\tr(\rho_n)$. Putting these together, we arrive at the desired bound.
\end{proof}

We conclude this section with a bound on the characteristic function of the centered truncated state.

\begin{lemma}\label{lem:bound-chi-tilde-sigma}
	Let $\rho$ be a centered $m$-mode state with finite $s$-th order moments such that $s \geq 3$. For any $\epsilon>0$, there are constants $\theta_0, \theta_1>0$ such that for any sufficiently large $n$ and any $\|z\|\leq \epsilon\sqrt n$ we have
	\begin{align*}
		\big|\chi_{\tilde \sigma_n} (z/\sqrt n)\big|^{n-2m} \leq  e^{2\theta_0 m+\epsilon \theta_1 \|z\|^2} \chi_{\rho_G}(z).
	\end{align*}
\end{lemma}

\begin{proof}
	Considering the Taylor expansion of $|\chi_{\tilde \sigma_n} (z)|^2$ around the origin and
	using the fact that $\tilde \sigma_n$ is centered, there is a constant $\theta'_1>0$ such that for $\|z\|\leq \epsilon$ we have
	\begin{align*}
		\big|\chi_{\tilde \sigma_n} (z)\big|^2 &\leq 1- \frac 12 \hat{z}^\dagger \Lambda_m^\dagger \boldsymbol{\gamma}(\tilde \sigma_n) \Lambda_m  \hat{z} +  2\theta'_1 M_3(\tilde \sigma_n) \|z\|^3\\
		&\leq \exp\Big(- \frac 12 \hat{z}^\dagger \Lambda_m^\dagger \boldsymbol{\gamma}(\tilde \sigma_n) \Lambda_m  \hat{z} +  2\theta'_1 M_3(\tilde \sigma_n) \|z\|^3\Big),
	\end{align*}
	where $\hat{z} = (z_1, \bar{z}_1, \cdots, z_m, \bar{z}_m)^\top$. This implies
	\begin{align}\label{eq:chi-tilde-bound-000}
		\big|\chi_{\tilde \sigma_n} \big(z/\sqrt n\big)\big|
		\leq \exp\Big(- \frac {1}{4n} \hat{z}^\dagger \Lambda_m^\dagger \boldsymbol{\gamma}(\tilde \sigma_n) \Lambda_m  \hat{z} +  \theta'_1 M_3(\tilde \sigma_n) \frac{1}{n^{3/2}} \|z\|^3\Big), \qquad \quad \|z\|\leq \epsilon\sqrt n.
	\end{align}
	We note that there is a constant $b>0$ such that 
	$$\hat{z}^\dagger \Lambda_m^\dagger \boldsymbol{\gamma}(\tilde \sigma_n) \Lambda_m  \hat{z} \leq b M_2(\tilde \sigma_m) \|z\|^2\leq b M_3(\tilde \sigma_m)^{2/3} \|z\|^2,$$ 
	where the second inequality follows from Lemma~\ref{lem:moments-monotone}. Therefore, 
	\begin{align*}
		&\frac {1}{4n} \hat{z}^\dagger \Lambda_m^\dagger \boldsymbol{\gamma}(\tilde \sigma_n) \Lambda_m  \hat{z} -  \theta'_1 M_3(\tilde \sigma_n) \frac{1}{n^{3/2}} \|z\|^3\\
		&\qquad\quad  \leq  \frac {1}{4n} b M_3(\tilde \sigma_m)^{2/3} \|z\|^2-  \theta'_1 M_3(\tilde \sigma_n) \frac{1}{n^{3/2}} \|z\|^3\\
		&\qquad \quad \leq  \max_{x\geq 0} \frac {1}{4} b x^2-  \theta'_1 x^3\\
		&\qquad \quad  =: \theta_0.
	\end{align*}
	Using the definition of $\theta_0$ and~\eqref{eq:chi-tilde-bound-000}, for any $\|z\|\leq \epsilon\sqrt n$ we obtain
	\begin{align*}
		\big|\chi_{\tilde \sigma_n} \big(z/\sqrt n\big)\big|^{n-2m}
		& \leq \exp\Big( 2m\theta_0  - \frac {1}{4} \hat{z}^\dagger \Lambda_m^\dagger \boldsymbol{\gamma}(\tilde \sigma_n) \Lambda_m  \hat{z} +  \theta'_1 M_3(\tilde \sigma_n) \frac{1}{n^{1/2}} \|z\|^3\Big)\\
		& \leq \exp\Big( 2m\theta_0  - \frac {1}{4} \hat{z}^\dagger \Lambda_m^\dagger \boldsymbol{\gamma}(\tilde \sigma_n) \Lambda_m  \hat{z} +  \theta'_1 M_3(\tilde \sigma_n)\epsilon \|z\|^2\Big).
	\end{align*}  
	Finally, by Lemma~\ref{lem:BindMoments-Truncation} for $s=3$, $M_3(\tilde \sigma_n)$ is bounded by a constant independent of $n$, and by Lemma~\ref{OperatorNormCovarianceMatrix} for sufficiently large $n$, we have $\|\boldsymbol{\gamma}(\tilde \sigma_n) - \boldsymbol{\gamma}(\rho)   \|\leq \epsilon$. Therefore, modifying the constant $\theta'_1$, there is some $\theta_1>0$ such that 
	\begin{align*}
		\big|\chi_{\tilde \sigma_n} \big(z/\sqrt n\big)\big|^{n-2m}
		& \leq \exp\Big( 2m\theta_0  - \frac {1}{4} \hat{z}^\dagger \Lambda_m^\dagger \boldsymbol{\gamma}(\rho) \Lambda_m  \hat{z} +  \theta_1 \epsilon \|z\|^2\Big)\\
		& = \exp\Big( 2m\theta_0  +  \theta_1 \epsilon \|z\|^2\Big)\chi_{\rho_G}(z),
	\end{align*} 
	as desired. 
\end{proof}

%%%%%%%%%%%%%%%%%%%%%%%%%%%%%%%%%%%%%%%%%%%%%%%%%%%%%%%%%%%%%%%%%%%%%%%%%%%%%%%%%%%%%%%%%%%%%%%%%%%%%%%%%%%%%%%%%%%%%%

\section{An upper bound on the quantum relative entropy}\label{upperbound}

This section is devoted to establishing a general upper bound on the relative entropy between two quantum states, with one being Gaussian, in terms of the Hilbert--Schmidt norm. This upper bound will be instrumental in the proof of Theorem~\ref{MainTheorem}, as it simplifies the task of estimating the relative entropy \( D\left(\rho^{\boxplus n} \big\| \rho_{G}\right) \) to that of a 2-norm. Utilizing the quantum Plancherel identity, this can be reformulated in terms of the characteristic functions and their derivatives.

\begin{theorem}\label{boundEntropy-Char}
	Let $\rho$ be an $m$-mode quantum state with finite moments of order $m+3$. Let $\tau = \prod_{j=1}^m \big(1-e^{-\beta_j}\big) e^{-\beta_j \ac_j^\dagger \ac_j}$ be an $m$-mode thermal state, where \(0<\beta_j<+\infty\). Let \(E(z)=E(z, \bar z)\) be an odd polynomial in \(z\) and \(\bar{z}\) such that \(E(-z) = -E(z) = \overline{E(z)}\). For \(\alpha \in \mathbb{R}\), we define the operator \(\tau_\alpha\) by \(\chi_{{\tau}_\alpha}(z) := \chi_{\tau}(z) \big( 1 + \alpha E(z)\big)\). Then, we have the following bound on the relative entropy:
	\[
	D(\rho \| \tau) \leq C_{\tau, E} \left( \alpha^2 + \left\|(\rho - {\tau_\alpha})(N_m + m)^{(m+3)/2}\right\|_2 \right),
	\]
	where \(C_{\tau, E} > 0\) is a constant depending on the polynomial \(E(z)\) and the thermal state \(\tau\), specifically on the parameters \(\beta = (\beta_1, \dots, \beta_m)\). The same result holds if \(\tau\) is replaced by any centered Gaussian state for which \(\boldsymbol{\gamma}(\tau) + i \Omega_m\) is full-rank.
\end{theorem}

As an immediate consequence of this result, we derive an upper bound for the \emph{relative entropy of non-Gaussianity} of a state $\rho$, symbolized as $d_\mathcal{G}(\rho)$~\cite{genoni2008quantifying,marian2013relative}. This quantity, defined as the relative entropy between $\rho$ and the collection $\mathcal{G}$ of all Gaussian states, serves as a metric for non-Gaussian features of quantum states:
\[
d_\mathcal{G}(\rho):=\min_{\sigma\in\mathcal{G}} D(\rho\|\sigma).
\]
The measure $d_\mathcal{G}(\rho)$ has several advantageous properties, making it a valuable indicator of non-Gaussianity. For instance, it is faithful: $d_\mathcal{G}(\rho)\geq0$, reaching zero exclusively when $\rho$ is a Gaussian state. In fact, the minimum value in this definition is achieved by the Gaussian reference state $\rho_{G}$, so that
\[
d_\mathcal{G}(\rho)=D(\rho\|\rho_{G}).
\]
Recently, this non-Gaussianity measure has been used in analyzing the robustness of specific quantum state tomography algorithms designed for continuous-variable systems \cite{bittel2024optimal,mele2024learning}.

Setting $\tau = \rho_{G}$ and $\alpha = 0$ in Theorem \ref{boundEntropy-Char}, we derive the following
\begin{corollary}\label{relent-nonG}
		Let $\rho$ be an $m$-mode quantum state with finite moments of order $m+3$, and let $\rho_{G}$ be its Gaussification. Then, for some constant $C_{\rho_G} > 0$, we have
		\begin{equation}\label{non-Gaussianity}
			d_\mathcal{G}(\rho) \leq C_{\rho_G} \left\|(\rho - {\rho_{G}})(N_m + m)^{(m+3)/2}\right\|_2.
		\end{equation}
\end{corollary}

To establish Theorem~\ref{boundEntropy-Char}, we start by proving two auxiliary lemmas. Initially, we assume that \(\tau\) is a thermal state, as specified in the theorem statement. Subsequently, we will address the proof in the general case.
\begin{lemma}\label{PreliminaryLemma}
	With the notation of Theorem~\ref{boundEntropy-Char}, for every Fock state $\ket{k}$, we have
	\begin{equation}\label{aux}
		\bra{k}\rho(\ln \rho-\ln \tau)\ket{k}\leq \bra{k}\big(\rho-\tau+(\rho-\tau)^2 \tau^{-1}\big)\ket{k}.
	\end{equation}
\end{lemma}
\begin{proof}
	Expanding the thermal state $\tau$ in Theorem~\ref{boundEntropy-Char} in the Fock basis, we have 
	$$ \tau = \sum_{k} q_k \ketbra{k}{k},$$ 
	where $q_k = \prod_{j=1}^m (1 - e^{-\beta_j}) e^{-\beta_j k_j}$. Let $\rho=\sum_{i=0}^{\infty} p_i\ketbra{v_i}{v_i}$ be the spectral decomposition of $\rho$. We then have
	\begin{align*}
		\bra{k}\rho\ln \rho\ket{k} =\sum_i |\braket{v_i}{k} |^2 p_i\ln  p_i,\qquad \quad		\bra{k}\rho\ln \tau\ket{k}=\sum_i |\braket{v_i}{k}|^2 p_i\ln  q_k.
	\end{align*}
	Thus, applying the inequality $t\ln  t\leq (t-1)+(t-1)^2$, $t\geq0$, with $t=\frac{p_i}{q_k}$ yields
	\begin{align}\label{ineqq}
		\bra{k}\rho(\ln \rho-\ln \tau)\ket{k} \notag&=\sum_i |\braket{v_i}{k}|^2 p_i\ln \frac{p_i}{q_k}\nonumber\\
		\notag&\leq \sum_i |\braket{v_i}{k}|^2 \Big(p_i-q_k+(p_i-q_k)^2q_k^{-1}\Big)\\
		&=\bra{k}(\rho-\tau)\ket{k}+\sum_i |\braket{v_i}{k}|^2 (p_i-q_k)^2q_k^{-1}.
	\end{align}
Next, we compute
	\begin{align*}
		q_k\bra{k}(\rho-\tau)^2 \tau^{-1}\ket{k} \notag&= \bra k(\rho-\tau)^2\ket k\\
		&=\bra k\big(\rho^2+\tau^2-\rho\tau-\tau\rho\big)\ket k\\
		&=\sum_i |\braket{v_i}{k}|^2 p_i^2-2q_k\sum_i |\braket{v_i}{k}|^2 p_i+q_k^2\\
		&=\sum_i |\braket{v_i}{k}|^2 (p_i-q_k)^2.
	\end{align*}	
	Combining this equation with~\eqref{ineqq}, we obtain the desired inequality~\eqref{aux}.
\end{proof}

The following lemma can be viewed as a quantum analogue of~\cite[Lemma 2.2]{BCG}.

\begin{lemma}\label{TruncationLemma}
	With the notation of Theorem~\ref{boundEntropy-Char}, for every $t\geq 0$ we have
	\begin{align}\label{preliminary-ineq}
		D(\rho\|\tau)\leq \sum_{\beta\cdot k\leq t}\bra{k}(\rho-\tau)^2 \tau^{-1}\ket{k} +\eta_\beta \sum_{\beta\cdot k> t} \big|\bra{k}(\rho-\tau)(N_m+m)\ket{k}\big|-\sum_{\beta\cdot k> t}\bra{k}\tau\ln \tau\ket{k},
	\end{align}
	where $\beta\cdot k = \beta_1k_1+\cdots+\beta_m k_m$ and $\eta_\beta=\beta_{\max}-\sum_{j=1}^m \ln  (1-e^{-\beta_j}) + 1$ with $\beta_{\max} = \max_j \beta_j$. 
\end{lemma}

\begin{proof}
	Applying Lemma \ref{PreliminaryLemma}, we obtain
	\begin{align*}
		D(\rho\|\tau)\notag&= \tr\big( \rho(\ln  \rho-\ln  \tau) \big)\nonumber\\
		& = \sum_{\beta\cdot k\leq t} \bra{k}\rho(\ln \rho-\ln \tau)\ket{k}+\sum_{\beta\cdot k> t}\bra{k}\rho(\ln \rho-\ln \tau)\ket{k}\\
		\notag&\leq\sum_{\beta\cdot k\leq t}\bra{k}\big(\rho-\tau+(\rho-\tau)^2 \tau^{-1}\big)\ket{k}-\sum_{\beta\cdot k> t}\bra{k}\rho\ln \tau\ket{k}\\
		\notag&=\sum_{\beta\cdot k\leq t}\bra{k}(\rho-\tau)^2 \tau^{-1}\ket{k}+\sum_{\beta\cdot k> t} \Big( \bra{k}(\tau-\rho)\ket{k}-\bra{k}\rho\ln \tau\ket{k}\Big),
	\end{align*}
	where in the last line we use $\tr(\rho-\tau)=0$. Using $\rho\ln  \tau = (\rho-\tau)\ln  \tau+ \tau\ln  \tau$ and $\ln  \tau = \sum_{j=1}^m \left(-\beta_j \bfa_j^\dagger\bfa_j + \ln (1-e^{-\beta_j})\right)$ and denoting 
	\begin{align}\label{eq:nu-beta}
		\nu_\beta:=\prod_{j=1}^m \big(1-e^{-\beta_j}\big),
	\end{align}
	we get
	\begin{align*}	
		D(\rho\|\tau) &\leq \sum_{\beta\cdot k\leq t}\bra{k}(\rho-\tau)^2 \tau^{-1}\ket{k}+\sum_{\beta\cdot k> t}\bra{k}(\rho-\tau) \ket{k}\big(\beta\cdot k -\ln  \nu_\beta -1\big)  -\sum_{\beta\cdot k> t}\bra{k}\tau\ln \tau\ket{k}\\
		&\leq \sum_{\beta\cdot k\leq t} \bra{k}(\rho-\tau)^2 \tau^{-1}\ket{k}+\eta_\beta\sum_{\beta\cdot k> t} (|k|+m) |\bra{k}(\rho-\tau) \ket{k}|  -\sum_{\beta\cdot k> t}\bra{k}\tau\ln \tau\ket{k}.
	\end{align*}
	This inequality is equivalent to~\eqref{preliminary-ineq}.  	
\end{proof}

\medskip

We are now ready to present the proof of Theorem~\ref{boundEntropy-Char}.
\begin{proof}[Proof of Theorem~\ref{boundEntropy-Char}]
	We first note that $\tau_\alpha$ is given by
	\begin{align}\label{eq:sigma-alpha-chi-rep}
		\tau_\alpha = \frac{1}{\pi^m} \int_{\mathbb{C}^m} \chi_\tau(z)(1+\alpha E(z)) D_{-z} \dd^{2m} z,
	\end{align}
	and it is well-defined since $\chi_{\tau}(z)$ is a Gaussian function and $E(z)$ is a polynomial. To obtain a more explicit expression for $\tau_\alpha$, from~\eqref{eq:characterisitc-derivative-0} and~\eqref{eq:characterisitc-derivative-1} we find that 
	$$ z_j\chi_T(z) = \chi_{T\bfa_j - \bfa_j T}(z), \qquad \quad \bar z_j\chi_T(z) = \chi_{T\bfa_j^\dagger - \bfa_j^\dagger T}(z).$$
	Next, using the fact that $\tau$ is thermal state, it is not hard to verify that $\tau \bfa_j = e^{\beta_j} \bfa_j\tau$ and $\tau \bfa_j^\dagger  = e^{-\beta_j}\bfa_j^\dagger \tau$. Putting these together and using the commutation relations~\eqref{eq:CCR-AnnCre} we conclude that $\tau_\alpha$ takes the form 
	\begin{align}\label{eq:sigma-alpha-explicit-form}
		\tau_{\alpha} = \tau + \alpha \sum_{p, q} e_{p, q} \big(\bfa_1^\dagger\big)^{p_1} \bfa_1^{q_1} \cdots \big(\bfa_m^\dagger\big)^{p_m} \bfa_m^{q_m} \tau.
	\end{align} 
	Since $E(z)$ is an odd polynomial, $e_{p, q}\neq 0$ only if $\sum_j (p_j+q_j)$ is an odd number not exceeding the total degree of $E(z)$. Note that to write down the above expression we in particular use $\bfa_j \bfa_j^\dagger = \bfa_j^\dagger \bfa_j+1$, yet this would not increase the degrees of the terms on the right hand side nor it changes their parity. 
	This representation of $\tau_\alpha$ implies that the diagonal entries of $\tau_\alpha$ with respect to the Fock basis are independent of $\alpha$:
	\begin{align}\label{eq:diag-sigma-alpha-indep}
		\bra k \tau_{\alpha}\ket k = \bra k \tau\ket k.
	\end{align}
	Also, since  $E(-z)=\overline{E(z)}$, equation~\eqref{eq:sigma-alpha-chi-rep} shows that $\tau_\alpha$ is self-adjoint.

	Applying Lemma \ref{TruncationLemma} we need  to bound each term on the right hand side of~\eqref{preliminary-ineq}.
	
	\bigskip
	\noindent
	\underline{\emph{First term on the right hand side of~\eqref{preliminary-ineq}:}} 
	We start with
	\begin{align*}
		\sum_{\beta\cdot k\leq t}\bra{k}(\rho-\tau)^2\tau^{-1}\ket{k}
		&=\sum_{\beta\cdot k\leq t}\bra{k}\tau^{-1/2}(\rho-\tau_\alpha+\tau_\alpha-\tau)^2\tau^{-1/2}\ket{k}\\
		&\leq 2\sum_{\beta\cdot k\leq t}\bra{k}\tau^{-\frac{1}{2}}(\rho-{\tau}_\alpha)^2\tau^{-\frac{1}{2}}\ket{k}+2\sum_{\beta\cdot k\leq t}\bra{k}\tau^{-\frac{1}{2}}({\tau_\alpha}-\tau)^2\tau^{-\frac{1}{2}}\ket{k},
	\end{align*}
	where in the preceding inequality we use $(A+B)^2\leq 2(A^2+B^2)$. To bound the first term, using the fact that 
	$$x\mapsto \frac {1}{(x+m)^m}e^x,$$ 
	is monotone increasing on $(0, +\infty)$, we compute
	\begin{align*}
		\sum_{\beta\cdot k\leq t}\bra{k}(\rho-{\tau_\alpha})^2\tau^{-1}\ket{k}
		&\leq \frac{1}{\nu_\beta}\sum_{\beta\cdot k\leq t}e^{\beta\cdot k}\bra{k}(\rho-{\tau_\alpha})^2\ket{k}\\
		&\leq \frac{e^{t}}{\nu_\beta(t + m)^m}  \sum_{\beta\cdot k\leq t} (\beta\cdot k+m)^m\bra{k}(\rho-{\tau_\alpha})^2\ket{k}\\
		&\leq   \frac{(\beta_{\max}+1)e^{t}}{\nu_\beta(t + m)^m}  \sum_{\beta\cdot k\leq t} (|k|+m)^m\bra{k}(\rho-{\tau_\alpha})^2\ket{k}\\
		&\leq  \frac{(\beta_{\max}+1)e^{t}}{\nu_\beta(t + m)^m}   \sum_{k} (|k|+m)^m\bra{k}(\rho-{\tau_\alpha})^2\ket{k}\\
		& =  \frac{(\beta_{\max}+1)e^{t}}{\nu_\beta(t + 1)^m} \tr\big(  (\rho-\tau_\alpha)^2 (N_m+m)^{m}  \big).
	\end{align*}
We continue by analyzing the second term, $\sum_{\beta\cdot k\leq t}\bra{k}\tau^{-\frac{1}{2}}({\tau_\alpha}-\tau)^2\tau^{-\frac{1}{2}}\ket{k}$. To this end, we let $\Delta:=\frac{1}{\alpha}({\tau_\alpha}-\tau)$. Then, invoking~\eqref{eq:sigma-alpha-explicit-form} we have
	\begin{align*}
		\Delta =  \sum_{p, q} e_{p, q} \big(\bfa_1^\dagger\big)^{p_1} \bfa_1^{q_1} \cdots \big(\bfa_m^\dagger\big)^{p_m} \bfa_m^{q_m}  \tau,
	\end{align*} 
	and hence
	\begin{align*}
		\Delta^2\tau^{-1}=\sum e_{p, q}e_{p', q'} \big(\bfa_1^\dagger\big)^{p_1} \bfa_1^{q_1} \cdots \big(\bfa_m^\dagger\big)^{p_m} \bfa_m^{q_m}  \tau \big(\bfa_1^\dagger\big)^{p'_1} \bfa_1^{q'_1} \cdots \big(\bfa_m^\dagger\big)^{p'_m} \bfa_m^{q'_m} .
	\end{align*}
	This gives
	\begin{align*}
		\sum_{\beta\cdot k\leq t}\bra{k}\tau^{-\frac{1}{2}}({\tau_\alpha}-\tau)^2\tau^{-\frac{1}{2}}\ket{k} 	\notag&\leq\alpha^2\sum_{k}\bra{k}\tau^{-\frac{1}{2}}\Delta^2\tau^{-\frac{1}{2}}\ket{k}\\
		&=\alpha^2 \tr(\Delta^2\tau^{-1})\\
		&\leq \alpha^2 \sum \bigg| e_{p, q}e_{p', q'}   \tr\bigg( \big(\bfa_1^\dagger\big)^{p'_1} \bfa_1^{q'_1}  \big(\bfa_m^\dagger\big)^{p'_m} \bfa_m^{q'_m}  \big(\bfa_1^\dagger\big)^{p_1} \bfa_1^{q_1} \cdots \big(\bfa_m^\dagger\big)^{p_m} \bfa_m^{q_m}  \tau  \bigg)\bigg|\\
		&=M_{\beta, E}  \alpha^2,
	\end{align*}
	where $M_{\beta,E}>0$ is a constant depending only on $\beta$ and the polynomial $E(z)$. Note that $M_{\beta, E}$ is finite since $\tau$ is Gaussian and all of whose moments are finite. 	We conclude that 
	\begin{align}\label{eq:prop1-first-term}
		&\sum_{\beta\cdot k\leq t}\bra{k}(\rho-\tau)^2\tau^{-1}\ket{k}\nonumber\\
		&\qquad \leq 2\sum_{\beta\cdot k\leq t}\bra{k}\tau^{-\frac{1}{2}}(\rho-{\tau}_\alpha)^2\tau^{-\frac{1}{2}}\ket{k}+2\sum_{\beta\cdot k\leq t}\bra{k}\tau^{-\frac{1}{2}}({\tau_\alpha}-\tau)^2\tau^{-\frac{1}{2}}\ket{k}\nonumber\\
		&\qquad \leq  2M_{\beta, E}  \alpha^2+  C'_{\beta}\frac{e^{t}}{(k+1)^m} \tr\big(  (\rho-\tau_\alpha)^2 (N_m+m)^{m}  \big),
	\end{align}
	where 
	$$C'_{\beta} =  2\frac{\beta_{\max}+1}{\nu_\beta}.$$
	
	\bigskip
	\noindent
	\underline{\emph{Second term on the right hand side of~\eqref{preliminary-ineq}:}}  For $s=1+\frac 1m$ by~\eqref{eq:diag-sigma-alpha-indep} and the Cauchy--Schwarz inequality we have
	\begin{align*}
		&\sum_{\beta\cdot k> t}\big|\bra{k}(\rho-\tau)(N_m+m)\ket{k}\big|\\
		&\qquad \leq \sum_{ k}\big|\bra{k}(\rho-\tau_\alpha)(N_m+m)\ket{k}\big|\\
		&\qquad = \sum_{k}\frac{1}{(|k|+m)^{\frac{s}{2}m}}\big|\bra{k}(\rho-{\tau_\alpha})(N_m+m)^{\frac {s}{2}m +1}\ket{k}\big|\\
		&\qquad \leq\left(\sum_{k}\frac{1}{(|k|+m)^{sm}}\right)^{\frac{1}{2}}\left(\sum_{k} \Big|\bra k(\rho-{\tau_\alpha})(N_m+m)^{\frac {s}{2}m +1}\ket{k}\Big|^2\right)^{\frac{1}{2}}.	
	\end{align*}
	Next, we note that 
	$$\big(|k|+m\big)^{m}=\big((k_1+1)+\cdots + (k_m+1)\big)^{m}\geq (k_1+1)\cdots (k_m+1).$$
	Therefore,
	\begin{align}\label{eq:prop1-second-term}
		\sum_{\beta\cdot k> t}\big|\bra{k}(\rho-\tau)N_m\ket{k}\big|
		\leq \zeta^{m/2}  \Big\|(\rho-{\tau_\alpha})(N_m+m)^{\frac {s}{2}m +1}\Big\|_2,
	\end{align}
	where $\zeta = \sum_{k=1}^\infty \frac{1}{k^s}.$

	\bigskip
	\noindent
	\underline{\emph{Third term on the right hand side of~\eqref{preliminary-ineq}:}}	For the last term we have
	\begin{align}\label{eq:prop1-third-term}
		-\sum_{\beta\cdot k> t}\bra{k}\tau\ln \tau\ket{k}
		&=\nu_\beta\sum_{\beta\cdot k> t} (\beta\cdot k -\ln  \nu_\beta)  e^{-\beta\cdot k }\nonumber\\
		&\leq \nu_\beta   \eta_\beta  \sum_{\beta\cdot k> t} (|k|+m)e^{-\beta\cdot k }\nonumber\\
		& = m\nu_\beta   \eta_\beta  \sum_{\beta\cdot k> t} (k_1+1)e^{-\beta\cdot k }\nonumber\\
		& \leq C''_{\beta} (t+1)^me^{-t},
	\end{align}
	where 
	$$C''_{\beta} = \frac{2^m m\eta_\beta (\beta_1+1)^2\cdots (\beta_m+1)^2}{\nu_\beta},$$
	and the last inequality is proven in Appendix~\ref{app:f-g-estimate}.
	
	\bigskip	
	
	Now, combining \eqref{eq:prop1-first-term}--\eqref{eq:prop1-third-term} yields
	\begin{align*}
		D(\rho\|\tau)&\leq  2M_{\beta, E}  \alpha^2+  C'_{\beta}\frac{e^{t}}{(t+1)^m} \tr\big(  (\rho-\tau_\alpha)^2 (N_m+m)^{m}  \big) \nonumber\\
		&\quad +  \eta_\beta \zeta^{m/2} \Big\|(\rho-{\tau_\alpha})(N_m+m)^{\frac {s}{2}m +1}\Big\|_2 +C''_{\beta} (t+1)^m e^{-t}.
	\end{align*}
	Note that $(N_m+m)^{m}\leq (N_m+m)^{m+3}$. Therefore, 
	\begin{align*}
		D(\rho\|\tau)&\leq  2M_{\beta, E}  \alpha^2+  C'_{\beta}\frac{e^{t}}{(t+1)^m} \Big\|  (\rho-\tau_\alpha) (N_m+m)^{(m+3)/2}  \Big\|^2 \\
		&\quad +  \eta_\beta \zeta^{m/2} \Big\|(\rho-{\tau_\alpha})(N_m+m)^{(m+3)/2}\Big\|_2 +C''_{\beta} (t+1)^m e^{-t}.
	\end{align*}
	Let $\epsilon=\Big\|(\rho-{\tau_\alpha})(N_m+m)^{(m+3)/2}\Big\|_2$. If $\epsilon=0$, take the limit $t\to +\infty$ in the above inequality to get $D(\rho\| \tau) \leq 2M_{\beta, E}\alpha^2$. 
	If $0<\epsilon< 1$, there exists $t> 0$ such that
	$$
	\frac{e^{ t}}{(t+1)^m}= \frac{1}{\epsilon}.
	$$
	Then, we find that
	$$D(\rho\| \tau) \leq \max\big\{2M_{\beta, E},\, C'_\beta+ \eta_\beta \zeta^{m/2} + C''_\beta\big\} \Big(\alpha^2 + \Big\|(\rho-{\tau_\alpha})(N_m+m)^{(m+3)/2}\Big\|_2 \Big).$$
	If $\epsilon\geq 1$, we apply Lemma~\ref{TruncationLemma} for $t=0$ to get
	\begin{align*}
		D(\rho\|\tau)&\leq \bra{0}(\rho-\tau)^2 \tau^{-1}\ket{0} +\eta_\beta \sum_{k\neq 0} \big|\bra{k}(\rho-\tau)(N_m+m)\ket{k}\big|-\sum_{k\neq 0}\bra{k}\tau\ln \tau\ket{k}\\
		& \leq \nu_\beta + \frac{1}{\nu_\beta} + \eta_{\beta}\zeta^{m/2} \epsilon + S(\tau)\\
		&\leq \Big(\nu_\beta + \frac{1}{\nu_\beta} + \eta_{\beta}\zeta^{m/2}  + S(\tau)\Big)(\alpha^2 + \epsilon),
	\end{align*}
	where $S(\tau) = -\tr(\tau \ln  \tau)$ is the von Neumann entropy of $\tau$ which depends only on $\beta$. 
	All in all, for any value of $\epsilon\geq 0$, the desired inequality holds for 
	$$C_{\beta, E} = \max\big\{2M_{\beta, E},\, C'_\beta+ \eta_\beta \zeta^{m/2} + C''_\beta,\, \nu_\beta + \frac{1}{\nu_\beta} + \eta_{\beta}\zeta^{m/2}  + S(\tau)\big\},$$ 
	
	To prove the theorem for the case where $\tau$ is not necessarily thermal, we can simply replace $\rho$ and  $\tau$ with $U\rho U^\dagger$ and $U \tau U^\dagger$ where $U$ is a Gaussian unitary that turns $\tau$ into its Williamson's form. Since $\boldsymbol{\gamma}(\tau) +i\Omega_m$ is full-rank, the thermal state $U\tau U^\dagger$ satisfies the assumption of the first part of the proposition (i.e., $0<\beta_j<+\infty$). Also, $\tau_\alpha$ should be replaced with $U \tau_\alpha U^\dagger$, and then the polynomial $E(z)$ should be modified as well. 
	We note that since $\tau$ is centered, the Gaussian unitary $U$ does not change the first moment of $\tau$~\cite{Serafini}. In this case, using~\eqref{eq:sigma-alpha-chi-rep} or~\eqref{eq:sigma-alpha-explicit-form}, it can be verified that the polynomial $E(z)$ still satisfies the assumptions of the theorem.
\end{proof}

%%%%%%%%%%%%%%%%%%%%%%%%%%%%%%%%%%%%%%%%%%%%%%%%%%%%%%%%%%%%%%%%%%%%%%%%%%%%%%%%%%%%%%%%%%%%%%%%%%%%%%%%%%%%%%%%%%%%%%

%**********************************
\section{Proof of Theorem \ref{MainTheoremTrace}}\label{Sec:Proof-Handle}

Since the case $m=1$ is already proven in~\cite{beigi2023towards}, we assume $m \geq 2$. By applying a suitable Gaussian unitary and using~\eqref{eq:Conv-Unit}, we may, without loss of generality, assume that $\rho$ is in the Williamson form with a diagonal covariance matrix as in~\eqref{eq:Moments-Williamson}. Let $\nu_{\min}$ and $\nu_{\max}$ denote the minimum and maximum eigenvalues of $\boldsymbol{\gamma}(\rho)$, respectively. As discussed in Subsection~\ref{subsec:Gaussian-states}, if $\nu_j=1$ for some $j$, then the $j$-th mode of $\rho$ is in the vacuum state and decoupled from the other modes. Since the $n$-fold convolution of the vacuum remains the vacuum, this mode does not contribute to the trace distance $\big\| \rho^{\boxplus n} - \rho_G \big\|_1$. Consequently, we can disregard the $j$-th mode in the proof of Theorem~\ref{MainTheoremTrace}. Thus, we assume in the following that the minimum eigenvalue of the covariance matrix is strictly greater than one:
\begin{align}\label{eq:assumption-nu-min-1}
\nu_{\min}>1.
\end{align}

Using the notation developed in Section~\ref{Sec:State-Truncation} and applying the triangle inequality, by Lemma~\ref{Aux1} we have
\begin{align}
	\big\| \rho^{\boxplus n} -  \rho_G \big\|_1 &\leq \big\| \rho^{\boxplus n} -  \sigma_n^{\boxplus n} \big\|_1 + \big\|  \sigma_n^{\boxplus n} - \tilde \sigma_{n,G} \big\|_1 + \big\| \tilde \sigma_{n,G}-\rho_G  \big\|_1 \nonumber \\
	&= \mathcal{O}\Big(\frac{1}{\sqrt{n}}\Big) +  \big\|  \sigma_n^{\boxplus n} - \tilde \sigma_{n,G} \big\|_1 + \big\| \rho_G - \tilde \sigma_{n,G} \big\|_1.\label{TruncationBound}
\end{align}
On the other hand,
\begin{align*}
	\big\|  \sigma_n^{\boxplus n} - \tilde \sigma_{n,G} \big\|_1 &\leq \big\|  \sigma_n^{\boxplus n} - D_{\sqrt n z_n}^\dagger \tilde \sigma_{n,G} D_{\sqrt n z_n}\big \|_1 +\big \| \tilde \sigma_{n,G} - D_{\sqrt n z_n}^\dagger \tilde \sigma_{n,G} D_{\sqrt n z_n} \big\|_1 \\
	&= \big\|  \tilde{\sigma}_n^{\boxplus n} - \tilde \sigma_{n,G} \big\|_1 + \big\| \tilde \sigma_{n,G} - D_{\sqrt n z_n}^\dagger \tilde \sigma_{n,G} D_{\sqrt n z_n} \big\|_1,
\end{align*}
and hence
\begin{equation}\label{FirstTriangle}
	\big\| \rho^{\boxplus n} -  \rho_G \big\|_1 \leq  \big\|  \tilde{\sigma}_n^{\boxplus n} - \tilde \sigma_{n,G} \big\|_1 + \big\| \tilde \sigma_{n,G} - D_{\sqrt n z_n}^\dagger \tilde \sigma_{n,G} D_{\sqrt n z_n} \big\|_1 + \big \| \rho_G - \tilde \sigma_{n,G} \big\|_1+\mathcal{O}\Big(\frac{1}{\sqrt{n}}\Big).
\end{equation}
As a result, the proof of the theorem amounts to obtaining the desired $\mathcal{O}\big(\frac{1}{\sqrt{n}}\big)$ upper bound for the remaining three terms on the right hand side of~\eqref{FirstTriangle}.

\begin{figure}[t!]
\label{fig:thm1}
\begin{center}
\begin{tikzpicture}
  [transition/.style={rectangle, rounded corners,draw=black!50,fill=black!20,
                      inner sep=3pt,minimum size=5mm, align=center}, point/.style={circle,inner sep=0pt,minimum size=0pt,fill=black}]
  \node[transition] (theorem 1)  {\small Theorem \ref{MainTheoremTrace}};
   \node[transition] (triangle)  [below=4mm of theorem 1]  {\small Triangle inequality yields three terms in \eqref{FirstTriangle}\\ \small that are handled in Claims \ref{claim:I} , \ref{claim:II} , \ref{claim:III} };
  \node[transition] (claim 2) [below=8mm of triangle] {\small Claim \ref{claim:II} };
   \node[transition] (proof 2) [below=8mm of claim 2.center] {\small Claim \ref{claim:I}, Lemma \ref{bound-FirstNorm}};
  \node[transition] (proof 1) [left=of proof 2] {\small Lemma \ref{lem:norm1-norm2}, Lemma \ref{OperatorNormCovarianceMatrix}};
    \node[transition] (proof 3) [right=of proof 2] {\small Lemma \ref{lem:norm1-norm2} \& splitting the\\ \small integral yield four terms in \eqref{splitInt} };

    \node[transition] (claim 3) [above=8mm of proof 3.center] {\small Claim \ref{claim:III}};
  \node[transition] (claim 1) [above=8mm of proof 1.center] {\small Claim \ref{claim:I}};
  \node[point] (dot L1)[above=4mm of claim 3]  {};
  \node[point] (dot R1)[above=4mm of claim 1]  {};

  \path [line width=0.5pt] (theorem 1) edge[->] (triangle);
  \path [line width=0.5pt] (triangle) edge[->] (claim 2);
  \draw [line width=0.5pt] (dot R1) -- (dot L1);
  \path [line width=0.5pt] (dot R1)   edge[->] (claim 1);
  \path [line width=0.5pt] (dot L1)   edge[->] (claim 3);
  \path [line width=0.5pt] (claim 1)   edge[->] (proof 1);
  \path [line width=0.5pt] (claim 2)   edge[->] (proof 2);
  \path [line width=0.5pt] (claim 3)   edge[->] (proof 3);

  \node[point] (dot ineq 1) [below=7mm of proof 3] {};
  \node[point] (dot ineq 2) [below=7mm of dot ineq 1] {};
  \node[point] (dot ineq 3) [below=7mm of dot ineq 2] {};
  \node[point] (dot ineq 4) [below=7mm of dot ineq 3] {};

  \node[transition] (ineq 1) [left=12mm of dot ineq 1] {\small First term in \eqref{splitInt}: Lemma \ref{BRMethod}, Lemma \ref{OperatorNormCovarianceMatrix}, Lemma \ref{EdgeWorth-BindDerivatives}, Lemma \ref{lem:BindMoments-Truncation}};
  \node[transition] (ineq 2) [left=4mm of dot ineq 2] {\small Second term in \eqref{splitInt}: \cite[Proposition 3]{beigi2023towards}, Proposition \ref{prop:CharSupOrigin}, Lemma \ref{OperatorNormCovarianceMatrix}, Lemma \ref{lem:tr-A-kappa-bounded}};
  \node[transition] (ineq 3) [left=18mm of dot ineq 3] {\small Third term in \eqref{splitInt}: Lemma \ref{lem:bound-chi-tilde-sigma}, Lemma \ref{lem:tr-A-kappa-bounded}, Lemma \ref{lem:BindMoments-Truncation}};
  \node[transition] (ineq 4) [left=18mm of dot ineq 4] {\small  Fourth term in \eqref{splitInt}: Lemma \ref{lem:bound-chi-tilde-sigma}, Lemma \ref{lem:tr-A-kappa-bounded}, Lemma \ref{OperatorNormCovarianceMatrix}};

  \draw [line width=0.5pt] (proof 3) -- (dot ineq 4);
  \path [line width=0.5pt] (dot ineq 1)   edge[->] (ineq 1);
  \path [line width=0.5pt] (dot ineq 2)   edge[->] (ineq 2);
  \path [line width=0.5pt] (dot ineq 3)   edge[->] (ineq 3);
  \path [line width=0.5pt] (dot ineq 4)   edge[->] (ineq 4);

  \end{tikzpicture}
\end{center}
\caption{To prove Theorem \ref{MainTheoremTrace} we first apply the triangle inequality to reduce the problem into bounding the three terms on the right hand side of (60). The bound on these three terms are proven in Claims \ref{claim:I}, \ref{claim:II}, \ref{claim:III}. The diagram presents results from the  previous sections that are used to prove these claims.}
\end{figure}

The upper bound for the third term can be established using Lemma~\ref{OperatorNormCovarianceMatrix} that states the fast convergence of the covariance matrix of  $\tilde \sigma_{n,G}$ to that of $\rho_G$ as $n\to \infty$. Nevertheless, we need to translate the convergence in terms of covariance matrices into the convergence in terms of trace distance. We bridge this gap by using Lemma~\ref{lem:norm1-norm2}. The details of this analysis are given in Claim~\ref{claim:I} below.

To get the desired upper bound for the second term, we may replace $\tilde \sigma_{n,G}$ with $\rho_G$, as by analyzing the third term, we already know that they are close to each other in terms of trace distance. The desired upper bound then follows from our previous bounds on the first-order moments in Lemma~\ref{bound-FirstNorm}. We handle the details in Claim~\ref{claim:II} below.

To prove the upper bound for the first term, we once again use Lemma~\ref{lem:norm1-norm2}. We note that all the moments of $\tilde{\sigma}_n$ and $\tilde \sigma_{n,G}$ are finite (although they depend on $n$). Thus, Lemma~\ref{lem:norm1-norm2} can be applied safely to get a bound in terms of the characteristic function of $A^{\dagger}(\tilde{\sigma}_n^{\boxplus n} - \tilde \sigma_{n,G})A$. Next, to bound the norm of this characteristic function, we split the resulting integral into two regions: one over points away from the origin, and the other one over points near the origin. By Proposition~\ref{prop:CharSupOrigin} we know that the modulus of the characteristic function over points in the first region is strictly less than $1$. This crucial property can be utilized to show that the first integral is at least sub-exponentially small. For the second integral we can use the Edgeworth-type expansion developed earlier to estimate the characteristic function near the origin. In particular, we can apply Lemma~\ref{BRMethod} to control the derivatives of the characteristic function of $\tilde{\sigma}_n^{\boxplus n} - \tilde \sigma_{n,G}$ near the origin. Also, with the help of Lemma~\ref{EdgeWorth-BindDerivatives}, we obtain an upper bound on the derivatives of the additional polynomial terms produced by Lemma~\ref{BRMethod}, expressed in terms of the moments of $\tilde{\sigma}_n$. The desired upper bound for the second integral then follows from Lemma~\ref{OperatorNormCovarianceMatrix}  and \eqref{BindMoments-Truncation}, which provide upper bounds for the moments of the truncated state. Moreover, Lemma~\ref{lem:bound-chi-tilde-sigma} allows us to control powers of the characteristic function of $\tilde{\sigma}_n$ in a straightforward manner. The detailed computation for this term is provided in Claim~\ref{claim:III} below.

\bigskip
\begin{claim}\label{claim:I} 
$\big\| \rho_G - \tilde \sigma_{n,G} \big\|_{1}  =\mathcal{O}\Big(\frac{1}{\sqrt{n}}\Big)$.
\end{claim}

\medskip

Although this claim can be proven using~\cite[Theorem S18]{bittel2025optimal} and Lemma~\ref{OperatorNormCovarianceMatrix}, we present an alternative argument here.

Invoking Lemma~\ref{lem:norm1-norm2} with $T_n := \rho_G - \tilde \sigma_{n,G}$, we get
\[
\Big(\frac{6}{\pi^2}\Big)^{m}  \| T_n \|_{1}^{2}  \leq \int \big|\chi_{A^{\dagger} T_n A}(z)\big|^{2} \dd^{2m} z,
\]
where $A = \ac_1 \cdots \ac_m$. Note that from~\eqref{eq:displacement-derivative}, we have
\begin{align*}
\chi_{\ac_j^{\dagger}T_n}(z) = \Big(\partial_{z_j} - \frac{1}{2} \bar z_{j}\Big) \chi_{T_n}(z),
\qquad
	\chi_{T_n \ac_j}(z) = \Big(- \partial_{\bar z_{j}} + \frac{1}{2} z_j \Big) \chi_{T_n}(z).
\end{align*}
Hence, there exist some constants $c_{\alpha, \alpha', \beta, \beta'}$ such that
\begin{align*}
	\chi_{A^{\dagger} T_n A} (z) 
	&=  \prod_{j=1}^{m} \Big(\partial_{z_j} - \frac{1}{2} \bar z_{j}\Big) \Big(- \partial_{\bar z_{j}} + \frac{1}{2} z_j \Big)  \chi_{T_n}(z) \\
	&= \sum_{\alpha, \alpha', \beta, \beta'} c_{\alpha, \alpha', \beta, \beta'} z^{\alpha'}\bar{z}^{\beta'} \partial_{z}^{\alpha} \partial_{\bar z}^\beta  \chi_{T_n}(z),
\end{align*}
where the sum runs over all $\alpha, \alpha', \beta, \beta'\in \{0,1\}^m$ satisfying $|\alpha| + |\alpha'|+|\beta| + |\beta'|\leq 2m$. We note that here the number of summands is a constant depending only on $m$, the number of modes.
Therefore, we have
\begin{equation}\label{ChiDerivativeGaussian}
	\big|\chi_{A^{\dagger} T_n A} (z)\big|^{2} \leq  C' \sum_{\alpha, \alpha', \beta, \beta'} \prod_{j=1}^{m} {|z_{j}|}^{(\alpha'_j+\beta'_j)} \big|\partial_z^\alpha\partial_{\bar z}^\beta \chi_{T_n}(z) \big|^{2},
\end{equation}
for some absolute constant $C'>0$. Note that since both $\rho_G$ and $\tilde \sigma_{n, G}$ are Gaussian, $\chi_{T_n}(z)$ can be expressed in terms of $\boldsymbol{\gamma}(\rho_G)$ and $\boldsymbol{\gamma}(\tilde \sigma_{n,G})$ by using~\eqref{eq:Char-NonWilliam}:
\begin{align*}
	\chi_{T_n}(z) &= \chi_{\rho_G}(z) - \chi_{\tilde \sigma_{n,G}}(z) \\
	&= \exp\Big(-\frac 14 \hat{z}^\dagger \Lambda_m^\dagger \boldsymbol{\gamma}(\rho_G) \Lambda_m  \hat{z}\Big) - \exp\Big(-\frac 14 \hat{z}^\dagger \Lambda_m^\dagger \boldsymbol{\gamma}(\tilde \sigma_{n,G}) \Lambda_m  \hat{z}\Big)\\
	&= \chi_{\rho_G}(z)\cdot \Big\{ 1 - \exp\Big(-\frac 14 \hat{z}^\dagger \Lambda_m^\dagger \big(\boldsymbol{\gamma}(\tilde \sigma_{n,G}) - \boldsymbol{\gamma}(\rho_G)\big) \Lambda_m \hat{z}\Big) \Big\},
\end{align*}
where $\hat{z} = (z_1, \bar{z}_1, \cdots, z_m, \bar{z}_m)^\top$ and $\Lambda_m$ is defined in~\eqref{eq:Lambda_m}. Hence, by using~\eqref{ChiDerivativeGaussian}  we get
\begin{align}\label{DerivativeDifferenceGaussian}
	&\big|\chi_{A^{\dagger} T_n A} (z)\big|^{2}\nonumber\\ 
	&\quad \leq \sum P_{\alpha, \beta}(\|z\|) \Big|\partial_z^{\alpha_1} \partial_{\bar z}^{\beta_1} \chi_{\rho_G}(z) \cdot \partial_z^{\alpha_2} \partial_{\bar z}^{\beta_2}  \Big\{ 1 - \exp\Big(-\frac 14 \hat{z}^\dagger \Lambda_m^\dagger \big(\boldsymbol{\gamma}(\tilde \sigma_{n,G}) - \boldsymbol{\gamma}(\rho_G)\big) \Lambda_m \hat{z}\Big) \Big\} \Big|^{2},
\end{align}
where the sum runs over all $\alpha_1,\alpha_2,\beta_1,\beta_2$ such that $\alpha=\alpha_1+\alpha_2$ and $\beta=\beta_1+\beta_2$ belong to $\{0,1\}^m$, and $P_{\alpha, \beta}(z)$ is a polynomial function in $z$ and $\bar z$ with coefficients that are independent of $n$. 

The first key observation is that $\partial_z^{\alpha_1} \partial_{\bar z}^{\beta_1} \chi_{\rho_G}(z)$ can be expressed as $Q_{\alpha_1, \beta_1}(z) \chi_{\rho_G}(z)$, where $Q_{\alpha_1, \beta_1}(z)$ is a polynomial that does not depend on $n$. For the second term, if $\alpha_2 = \beta_2 = 0$, we apply $|1 - e^x| \leq |x| e^{|x|}$ along with the unitarity of $\Lambda_m$ to obtain
\begin{align}
	\Big|1 - \exp\Big(-\frac 14 \hat{z}^\dagger &\Lambda_m^\dagger \big(\boldsymbol{\gamma}(\tilde \sigma_{n,G}) - \boldsymbol{\gamma}(\rho_G)\big) \Lambda_m \hat{z}\Big)\Big| \nonumber \\
	&\leq \Big|\frac 14 \hat{z}^\dagger \Lambda_m^\dagger \big(\boldsymbol{\gamma}(\tilde \sigma_{n,G}) - \boldsymbol{\gamma}(\rho_G)\big) \Lambda_m \hat{z}\Big| \cdot \exp\Big(\Big|\frac 14 \hat{z}^\dagger \Lambda_m^\dagger \big(\boldsymbol{\gamma}(\tilde \sigma_{n,G}) - \boldsymbol{\gamma}(\rho_G)\big) \Lambda_m \hat{z}\Big|\Big) \nonumber \\
	&\leq \frac 14 \big\| \boldsymbol{\gamma}(\tilde \sigma_{n,G}) - \boldsymbol{\gamma}(\rho_G) \big\|\cdot \big\| \Lambda_m \hat{z}\big\|^2 \cdot \exp\Big(\Big|\frac 14 \hat{z}^\dagger \Lambda_m^\dagger \big(\boldsymbol{\gamma}(\tilde \sigma_{n,G}) - \boldsymbol{\gamma}(\rho_G)\big) \Lambda_m \hat{z}\Big|\Big) \nonumber \\
	&\leq \mathcal{O}\Big(\frac{1}{\sqrt n}\Big) \|  \hat{z}\|^2 \exp\Big(\frac{\delta}{4} \|  \hat{z}\|^2 \Big) \label{diverg_point_Ent1},
\end{align}
where in the last line we use Lemma~\ref{OperatorNormCovarianceMatrix}, and $\delta>0$ is an arbitrary constant to be determined later. The point is that, as a consequence of Lemma~\ref{OperatorNormCovarianceMatrix}, for any arbitrary $\delta>0$ and sufficiently large $n$, we have $\big\| \Lambda_m^\dagger \left( \boldsymbol{\gamma}(\tilde \sigma_{n,G}) - \boldsymbol{\gamma}(\rho_G)\right) \Lambda_m \big\| \leq \delta$. If $\alpha_2 \neq 0$ or $\beta_2 \neq 0$, we use $\big|\partial_z^{\alpha_2} \partial_{\bar z}^{\beta_2} (1-\exp f(z))\big| = \big|\partial_z^{\alpha_2} \partial_{\bar z}^{\beta_2} \exp f(z)\big|$, along with $|e^x| \leq e^{|x|}$, and apply a similar argument as above to obtain
\begin{align}
\Big| \partial_z^{\alpha_2} \partial_{\bar z}^{\beta_2}  \big( 1 - \exp(-\frac 14 \hat{z}^\dagger &\Lambda_m^\dagger (\boldsymbol{\gamma}(\tilde \sigma_{n,G}) - \boldsymbol{\gamma}(\rho_G)) \Lambda_m \hat{z}) \big) \Big| \nonumber \\
&\leq \mathcal{O}\Big(\frac{1}{\sqrt n}\Big) \tilde Q_{\alpha_2, \beta_2}(\|z\|) \exp\Big(\frac{\delta}{4} \|  \hat{z}\|^2 \Big)\label{diverg_point_Ent2},
\end{align}
where $\tilde Q_{\alpha_2, \beta_2}(\cdot)$ is some polynomial independent of $n$. 
Putting all these together in~\eqref{DerivativeDifferenceGaussian}, we obtain
\[
	\big|\chi_{A^{\dagger} T_n A} (z)\big|^{2} \leq \mathcal{O}\Big(\frac{1}{ n}\Big) P(\| z\|) \chi_{\rho_G}(z)^2 \exp\Big(\frac{\delta}{2} \|  \hat{z}\|^2 \Big),
\]
for some polynomial $P(\cdot)$ independent of $n$. Now recall that the covariance matrix $\boldsymbol{\gamma}(\rho) = \boldsymbol{\gamma}(\rho_G)$ is positive definite with minimum eigenvalue $\nu_{\min}>1$. Therefore, $\chi_{\rho_G}(z) \leq  \exp\big(-\nu_{\min} \|z\|^2/2\|\big)$ and
\[
	 \int  \dd^{2m} z P(\|z\|) \chi_{\rho_G}(z)^2 \exp\Big(\frac{\nu_{\min}}{2} \|  \hat{z}\|^2 \Big)  < +\infty.
\]
Thus, letting $\delta=\nu_{\min}$, we find that
\[
\Big(\frac{6}{\pi^2}\Big)^{m}  \| \rho_G - \tilde \sigma_{n,G} \|_{1}^{2}  \leq \int \big|\chi_{A^{\dagger} T_n A}(z)\big|^{2} \dd^{2m} z = \mathcal{O}\Big(\frac{1}{ n}\Big),
\]
as desired.

\bigskip
\begin{claim}\label{claim:II} 
$\big\| \tilde \sigma_{n,G} - D_{\sqrt n z_n}^\dagger \tilde \sigma_{n,G} D_{\sqrt n z_n} \big\|_1 = \mathcal{O}\big(\frac{1}{\sqrt n}\big)$.
\end{claim}

\medskip
By Claim~\ref{claim:I} and Pinsker's inequality, it is sufficient to show that
\begin{equation}\label{RelativeEntropyGaussian}
	D\big(\rho_G \big\| D_{\sqrt n z_n}^\dagger \rho_G D_{\sqrt n z_n}\big) = \mathcal{O}\Big(\frac{1}{n}\Big).
\end{equation}
Since $\rho$ is in the Williamson form,  the Gaussian state $\rho_G$ is thermal and is equal to $\rho_G = \tau_1\otimes\cdots \otimes \tau_n$ where 
\[
	\tau_j = \big( 1 - e^{-\beta_j}) \exp\big(-\beta_j \ac_j^\dagger \ac_j\big),
\]
with $\beta_j >0$ being given by
\[
	\nu_j = \frac{1 + e^{-\beta_j}}{1 - e^{-\beta_j}} \geq 1.
\]
Thus, using~\eqref{eq:Dz-displace} we have 
\begin{equation}\label{GaussianHamiltonian}
D_{\sqrt n z_n}^\dagger \rho_G D_{\sqrt n z_n} = \bigotimes_{j=1}^m (1-e^{\beta_j}) \exp\big( -\beta_j (\bfa_j^\dagger +\sqrt n \bar z_{n,j})(\bfa_j + \sqrt n z_{n, j})\big),
\end{equation}
Therefore, 
\begin{align*}
 D\big(\rho_G\big \| D_{\sqrt n z_n}^\dagger \rho_G D_{\sqrt n z_n}\big) &=  \tr \Big( \rho_G \big( \ln (\rho_G) - \ln (D_{\sqrt n z_n}^\dagger \rho_G D_{\sqrt n z_n})\big)\Big)  \\
&=  \sum_{j=1}^m \beta_j \tr \Big( \rho_G \big( (\bfa_j^\dagger +\sqrt n \bar z_{n,j})(\bfa_j + \sqrt n z_{n, j}) -  \bfa_j^\dagger \bfa_j\big)\Big) \\
&=   n \sum_{j=1}^m \beta_j \|z_{n, j}\|^2 \\
&= \mathcal{O}\Big(\frac 1n\Big),
\end{align*}
where in the third line, we use the fact that $\rho_G$ is centered, and in the last line, we use Lemma~\ref{bound-FirstNorm}.

\bigskip
\begin{claim}
\label{claim:III} 
$ \big\|  \tilde{\sigma}_n^{\boxplus n} - \tilde \sigma_{n,G} \big\|_1 = \mathcal{O}\big(\frac{1}{\sqrt{n}}\big)$.
\end{claim}

\medskip
Define $T_n = \tilde{\sigma}_n^{\boxplus n} - \tilde \sigma_{n,G}$. Let 
$$\epsilon_n = \frac 12 C M_{2m}(\tilde{\sigma}_n)^{-1/(2m-2)},$$ 
where $C=C(\nu_{\min}, \nu_{\max})>0$ is the constant in Lemma~\ref{BRMethod} associated with the state $\rho$. Let $\epsilon\geq \epsilon_n$ be a constant to be determined. Then, by using Lemma~\ref{lem:norm1-norm2} for $A = \ac_1 \cdots \ac_m$, we can write
\begin{align}
	 \Big(\frac{6}{\pi^2}\Big)^{m}  \Big\| \tilde{\sigma}_n^{\boxplus n} - & \tilde \sigma_{n,G}\Big\|_{1}^{2}  \leq \int \big|\chi_{A^{\dagger} T_n A}(z)\big|^{2} \dd^{2m} z \nonumber \\
	 & =   \int_{\|z\| \leq \epsilon_n \sqrt{n} } \big|\chi_{A^{\dagger} T_n A}(z)\big|^{2} \dd^{2m} z + \int_{\|z\| > \epsilon_n \sqrt{n} } \big|\chi_{A^{\dagger} T_n A}(z)\big|^{2} \dd^{2m} z \nonumber \\
	 & \leq \int_{\|z\| \leq \epsilon_n \sqrt{n} } \big|\chi_{A^{\dagger} T_n A}(z)\big|^{2} \dd^{2m} z  + 2 \int_{\|z\| > \epsilon_n \sqrt{n} } \big|\chi_{A^{\dagger} \tilde{\sigma}_n^{\boxplus n} A}(z)\big|^{2} \dd^{2m} z  \nonumber\\
	 & \quad + 2 \int_{\|z\| > \epsilon_n \sqrt{n} } \big|\chi_{A^{\dagger} \tilde \sigma_{n,G} A}(z)\big|^{2} \dd^{2m} z \nonumber \\
	 & = \int_{\|z\| \leq \epsilon_n \sqrt{n} } \big|\chi_{A^{\dagger} T_n A}(z)\big|^{2} \dd^{2m} z + 2 \int_{  \|z\| > \epsilon \sqrt{n}  } \big|\chi_{A^{\dagger} \tilde{\sigma}_n^{\boxplus n} A}(z)\big|^{2} \dd^{2m} z \nonumber\\
	  & \quad + 2 \int_{ \epsilon \sqrt{n} \geq \|z\| > \epsilon_n \sqrt{n} } \big|\chi_{A^{\dagger} \tilde{\sigma}_n^{\boxplus n} A}(z)\big|^{2} \dd^{2m} z + 2 \int_{\|z\| > \epsilon_n \sqrt{n} } \big|\chi_{A^{\dagger} \tilde \sigma_{n,G} A}(z)\big|^{2} \dd^{2m} z.  \label{splitInt}
\end{align}
In the following, we will show that the first term in~\eqref{splitInt} is of order $\mathcal{O}(\frac{1}{n})$ while the other terms are either exponentially or sub-exponentially small.

\medskip
\noindent
\underline{\emph{First term in~\eqref{splitInt}:}} To handle the first term, we start by bounding $\big|\chi_{A^{\dagger} T_n A}(z)\big|^{2}$ in terms of the derivatives of the characteristic function of $T_n$. As a consequence of~\eqref{ChiDerivativeGaussian}, we have
\[
\big|\chi_{A^{\dagger} T_n A} (z)\big|^{2} \leq  C' \sum_{\alpha, \alpha', \beta, \beta'} \prod_{j=1}^{m} {|z_{j}|}^{(\alpha'_j+\beta'_j)} \big|\partial_z^\alpha\partial_{\bar z}^\beta \chi_{T_n}(z) \big|^{2}.
\]
To bound the derivatives on the right hand side, we apply Lemma~\ref{BRMethod}.  Using this lemma for $\kappa = 2m$, for any $z$ with $\|z\|\leq C(\nu_{\min}(\tilde{\sigma}_n), \nu_{\max}(\tilde{\sigma}_n)) M_{2m}(\tilde{\sigma}_n)^{-1/(2m-2)} \sqrt n$, we can write
\begin{align}
    \big|\partial_z^\alpha\partial_{\bar z}^\beta \chi_{T_n}(z) \big| &\leq \sum_{r=1}^{2m-2}  n^{-\frac r2}\Big|\partial_{z}^\alpha  \partial_{\bar z}^\beta \Big( \chi_{\tilde \sigma_{n, G}}(z)  E_{\tilde \sigma_n, r} (z) \Big)\Big| \nonumber\\
    &\qquad + n^{-\frac{2m-2}{2}} M_{n, 2m} L_{\nu_{\min}(\tilde{\sigma}_n), \nu_{\max}(\tilde{\sigma}_n)}(\| z\|) e^{- \frac{\nu_{\min}(\tilde{\sigma}_n)-1}{4} \|z \|^2} \nonumber\\
    &\leq \sum_{r=1}^{2m-2}  n^{-\frac r2}\Big|\partial_{z}^\alpha  \partial_{\bar z}^\beta \Big( \chi_{\tilde \sigma_{n, G}}(z)  E_{\tilde \sigma_n, r} (z) \Big)\Big| + 2n^{-\frac{2m-2}{2}} M_{n, 2m} L_{\nu_{\min}, \nu_{\max}}(\| z\|) e^{- \frac{\nu_{\min}-1}{8} \|z \|^2}, \label{BindEdgeWorth}
\end{align}
where $M_{n, 2m}= M_{2m}\big(\tilde{\sigma}_n\big)$.
Here, to write the second inequality, we use Lemma~\ref{OperatorNormCovarianceMatrix} to conclude that for sufficiently large $n$, $\nu_{\min}(\tilde{\sigma}_n)$ and $\nu_{\max}(\tilde{\sigma}_n)$ are close to $\nu_{\min}$ and $\nu_{\max}$, respectively. Consequently, the polynomial associated with $\tilde \sigma_n$ is bounded by $2L_{\nu_{\min}, \nu_{\max}}(\|z\|)$. Moreover, we have $2 (\nu_{\min}(\tilde{\sigma}_n)-1 ) \geq (\nu_{\min}-1)$. By a similar argument, the constant $C$ associated with $\tilde \sigma_n$ is at least $C/2$ for sufficiently large $n$. As a result, inequality~\eqref{BindEdgeWorth} holds for $\|z\| \leq \epsilon_n \sqrt{n}$.

Next, using Lemma~\ref{EdgeWorth-BindDerivatives} and Lemma~\ref{OperatorNormCovarianceMatrix}, we find that for sufficiently large $n$, we have
\begin{equation}\label{Aux2-bindEdge}
	\Big|\partial_{z}^\alpha  \partial_{\bar z}^\beta \Big( \chi_{\tilde \sigma_{n, G}}(z)  E_{\tilde \sigma_n, r} (z) \Big)\Big| \leq    M_{n, r+2}  Q(\|z\|) \chi_{\tilde \sigma_{n, G}}(z),
\end{equation}
where, as before, $M_{n, r+2} = M_{r+2}\big(\tilde{\sigma}_n\big)$ and $ Q(\cdot)$ is some polynomial whose coefficients depend only on the covariance matrix of $\rho$. By Lemma~\ref{lem:BindMoments-Truncation}, for $s =3$, we can write 
\begin{align}\label{BindMoments-Truncation}
M_{n, r} = M_r(\tilde \sigma_n)= \mathcal{O}\big(n^{(r-3)/2}\big),
\end{align}
for any $r \geq 3$. Using this inequality in~\eqref{Aux2-bindEdge}, there exists a constant $C_1$ such that for any $1 \leq r \leq 2m-2$, we have
\[
	n^{-\frac{2m-2}{2}} M_{n, 2m} \leq C_1 n^{-\frac 12},
\]
and
\begin{equation}
	n^{-\frac r2}\Big|\partial_{z}^\alpha  \partial_{\bar z}^\beta \Big( \chi_{\tilde \sigma_{n, G}}(z) (  E_{\tilde \sigma_n,  r} (z)) \Big)\Big| \leq C_1 n^{-\frac 12}  Q(\|z\|)\chi_{\tilde \sigma_{n, G}}.
\end{equation}
Now, invoking all these bounds in~\eqref{BindEdgeWorth}, there exists a polynomial $L_1(\cdot)$, such that for sufficiently large $n$ and any $\|z\|\leq \epsilon_n \sqrt n$, we have,
\[
	\Big|\partial_z^\alpha\partial_{\bar z}^\beta \chi_{T_n}(z) \Big| \leq n^{-\frac 12} L_1(\|z\|)  \chi_{\tilde \sigma_{n, G}} + n^{-\frac 12} L(\| z\|) e^{- \frac{\nu_{\min}-1}{8} \|z \|^2}.
\]
Using this bound in~\eqref{ChiDerivativeGaussian}, taking integration and applying~\eqref{eq:assumption-nu-min-1}, we conclude that
\[
	\int_{\|z\| \leq \epsilon_n \sqrt{n} } \big|\chi_{A^{\dagger} T_n A}(z)\big|^{2} \dd^{2m} z = \mathcal{O}\Big(\frac 1n\Big),
\]
as desired.

\medskip
\noindent
\underline{\emph{Second term in~\eqref{splitInt}:}} For the second term in~\eqref{splitInt}, we use~\cite[Proposition 3]{beigi2023towards} to write
\begin{align}\label{eq:tr-A-rho-A-D} 
\chi_{A^{\dagger}  \tilde{\sigma}_n^{\boxplus n} A}(z)=\tr\Big( A^{\dagger}  \tilde{\sigma}_n^{\boxplus n} A D_z\Big) = \tr\Big(\widetilde A^\dagger \tilde{\sigma}_n^{\otimes n} \widetilde A D_{\frac {z}{\sqrt n}}^{\otimes n}\Big),
\end{align}
where 
$$\widetilde A = \prod_{j=1}^m \frac{\ac_{j,1}+ \cdots +\ac_{j, n} }{\sqrt n}.$$ 
Expanding $\widetilde A$, we find that $\chi_{A^{\dagger}  \tilde{\sigma}_n^{\boxplus n} A}(z)$ can be written in the form
\begin{align}\label{eq:chi-A-z-F-z}
\tr\Big( A^{\dagger}  \tilde{\sigma}_n^{\boxplus n} A D_z\Big) =  \frac{1}{n^m} \tr\Big(\tilde{\sigma}_n D_{\frac{z}{\sqrt n}}\Big)^{n-2m}F(z)= \frac{1}{n^m} \chi_{\tilde{\sigma}_n}\big(z/\sqrt n\big)^{n-2m} F(z),
\end{align}
with
\begin{align}\label{eq:def-F(z)-S}
F(z) = \sum_{k, \ell}  f_{k, \ell}  \, \tr\Big(  S_{k, \ell}  \,   D_{\frac{z}{\sqrt n}}^{\otimes 2m} \Big),\qquad \quad S_{k, \ell} := \bigg(\prod_{j=1}^m \ac_{j, k_j}^\dagger\bigg) \tilde{\sigma}_n^{\otimes 2m} \bigg(\prod_{j=1}^m \ac_{j, \ell_j}\bigg),
\end{align}
where the sum runs over all tuples $k=(k_1, \dots, k_m)$ and $\ell = (\ell_1, \dots, \ell_m)$ satisfying $1 \leq k_j,\ell_{j'} \leq 2m$, and the coefficients $f_{k,\ell}$ are at most $|f_{k, \ell}| =\mathcal O(n^{2m})$. 
By Proposition~\ref{prop:CharSupOrigin}, there exists $\delta_n>0$ such that $\big|\chi_{\sigma_n} \big(z/\sqrt n\big)\big|\leq 1-\delta_n$ for all $\|z\|\geq \epsilon\sqrt n$. Note that $\delta_n$ is determined in terms of the second moments of $\tilde{\sigma}_n$ (see ~\cite[Proposition 15]{CRLimitTheorem} for the exact form of $\delta_n$), and by Lemma~\ref{OperatorNormCovarianceMatrix} for sufficiently large $n$ the second moments of $\tilde \sigma_n$ can be bounded in terms of the second moments of $\rho$.  Hence, there exists a $\delta >0$ such that for sufficiently large $n$, we have $\delta_n \geq \delta$. Thus, for sufficiently large $n$, we have
\begin{align}\label{eq:z-eps-sqrt-n-A-tilde-sigma}
 \int_{\|z\| > \epsilon \sqrt{n} } &\big| \chi_{A^\dagger \tilde{\sigma}_n^{\boxplus n} A} (z)\big|^{2} \dd^{2m} z  \leq  \frac{1}{n^{2m}} (1-\delta)^{2(n-2m)}  \int_{\|z\| > \epsilon \sqrt{n} }   \big| F(z)\big|^2 \dd^{2m} z.
\end{align} 
In the following we show that the integral on the right hand side grows at most polynomially with $n$. Then, considering the exponentially small factor $(1-\delta)^{2(n-2m)}$, we find that the left hand side is exponentially small as desired.
To bound the integral on the right hand side, we apply the Cauchy--Schwarz inequality to write
\begin{align*}
\int_{\|z\| > \epsilon \sqrt{n} }   \big| F(z)\big|^2 \dd^{2m} z & \leq  (2m)^{2m} \sum_{k, \ell}  f_{k, \ell}^2  \, \int_{\|z\| > \epsilon \sqrt{n} }  \bigg| \tr\bigg(S_{k, \ell}\, D_{\frac{z}{\sqrt n}}^{\otimes 2m} \bigg)\bigg|^2\dd^{2m} z\\
& \leq  (2m)^{2m} \sum_{k, \ell}  f_{k, \ell}^2  \, \int  \big|\chi_{S_{k, \ell}}(z)\big|^2\dd^{2m} z\\
& =  \pi^m(2m)^{2m} \sum_{k, \ell}  f_{k, \ell}^2 \big\|S_{k, \ell}\big\|_2^2,
\end{align*}
where the last line follows from the Plancherel identity~\eqref{eq:Plancherel}.
Next, using Lemma~\ref{lem:tr-A-kappa-bounded}, we observe that
\begin{align}
 \big\|S_{k, \ell}\big\|_2^2 & =\tr\bigg[  \bigg(\prod_{j=1}^m \ac_{j, k_j} \ac_{j, k_j}^\dagger\bigg) \tilde{\sigma}_n^{\otimes 2m} \bigg(\prod_{j=1}^m \ac_{j, \ell_j}  \ac_{j, \ell_j}^\dagger \bigg)   \tilde{\sigma}_n^{\otimes 2m}   \bigg] \nonumber\\
 & \leq \frac 12 \tr\bigg[  \bigg(\prod_{j=1}^m \ac_{j, k_j} \ac_{j, k_j}^\dagger\bigg) \tilde{\sigma}_n^{\otimes 2m} \bigg(\prod_{j=1}^m \ac_{j, k_j} \ac_{j, k_j}^\dagger\bigg)   \tilde{\sigma}_n^{\otimes 2m}   \bigg] \nonumber\\
 &\quad + \frac 12 \tr\bigg[  \bigg(\prod_{j=1}^m \ac_{j, \ell_j}  \ac_{j, \ell_j}^\dagger \bigg)  \tilde{\sigma}_n^{\otimes 2m} \bigg(\prod_{j=1}^m \ac_{j, \ell_j}  \ac_{j, \ell_j}^\dagger \bigg)   \tilde{\sigma}_n^{\otimes 2m}   \bigg] \nonumber\\
 & \leq \frac 12 \tr\bigg[  \bigg(\prod_{j=1}^m \ac_{j, k_j} \ac_{j, k_j}^\dagger\bigg) \tilde{\sigma}_n^{\otimes 2m}     \bigg]^2  + \frac 12 \tr\bigg[  \bigg(\prod_{j=1}^m \ac_{j, \ell_j}  \ac_{j, \ell_j}^\dagger \bigg)  \tilde{\sigma}_n^{\otimes 2m}     \bigg]^2 \nonumber \\
 & = \mathcal O\Big(M_{2m}\big(\tilde \sigma_n^{\otimes 2m}\big)^2\Big).\label{eq:norm-S-k-ell}
\end{align}
Note that by definition it holds that $M_{4m}\big(\tilde \sigma_n^{\otimes 2m}\big) = \mathcal O\big(M_{4m}(\tilde \sigma_n)\big)$. Also, by~\eqref{BindMoments-Truncation}, we have $M_{4m}(\tilde \sigma_n) = \mathcal O\big(n^{(4m-3)/2}\big)$. Moreover, we note that the indices $k, \ell$ are restricted to tuples $k=(k_1, \dots, k_m)$ and $\ell = (\ell_1, \dots, \ell_m)$ satisfying $1 \leq k_j,\ell_{j'} \leq 2m$, and take a constant number of values (depending only on $m$). Thus, substituting these bounds into~\eqref{eq:z-eps-sqrt-n-A-tilde-sigma}, we find that the second term in~\eqref{splitInt} is exponentially small.

\medskip
\noindent
\underline{\emph{Third term in~\eqref{splitInt}:}} 
For this term, we aim to show that 
\begin{equation}\label{ThirdTermSplitInt}
\int_{ \epsilon \sqrt{n} \geq \|z\| > \epsilon_n \sqrt{n} } \big|\chi_{A^{\dagger} \tilde{\sigma}_n^{\boxplus n} A}(z)\big|^{2} \dd^{2m} z 
\end{equation}
is sub-exponentially small. As we will see, this holds mainly because the tail of a Gaussian function is always sub-exponentially small. 

Assuming $\|z\|\leq \epsilon \sqrt n$, we use~\eqref{eq:chi-A-z-F-z} and Lemma~\ref{lem:bound-chi-tilde-sigma} to obtain
\begin{align*}
\big|\chi_{A^{\dagger} \tilde{\sigma}_n^{\boxplus n} A}(z)\big|^{2} & = \frac{1}{n^{2m}} \big|\chi_{\tilde{\sigma}_n}\big(z/\sqrt n\big)\big|^{2n-4m} |F(z)|^2\\
&\leq \frac{1}{n^{2m}}   
|F(z)|^2   e^{4\theta_0 m+2\epsilon \theta_1 \|z\|^2}\chi_{\rho_G}(z)^2.
\end{align*}
Next, recall that $F(z) = \sum_{k, \ell}  f_{k, \ell}  \, \tr\Big(  S_{k, \ell}  \,   D_{\frac{z}{\sqrt n}}^{\otimes 2m} \Big)$, where $|f_{k, \ell}| =\mathcal O(n^{2m})$. Moreover, by the Cauchy--Schwarz inequality and Lemma~\ref{lem:tr-A-kappa-bounded}, we have
\begin{align}
\tr\Big(  S_{k, \ell}  \,   D_{\frac{z}{\sqrt n}}^{\otimes 2m} \Big) & \leq   \big\| S_{k, \ell}\big\|_1 \nonumber \\
&\leq   \tr \bigg(\tilde{\sigma}_n^{\otimes 2m}\prod_{j=1}^m \ac_{j, k_j} \ac_{j, k_j}^\dagger\bigg)^{1/2}\cdot  \tr \bigg(\tilde{\sigma}_n^{\otimes 2m}  \prod_{j=1}^m \ac_{j, \ell_j}\ac_{j, \ell_j}^\dagger\bigg)^{1/2} \nonumber \\
&  = \mathcal O\big(M_{2m}(\tilde \sigma_n^{\otimes 2m})\big) \nonumber\\
&= \mathcal O\big(M_{2m}(\tilde \sigma_n)\big).   \label{eq:F(z)-bound}
\end{align}
Therefore, applying~\eqref{BindMoments-Truncation} we find that  $\frac{1}{n^{2m}}  e^{4\theta_0 m} |F(z)|^2 $ is at most a polynomial of $n$. We then obtain
\begin{align*}
\int_{ \epsilon \sqrt{n} \geq \|z\| > \epsilon_n \sqrt{n} } \big|\chi_{A^{\dagger} \tilde{\sigma}_n^{\boxplus n} A}(z)\big|^{2} \dd^{2m} z  &\leq   \poly(n)\cdot \int_{ \epsilon \sqrt{n} \geq \|z\| > \epsilon_n \sqrt{n} } e^{4\epsilon \theta_1\|z\|^2}  \chi_{\rho_G}(z)^2 \dd^{2m} z  \\ 
&\leq  \poly(n)\cdot \int_{ \epsilon \sqrt{n} \geq \|z\| > \epsilon_n \sqrt{n} } e^{4\epsilon \theta_1\|z\|^2}  \chi_{\rho_G}(z)^2 \dd^{2m} z\\
&\leq \poly(n)\cdot\int_{ \epsilon \sqrt{n} \geq \|z\| > \epsilon_n \sqrt{n} }  e^{-(\nu_{\min} -4\epsilon \theta_1) \|z\|^2} \dd^{2m} z\\
&\leq \poly(n)\cdot \int_{  \|z\| > \epsilon_n \sqrt{n} }  e^{-(\nu_{\min} -4\epsilon \theta_1) \|z\|^2} \dd^{2m} z.
\end{align*}
Now suppose that $\epsilon>0$ is chosen in such a way that $\xi:=\nu_{\min} -4\epsilon \theta_1>0$. Then, we continue 
\begin{align}\label{eq:Gaussian-tail-subexp}
\int_{ \epsilon \sqrt{n} \geq \|z\| > \epsilon_n \sqrt{n} } \big|\chi_{A^{\dagger} \tilde{\sigma}_n^{\boxplus n} A}(z)\big|^{2} \dd^{2m} z &\leq\poly(n)\cdot \int_{  \|z\| > \epsilon_n \sqrt{n} }  e^{-\xi \|z\|^2} \dd^{2m} z\nonumber\\
& \leq \poly(n)\cdot e^{-\frac{1}{2}\xi \epsilon_n^2 n }  \int_{  \|z\| > \epsilon_n \sqrt{n} }  e^{- \frac 12\xi \|z\|^2} \dd^{2m} z\nonumber\\
& \leq  \poly(n)\cdot e^{-\frac{1}{2}\xi \epsilon_n^2 n }  \int e^{- \frac 12\xi \|z\|^2} \dd^{2m} z\nonumber\\
& =\mathcal O\Big(\poly(n)\cdot e^{-\frac{1}{2}\xi \epsilon_n^2 n }\Big). 
\end{align}
Recall that $\epsilon_n = \frac 12C M_{2m}(\tilde{\sigma}_n)^{-1/(2m-2)}$, so according to~\eqref{BindMoments-Truncation},
\begin{align}\label{boundepsilonN}
n \epsilon_n^2 \geq c n^{1  -\frac{2m-3}{2(m-1)}} = c n^{\frac{1}{2(m-1)}}.
\end{align}
Invoking this in the above bound, we find that the third term is indeed sub-exponentially small. 

\medskip
\noindent
\underline{\emph{Fourth term in~\eqref{splitInt}:}} For the fourth term, we aim to show that
\[
	\int_{\|z\| > \epsilon_n \sqrt{n} } \big|\chi_{A^{\dagger} \tilde \sigma_{n,G} A}(z)\big|^{2} \dd^{2m} z
\]
is sub-exponentially small.  The proof of this fact is similar to, and in fact easier than, that of the third term. First, similarly to~\eqref{eq:chi-A-z-F-z} we can write
\begin{align*}
\tr\Big( A^{\dagger}  \tilde{\sigma}_{n, G}^{\boxplus n} A D_z\Big) = \frac{1}{n^m} \chi_{\tilde{\sigma}_{n, G}}\big(z/\sqrt n\big)^{n-2m} F_G(z) = \frac{1}{n^m} \chi_{\tilde{\sigma}_{n, G}}\big(z/\sqrt n\big)^{n-2m} F_G(z),
\end{align*}
where $F_G(z)$ is a function that, similar to the analysis above, can be shown to satisfy
$$|F_G(z)|=\poly(n).$$
Also, in view of Lemma~\ref{OperatorNormCovarianceMatrix},  for sufficiently large $n$ we have
\begin{align*}
\chi_{\tilde{\sigma}_{n, G}}\big(z/\sqrt n\big)^{n-2m} &  = \chi_{\tilde\sigma_{n, G}} (z) \cdot \chi_{\tilde{\sigma}_{n, G}}\big(z/\sqrt n\big)^{-2m}\\
& \leq e^{-\frac 14 \nu_{\min}\|z\|^2} e^{\frac{m}{n} \nu_{\max}\|z|^2}\\
& \leq e^{-\frac 18 \nu_{\min}\|z\|^2},
\end{align*}
where for the last line we assume that $n\geq \frac{4m\nu_{\max}}{\nu_{\min}}$. Putting these together, for sufficiently large $n$ we have
\begin{align*}
\int_{\|z\| > \epsilon_n \sqrt{n} } \big|\chi_{A^{\dagger} \tilde \sigma_{n,G} A}(z)\big|^{2} \dd^{2m} z&\leq \poly(n)\cdot \int_{\|z\| > \epsilon_n \sqrt{n} } e^{-\frac 18 \nu_{\min}\|z\|^2},
\end{align*}
which, applying the same trick as in~\eqref{eq:Gaussian-tail-subexp} and using~\eqref{boundepsilonN}, can be shown to be sub-exponentially small.

%*******************************************************************
\section{Proof of Theorem~\ref{MainTheorem}}\label{proofMain}

First of all, by the same argument as in the proof of Theorem~\ref{MainTheoremTrace} we assume that $\rho$ is in Williamson's form and 
\begin{align}\label{eq:assumption-nu-min-2}
\nu_{\min}= \min_{j} \nu_j>1,
\end{align}
which equivalently means that $\boldsymbol{\gamma}(\rho) + i\Omega_m$ is full-rank. We also assume that $m>1$ and will later comment on the case of $m=1$. 

Let $\sigma_n$, $\tilde{\sigma}_n$, and $\tilde \sigma_{n,G}$ be the same operators as defined in Section~\ref{Sec:State-Truncation}. Applying triangle's inequality, we have
\begin{align*}
	\big\| \rho^{\boxplus n} - \tilde \sigma_n^{\boxplus n}\big\|_1 
	& \leq  \big\| \rho^{\boxplus n} -  \sigma_n^{\boxplus n}\big\|_1 +\big\| \sigma_n^{\boxplus n} - \sigma_{n, G}\big\|_1 +\big\| \sigma_{n, G} - \tilde \sigma_{n, G}\big\|_1 +  \big\| \tilde \sigma_n^{\boxplus n} - \tilde \sigma_{n, G}\big\|_1\\
	&= \big\| \rho^{\boxplus n} -  \sigma_n^{\boxplus n}\big\|_1  +\big\| \sigma_{n, G} - \tilde \sigma_{n, G}\big\|_1 +  2\big\| \tilde \sigma_n^{\boxplus n} - \tilde \sigma_{n, G}\big\|_1.
\end{align*}
Then, using Lemma~\ref{Aux1} with $s=4+\delta$, we find that the first term is bounded by $\mathcal{O}\big(n^{-(2+\delta)/2}\big)$. Also, applying similar computations as in Claim~\ref{claim:II} and Claim~\ref{claim:III} in Section~\ref{Sec:Proof-Handle}, but with the stronger condition of finiteness of moments of order $4+\delta$, we find that the other two terms are also bounded by $\mathcal{O}\big(n^{-(2+\delta)/2}\big)$. Therefore, 
\begin{equation}\label{eq:trunc_Conv_Ent}
	\big\| \rho^{\boxplus n} - \tilde \sigma_n^{\boxplus n}\big\|_1=\mathcal{O}\Big(\frac{1}{n^{(2+\delta)/2}}\Big).
\end{equation}
Moreover, using Lemma~\ref{OperatorNormCovarianceMatrix} with $s=4+\delta$, we have
\begin{equation}\label{eq:Stronger_Gauss_Ent}
	\big\| \boldsymbol{\gamma}(\rho) - \boldsymbol{\gamma}(\tilde \sigma_{n})\big\| =\big\| \boldsymbol{\gamma}(\rho_G) - \boldsymbol{\gamma}(\tilde \sigma_{n,G})\big\| = \mathcal{O}\Big(\frac{1}{n^{(2+\delta)/2}}\Big).
\end{equation}
Then, writing our argument of Claim~\ref{claim:I} in Section~\ref{Sec:Proof-Handle}, and using~\eqref{eq:Stronger_Gauss_Ent} in~\eqref{diverg_point_Ent1} and~\eqref{diverg_point_Ent2}, it is readily verified that
\begin{equation}\label{eq:trunc_Gauss_Ent}
\big\| \rho_G - \tilde \sigma_{n,G} \big\|_{1}  =\mathcal{O}\Big(\frac{1}{n^{(2+\delta)/2}}\Big).
\end{equation}
We use~\eqref{eq:trunc_Conv_Ent} and~\eqref{eq:trunc_Gauss_Ent} to approximate the quantum relative entropy $D\big(\rho^{\boxplus n}\big\| \rho_G\big)$ with $D\big(\tilde{\sigma}_n^{\boxplus n}\big\| \tilde \sigma_{n,G} \big)$. 
Note that $\rho^{\boxplus n}$ is a centered state whose Gaussification equals $\rho_G$. Then, $\tr\big(\rho^{\boxplus n} \ln \rho_G\big)$ depends only on the covariance matrix of $\rho^{\boxplus n}$, which is equal to that of $\rho_G$. As a result, we have 
\begin{align}
	D\big(\rho^{\boxplus n}\big\| \rho_G\big) = S(\rho_G) - S(\rho^{\boxplus n}), \quad\qquad D\big(\tilde{\sigma}_n^{\boxplus n}\big\| \tilde \sigma_{n,G}\big) = S(\tilde \sigma_{n,G}) - S(\tilde{\sigma}_n^{\boxplus n}),
\end{align}
where the second equality follows by the same argument. 
Therefore, the uniform continuity bound for the von Neumann entropy~\cite[Lemma 18]{winter_tight_2016} along with ~\eqref{eq:trunc_Conv_Ent} and~\eqref{eq:trunc_Gauss_Ent} yield
\[
	\Big|D\big(\rho^{\boxplus n}\big\| \rho_G\big) -  D\big(\tilde{\sigma}_n^{\boxplus n}\big\| \tilde \sigma_{n,G}\big)\Big|\leq \Big|S\big(\rho^{\boxplus n}\big) -  S\big(\tilde{\sigma}_n^{\boxplus n}\big)\Big| +  \Big| S(\rho_G) - S (\tilde \sigma_{n,G})\Big| = \mathcal O\Big( \frac{\ln n}{n^{(2+\delta)/2}}\Big).
	\]
Thus, to prove the desired convergence rate~\eqref{AsympBoundEntropy}, it is sufficient to show that
\begin{equation}\label{eq:Ent_Rate_Trunc}
		D\big(\tilde{\sigma}_n^{\boxplus n}\big\| \tilde \sigma_{n, G} \big) = \mathcal{O}\Big(\frac{1}{n}\Big) \quad \text{as} \quad n\rightarrow\infty.
\end{equation}

\begin{figure}[t!]
\label{fig:thm2}
\begin{center}
\begin{tikzpicture}
  [transition/.style={rectangle, rounded corners,draw=black!50,fill=black!20,
                      inner sep=3pt,minimum size=5mm, align=center}, point/.style={circle,inner sep=0pt,minimum size=0pt,fill=black}]
  \node[transition] (theorem 2)  {\small Theorem \ref{MainTheorem}};
   \node[transition] (continuity)  [below=4mm of theorem 1]  {\small Continuity of the von Neumann entropy \cite[Lemma 18]{winter_tight_2016}};
  \node[transition] (Dtilde) [below=8mm of continuity] {\small $D\big(\tilde{\sigma}_n^{\boxplus n}\big\| \tilde \sigma_{n, G} \big) = \mathcal{O}\Big(\frac{1}{n}\Big)$};
   \node[transition] (proof Dtilde) [below=8mm of Dtilde.center] {\small Theorem \ref{boundEntropy-Char} reduces \\  \small the problem to \eqref{eq:tilde-sigma-n-tau-1/n}};
  \node[transition] (proof norm11) [left=of proof Dtilde] {\small Similar step as in the\\ \small proof of Theorem \ref{MainTheoremTrace} };
    \node[transition] (proof norm12) [right=of proof Dtilde] {\small Lemma \ref{OperatorNormCovarianceMatrix} \& similar steps \\ \small as in the proof of Claim \ref{claim:I} };

  \node[transition] (norm12) [above=8mm of proof norm12.center] {\small $\big\| \rho_G - \tilde \sigma_{n,G} \big\|_{1}  =\mathcal{O}\Big(\frac{1}{n^{(2+\delta)/2}}\Big)$};
  \node[transition] (norm11) [above=8mm of proof norm11.center] {\small $\big\| \rho^{\boxplus n} - \tilde \sigma_n^{\boxplus n}\big\|_1=\mathcal{O}\Big(\frac{1}{n^{(2+\delta)/2}}\Big)$};
  
  \node[point] (dot L1)[above=4mm of norm12]  {};
  \node[point] (dot R1)[above=4mm of norm11]  {};

  \path [line width=0.5pt] (theorem 2) edge[->] (continuity);
  \path [line width=0.5pt] (continuity) edge[->] (Dtilde);
  \draw [line width=0.5pt] (dot R1) -- (dot L1);
  \path [line width=0.5pt] (dot R1)   edge[->] (norm11);
  \path [line width=0.5pt] (dot L1)   edge[->] (norm12);
  \path [line width=0.5pt] (norm11)   edge[->] (proof norm11);
  \path [line width=0.5pt] (Dtilde)   edge[->] (proof Dtilde);
  \path [line width=0.5pt] (norm12)   edge[->] (proof norm12);
  
  \node[transition] (split) [below=6mm of proof Dtilde] {\small Split the integral into regions \\ \small as in the proof of Claim \ref{claim:III} };
  \node[transition] (proof split) [left=4mm of split] {\small Lemma \ref{BRMethod}, Lemma \ref{EdgeWorth-BindDerivatives}, \\ \small Lemma \ref{OperatorNormCovarianceMatrix}, Lemma \ref{lem:BindMoments-Truncation}};
  
  \path [line width=0.5pt] (proof Dtilde)   edge[->] (split);  
  \path [line width=0.5pt] (split)   edge[->] (proof split);

\end{tikzpicture}
\end{center}
\caption{To prove Theorem \ref{MainTheorem}, using the continuity of the entropy function we reduce the problem to proving $D\big(\tilde{\sigma}_n^{\boxplus n}\big\| \tilde \sigma_{n, G} \big) = \mathcal{O}\big({n}^{-1}\big)$. To apply this continuity bound we need $\big\| \rho^{\boxplus n} - \tilde \sigma_n^{\boxplus n}\big\|_1=\mathcal{O}\big(n^{-(2+\delta)/2}\big)$ and $\big\| \rho_G - \tilde \sigma_{n,G} \big\|_{1}  =\mathcal{O}\big(n^{-(2+\delta)/2}\big)$. Next, we use Theorem \ref{boundEntropy-Char} to reduce the bound on $D\big(\tilde{\sigma}_n^{\boxplus n}\big\| \tilde \sigma_{n, G} \big)$ to~\eqref{eq:tilde-sigma-n-tau-1/n}. Proof of~\eqref{eq:tilde-sigma-n-tau-1/n} is similar to that of Claim~\ref{claim:III} and is based on splitting the resulting integral.}
\end{figure}

We employ Theorem~\ref{boundEntropy-Char} to prove the above inequality. Notably, since we have assumed that $\nu_j>1$, it follows from~\eqref{eq:trunc_Gauss_Ent} that the covariance matrix $\boldsymbol{\gamma}(\tilde \sigma_{n, G}) + i \Omega_m$ is invertible for sufficiently large $n$. Consequently, $\tilde \sigma_{n, G}$ as a centered Gaussian state satisfies the hypothesis of Theorem~\ref{boundEntropy-Char}.

Let $E_n(z) = E_{\tilde{\sigma}_n, 1}(z)$ be the polynomial appearing in the Edgeworth-type expansion~\eqref{eq:expansion-ET} for the state $\tilde{\sigma}_n$, and recall that all terms in $E_n(z)$ have degree $3$. Also, since $\tilde{\sigma}_n$ is self-adjoint, $E_n(-z) = \overline{E_n(z)}$. Thus, the polynomial $E_n(z)$ satisfies the hypothesis of Theorem~\ref{boundEntropy-Char}. Therefore, letting
$$\alpha = \frac{1}{\sqrt{n}},$$ 
in view of Theorem~\ref{boundEntropy-Char}, it suffices to verify that
\begin{equation}\label{eq:tilde-sigma-n-tau-1/n}
	\Big\|  \big(\tilde{\sigma}_n^{\boxplus n} - \tau_{n, \alpha}\big) (N_m+m)^{(m+3)/2} \Big\|_2 = \mathcal{O}\Big(\frac{1}{n}\Big),
\end{equation}
where $\tau_n := \tilde \sigma_{n, G}$, and $\tau_{n, \alpha}$ is given by
\[
	\chi_{\tau_{n,\alpha}}(z) = \chi_{\tau_n}(z) \big( 1 + \alpha E_n (z) \big).
\]

 We put $T_n := \tilde{\sigma}_n^{\boxplus n} - \tau_{n,\alpha}$ and compute
\begin{align*}
\Big\|  \big( \tilde{\sigma}_n^{\boxplus n} - \tau_{n,\alpha} \big) (N_m+m)^{(m+3)/2} \Big\|_2^2 
& =  \tr\Big(T_n^2 \big(\bfa_1^\dagger\bfa_1+ \cdots+ \bfa^\dagger_m\bfa_m+ m\big)^{m+3}\Big)\\
& = \sum_{\mu, \nu} b_{\mu, \nu} \tr\Big(T_n^2  \big(\bfa_1^\dagger\big)^{\mu_1} \bfa_1^{\nu_1}  \cdots \big(\bfa_m^\dagger\big)^{\mu_m}\bfa_m^{\nu_m}\Big)\\
& = \sum_{\mu, \nu} |b_{\mu, \nu}| \cdot \Big|\tr\Big(T_n^2  \big(\bfa_1^\dagger\big)^{\mu_1} \bfa_1^{\nu_1}  \cdots \big(\bfa_m^\dagger\big)^{\mu_m}\bfa_m^{\nu_m}\Big)\Big|,
\end{align*} 
where the sum is over all tuples $\mu, \nu$ of non-negative integers satisfying $\sum_j (\mu_j+\nu_j)\leq 2(m+3)$, and $b_{\mu, \nu}$'s are some constants. We note that there are a constant number (depending only on $m$) of such indices $\mu, \nu$. Thus, it suffices to verify that 
$$\Big|\tr\Big(T_n^2  A_1\cdots A_{2(m+3)}\Big)\Big| = \mathcal O\Big(\frac{1}{n^2}\Big),$$
for any set of operators $A_1, \dots, A_{2(m+3)} \in \big\{ \bfa_1, \bfa_1^\dagger, \dots, \bfa_m, \bfa_m^\dagger, \mathbb I \big\}$. To this end, we apply the Cauchy--Schwarz inequality, yielding
\begin{align*}
\Big|\tr\Big(T_n^2  A_1\cdots A_{2(m+3)}\Big)\Big| &\leq \big\|T_n A_1\cdots A_{m+3}\big\|_2\cdot  \big\|T_n A_{m+4}\cdots A_{2(m+3)}\big\|_2.
\end{align*}
Thus, it suffices to demonstrate that 
$$\big\|T_n A_1\cdots A_{m+3}\big\|_2^2= \mathcal O\Big(\frac{1}{n^2}\Big),$$
for any set of operators $A_1, \dots, A_{m+3} \in \big\{ \bfa_1, \bfa_1^\dagger, \dots, \bfa_m, \bfa_m^\dagger, \mathbb I \big\}$.

Let $A := A_1\cdots A_{m+3}$ and 
$$\epsilon_n = \frac 12 C M_{m+3}(\tilde{\sigma}_n)^{-1/(m+1)},$$ 
where $C=C(\nu_{\min}, \nu_{\max})>0$ is the constant in Lemma~\ref{BRMethod} associated with the state $\rho$. Also, let $\epsilon\geq \epsilon_n$ be a constant to be determined. By using Plancherel's identity, we can write
\begin{align}
	 \big\|T_n A\big\|_2^2  &= \int \big|\chi_{T_n A}(z)\big|^{2} \dd^{2m} z \nonumber \\
	 & =   \int_{\|z\| \leq \epsilon_n \sqrt{n} } \big|\chi_{T_n A}(z)\big|^{2} \dd^{2m} z + \int_{\|z\| > \epsilon_n \sqrt{n} } \big|\chi_{T_n A}(z)\big|^{2} \dd^{2m} z \nonumber \\
	 & \leq \int_{\|z\| \leq \epsilon_n \sqrt{n} } \big|\chi_{T_n A}(z)\big|^{2} \dd^{2m} z  + 2 \int_{\|z\| > \epsilon_n \sqrt{n} } \big|\chi_{ \tilde{\sigma}_n^{\boxplus n} A}(z)\big|^{2} \dd^{2m} z  \nonumber\\
	 & \quad + 2 \int_{\|z\| > \epsilon_n \sqrt{n} } \big|\chi_{ \tau_{n,\alpha} A}(z)\big|^{2} \dd^{2m} z \nonumber \\
	 & = \int_{\|z\| \leq \epsilon_n \sqrt{n} } \big|\chi_{ T_n A}(z)\big|^{2} \dd^{2m} z + 2 \int_{  \|z\| > \epsilon \sqrt{n}  } \big|\chi_{\tilde{\sigma}_n^{\boxplus n} A}(z)\big|^{2} \dd^{2m} z \nonumber\\
	  & \quad + 2 \int_{ \epsilon \sqrt{n} \geq \|z\| > \epsilon_n \sqrt{n} } \big|\chi_{ \tilde{\sigma}_n^{\boxplus n} A}(z)\big|^{2} \dd^{2m} z + 2 \int_{\|z\| > \epsilon_n \sqrt{n} } \big|\chi_{ \tau_{n,\alpha} A}(z)\big|^{2} \dd^{2m} z.  \label{splitInt_Ent}
\end{align}
In the following, we will show that the first term in~\eqref{splitInt_Ent} is of order $\mathcal{O}(\frac{1}{n^2})$ while the other terms are either exponentially or sub-exponentially small. We employ similar ideas as those in the proof of Claim~\ref{claim:III} in the proof of Theorem~\ref{MainTheoremTrace}. Since the computations are quite similar to those in the proof of Theorem~\ref{MainTheoremTrace}, we will provide a high-level overview of the proof of these claims and omit the details.

For the first term, we can use~\eqref{eq:characterisitc-derivative-0} and~\eqref{eq:characterisitc-derivative-1} to express $\chi_{T_n A}(z)$ in terms of the derivatives of $\chi_{T_n}(z)$ up to order at most $m+3$ times powers of $z_1, \overline{z_1}, \dots, z_m, \overline{z_m}$. Next, applying Lemma~\ref{BRMethod} for $\kappa=m+3$, we find that for any such derivative we have 
\begin{equation}\label{BindEdgeWorth_Ent}
	\big|\partial_z^\alpha\partial_{\bar z}^\beta \chi_{T_n}(z) \big| \leq \sum_{r=2}^{m+1}  n^{-\frac r2}\Big|\partial_{z}^\alpha  \partial_{\bar z}^\beta \Big( \chi_{\tilde \sigma_{n, G}}(z)  E_{\tilde \sigma_n, r} (z) \Big)\Big| + 2n^{-\frac{m+1}{2}} M_{n, m+3} L_{\nu_{\min}, \nu_{\max}}(\| z\|) e^{- \frac{\nu_{\min}-1}{4} \|z \|^2},
\end{equation}
where $M_{n, r+2} = M_{r+2}\big(\tilde{\sigma}_n\big)$. Also, using Lemma~\ref{EdgeWorth-BindDerivatives} and~\eqref{eq:Stronger_Gauss_Ent}, for sufficiently large $n$ we have
\begin{equation}\label{Aux2-bindEdge_Ent}
	\Big|\partial_{z}^\alpha  \partial_{\bar z}^\beta \Big( \chi_{\tilde \sigma_{n, G}}(z)  E_{\tilde \sigma_n, r} (z) \Big)\Big| \leq    M_{n, r+2}  Q(\|z\|) \chi_{\tilde \sigma_{n, G}}(z),
\end{equation}
where $ Q(\cdot)$ is a polynomial with coefficients depending only on the covariance matrix of $\rho$. Now, using Lemma~\ref{lem:BindMoments-Truncation} with $s=4+\delta$, it is verified that there is a constant $C_1$ such that 
\[
	n^{-\frac{m+1}{2}} M_{n, m+3} \leq C_1 n^{-\frac {(2+\delta)}{2}},
\]
and for any $3 \leq r \leq m+1$,
\begin{equation}
	n^{-\frac r2}\Big|\partial_{z}^\alpha  \partial_{\bar z}^\beta \Big( \chi_{\tilde \sigma_{n, G}}(z) (  E_{\tilde \sigma_n,  r} (z)) \Big)\Big| \leq C_1 n^{-\frac {(2+\delta)}{2}}  Q(\|z\|)\chi_{\tilde \sigma_{n, G}}(z).
\end{equation}
Additionally, by the same lemma, we have \( M_{4, n} = \mathcal{O}(M_4(\rho)) \). Therefore, for \( r = 2 \), it holds that
\[
n^{-\frac{r}{2}} \left| \partial_{z}^\alpha  \partial_{\bar z}^\beta \left( \chi_{\tilde \sigma_{n, G}}(z) E_{\tilde \sigma_n, r}(z) \right) \right| \leq C_1' n^{-1} M_{4}(\rho) Q(\|z\|) \chi_{\tilde \sigma_{n, G}}(z),
\]
for some constant \( C_1' > 0 \).

Putting all these bounds in~\eqref{BindEdgeWorth_Ent}, once again there exist polynomials $L_1(\cdot)$ and $L_2(\cdot)$ such that for sufficiently large $n$ and any $\|z\|\leq \epsilon_n \sqrt n$, we have
\[
	\Big|\partial_z^\alpha\partial_{\bar z}^\beta \chi_{T_n}(z) \Big| \leq n^{-1} L_1(\|z\|)  \chi_{\tilde \sigma_{n, G}}(z) + n^{-\frac{(2+\delta)}2} L_2(\|z\|)  \chi_{\tilde \sigma_{n, G}}(z) + n^{-\frac {(2+\delta)}{2}} L(\| z\|) e^{- \frac{\nu_{\min}-1}{4} \|z \|^2}.
\]
By applying this bound for the derivatives of \( \chi_{T_n}(z) \), and integrating while utilizing \( \eqref{eq:assumption-nu-min-2} \), we find that \( \nu_{\min}(\tilde \sigma_n) = \nu_{\min}(\tilde \sigma_{n, G}) > 1 \) for sufficiently large \( n \), and we arrive at
\[
	\int_{\|z\| \leq \epsilon_n \sqrt{n} } \big|\chi_{ T_n A}(z)\big|^{2} \dd^{2m} z = \mathcal{O}\Big(\frac{1}{n^2}\Big).
\]

For the second and third terms in \( \eqref{splitInt_Ent} \), we can apply the same argument used to address the second and third terms in Claim~\ref{claim:III} within the proof of Theorem~\ref{MainTheoremTrace}, demonstrating that these terms are exponentially and sub-exponentially small, respectively.

For the fourth term in \( \eqref{splitInt_Ent} \), we first express \( \chi_{\tau_{n,\alpha} A}(z) \) in terms of the derivatives of \( \chi_{\tau_{n,\alpha}}(z) = \chi_{\tau_n}(z) \big( 1 + \alpha E_n (z) \big) \) as discussed earlier. Since \( \chi_{\tau_n}(z) = \chi_{\tilde \sigma_{n, G}}(z) \) is Gaussian, it follows that \( \chi_{\tau_{n,\alpha} A}(z) \) is equal to \( \chi_{\tilde \sigma_{n, G}}(z) \) multiplied by a polynomial. Consequently, we can bound 
\[
\int_{\|z\| > \epsilon_n \sqrt{n}} \big|\chi_{\tau_{n,\alpha} A}(z)\big|^{2} \dd^{2m} z
\]
by the tail of a Gaussian function. Utilizing the same argument we employed for the fourth term in Claim~\ref{claim:III} within the proof of Theorem~\ref{MainTheoremTrace}, we can conclude that this Gaussian tail is indeed sub-exponentially small.

In the case of \( m=1 \), we do not need the truncated states \( \sigma_n \), and the above proof can be directly framed in terms of \( \rho \). The key point is that to write \( \eqref{BindEdgeWorth_Ent} \), we need the finiteness of moments up to order \( m+3 \). This necessitates the use of the approximating states \( \sigma_n \). For \( m=1 \), under the given assumptions, the moment of order \( m+3 = 4 \) is finite, allowing \( \eqref{BindEdgeWorth_Ent} \) to be directly applied to \( \rho \). In this scenario, \( \epsilon_n \) is independent of \( n \) and is a constant, resulting in the omission of the third term in \( \eqref{splitInt_Ent} \). The rest of the proof closely follows that of the multi-mode case.

%***********************************************************************
\section{Minimality of the assumptions} \label{secMinAssumption}

In this section, we prove Theorem~\ref{thm:examples} and provide examples demonstrating that the moment assumptions in Theorem~\ref{MainTheoremTrace} and Theorem~\ref{MainTheorem} are essentially minimal. To construct these examples we employ ideas and techniques from the literature in the classical case. In particular, we use~\cite{ibragimov1975independent} to prove the first part of the theorem and employ ideas in~\cite{bobkov2013rate} for the second part.

The constructions of the example for both parts are similar and are based on classical mixtures of thermal states. Thus, before discussing each case in detail, we present the construction in a more general form. This general setup will then be used to develop our specific examples for trace distance and relative entropy in subsequent subsections. 

Slightly deviating from our notation in previous sections, for any $s \geq 1/2$, we let  
\[
\tau_{s} = (1 - e^{-\beta_s}) e^{-\beta_s \ac^\dagger \ac}, \qquad \qquad\beta_s = \ln \frac{4s^2 + 1}{4 s^2  -1},
\]
be a thermal state.\footnote{We let $\tau_{1/2}=|0\rangle\!\langle 0|$ be the vacuum state. }
We also let $w(s)$ be a density function on $[1/2, \infty)$ satisfying 
\[
\int s^2 w(s) \dd s = 1.
\]
We then define $\rho$ as a classical mixture of thermal states, according to the density function $w(s)$:
\begin{equation}\label{eq:DefExam}
    \rho := \int_{1/2}^{+\infty} w(s) \tau_s \dd s.
\end{equation}
First, we note that since each thermal state $\tau_s$ is centered, $\rho$ is a centered quantum state. Moreover, regarding the second-order moments of $\rho$, we have
\begin{align}
    \tr (\rho \ac^\dagger \ac) &= \int_{1/2}^{+\infty} w(s) \tr (\tau_s \ac^\dagger \ac) \dd s \nonumber\\
    &= \int_{1/2}^{+\infty}  \Big(2 s^2 - \frac{1}{2}\Big) w(s) \dd s = 2 - \frac 12.\label{eq:example-second-order-m}
\end{align}
Thus, the Gaussification of $\rho$ is equal to $\rho_G=\tau_1$. To construct our desired examples, we aim to identify suitable candidates for the density function $w(s)$. To this end, as a crucial step in our arguments, we establish a bound on the moments of $\rho$ in terms of those of $w(s)$.

\begin{lemma}
    Let $\rho$ be the single-mode bosonic quantum state defined in~\eqref{eq:DefExam}. If $\int_{1/2}^{+\infty} s^\kappa w(s) \dd s < +\infty$ for some $\kappa \geq 2$, then $\rho$ has finite moments of order $\kappa$. 
\end{lemma}

\begin{proof}
    Starting with the moment of $\tau_s$,  we can write
    \begin{align*}
        M_\kappa (\tau_s) &= \tr \Big( \tau_s \big(\ac^\dagger \ac +1\big)^{\frac{\kappa}{2}}\Big) \\
        & = (1 - e^{-\beta_s}) \sum_{k=0}^\infty (k+1)^{\kappa/2} e^{-\beta_s k} \\
        &\leq (1 - e^{-\beta_s}) \int_{1}^\infty x^{\kappa/2} e^{-\beta_s (x-2)} \dd x \\
        &\leq e^{2\beta_s} (1 - e^{-\beta_s}) \int_{0}^\infty x^{\kappa/2} e^{-\beta_s x} \dd x \\
        &= \Gamma(\kappa/2+1) e^{2\beta_s} (1 - e^{-\beta_s}) \beta_s^{-\kappa/2 -1} \\
        &\leq \Gamma(\kappa/2+1) e^{2\beta_s} \frac{1}{(1 - e^{-\beta_s})^{\kappa/2}},
    \end{align*}
where $\Gamma(\cdot)$ denotes the gamma function and the last line follows from $1 - e^{-\beta_s} \leq \beta_s$. Now, using the fact that for $s \geq 1/2$, it holds that $1 - e^{-\beta_s} = 1 - \frac{4s^2 - 1}{4 s^2+1} \geq \frac{1}{4s^2}$, we have
    \[
    M_\kappa (\tau_s) \leq \Gamma(\kappa/2+1) \Big(\frac{4s^2 + 1}{4 s^2-1} \Big)^2 (4 s^2)^{\kappa/2}.
    \]
    We therefore have
    
    \begin{align*}
        M_\kappa (\rho) 
        = \int_{1/2}^\infty w(s) M_\kappa (\tau_s) \dd s 
         \leq \max_{1/2\leq s\leq 1} M_\kappa (\tau_s)  + 2^{\kappa+2}\Gamma(\kappa/2+1) \int_{1}^\infty  s^\kappa w(s)  \dd s.
    \end{align*}
    Then, the desired inequality $M_\kappa (\rho)<+\infty$ follows once we note that $M_\kappa (\tau_s)$ is uniformly bounded for $s\in [1/2,1]$.
\end{proof}

The characteristic function of $\rho$ and its Gaussification $\rho_G=\tau_1$ are easily computed in terms of the characteristic function of thermal states:
\begin{equation}\label{eq:CharExample}
    \chi_\rho(z) = \int w(s) e^{-2 s^2 |z|^2} \dd s, \quad\qquad  \chi_{\rho_G}(z) = e^{-2 |z|^2}.
\end{equation}
Then, using the convexity of the exponential function, and the fact that $w(s)$ is supported in $[1/2, +\infty)$, we have
\begin{equation}\label{eq:TriBoundChar}
    e^{-2 |z|^2} \leq \chi_\rho(z) \leq e^{-\frac 12 |z|^2}.
\end{equation}

\subsection{Proof of part (i) of Theorem~\ref{thm:examples}}\label{secMinAssumption-sub1}

Fix $0<\theta< 1$, and let $w(s)$ be a probability distribution on $[1/2, +\infty)$ such that it has a finite moment of order $\kappa$  for any $\kappa < 3-\theta$,  but its moment of order $3 - \theta$ is infinite, meaning that
\begin{equation}\label{eq:unboundedMoment}
    \int |s|^{3 - \theta} w(s) \dd s = +\infty.
\end{equation}
One way of constructing such a distribution is to take
\[
    w(s) = \begin{cases}
\frac{a(3-\theta)}{s^{4 - \theta}} \quad~~~ s\geq 1 \\
\nu(s) \quad~~~ 1 > s \geq 1/2 \\
0  \quad \quad~~~~ \text{otherwise},
    \end{cases}
\]
where the constant $a > 0$ and the bounded function $\nu: [1/2, 1) \to \mathbb{R}^+$ are chosen in such a way that the assumptions $\int w(s) \dd s = \int s^2 w(s) \dd s = 1$ are satisfied. It is straightforward to verify that $w(s)$ has a finite moment of order $\kappa$ whenever $\kappa<3-\theta$, while its moment of order $(3 - \theta)$ is infinite.

Let $\rho$ be the centered, single-mode bosonic quantum state constructed using $w(s)$ according to~\eqref{eq:DefExam}. We show that $\big\| \rho^{\boxplus n} - \rho_G\big\|_1$ cannot be of order $\mathcal{O}\big(1/\sqrt{n}\big)$. To prove this by contradiction, we assume that 
\begin{align}\label{eq:CE-tr-assumption}
\big\| \rho^{\boxplus n} - \rho_G \big\|_1 = \mathcal O\Big(\frac{1}{\sqrt n}\Big).
\end{align}

Define the function $f : \mathbb C \to \mathbb R$ by
\[
f(z) = \begin{cases}
    e^{2 |z|^2} \quad ~~|z| \leq 1, \\
    0  \quad\quad ~~ \text{ otherwise},
\end{cases}
\]
and let $F$ be the operator with characteristic function $\chi_F(z)=f(z)$. We note that by triangle's inequality, $F$ is a bounded operator since
$$
\|F\| = \bigg\|\frac{1}{\pi} \int_{\mathbb{C}} f(z) D_{-z} \dd^{2} z\bigg\| \leq  \frac{1}{\pi} \int_{\mathbb{C}} |f(z)|\dd^{2} z<+\infty.
$$
Moreover, $F$ is Hermitian since $f(z)$ is real and $f(-z)=f(z)$. Therefore, by the Plancherel identity and H\"{o}lder's inequality, we can write
\begin{align*}
    \bigg| \int_{\mathbb C} \Big(\chi_{\rho^{\boxplus n}}(z) - \chi_{\rho_G}(z)\Big) f(z) \dd^2 z\bigg| & = \pi \tr \Big( F \big(\rho^{\boxplus n} - \rho_G\big)\Big) \\
    &\leq \pi \| F\| \cdot \big\| \rho^{\boxplus n} - \rho_G\big\|_1\\
    &  = \mathcal O \Big(\frac{1} {\sqrt n}\Big),
\end{align*}
where in the last line we use the assumption~\eqref{eq:CE-tr-assumption}. 
Define $\eta(z)$ by
$$\chi_\rho(z) = e^{-2 |z|^2 (1+ \eta(z))}.$$ 
Then, the above inequality can equivalently be written as
\[
\bigg| \int_{| z| \leq 1} \Big(e^{-2 |z|^2 \eta(z/\sqrt n)} -1 \Big) \dd^2 z\bigg| = \mathcal O \Big(\frac{1} {\sqrt n}\Big).
\]
Using~\eqref{eq:TriBoundChar}, it is straightforward to see that $\eta(z) \leq 0$ for any $z \in \mathbb C$. Therefore, applying $e^x-1 \geq x$, the above bound yields 
\[
\int_{| z| \leq 1} |z|^2 \cdot \big|\eta(z/\sqrt n)\big|  \dd^2 z = \mathcal O \Big(\frac{1} {\sqrt n}\Big),
\]
or equivalently
\[
 \int_{| z| \leq 1/\sqrt{n}} |z|^2 \cdot |\eta(z)| \dd^2 z = \mathcal O \Big(\frac{1} {n^2 \sqrt n}\Big).
\]
Thus, we can write
\begin{equation}\label{eq:AuxExa1}
     \int_{| z| \leq \epsilon} |z|^2 |\eta(z)| \dd^2 z = \mathcal O (\epsilon^5),\qquad \quad  \text{ as }  \epsilon \to 0.
\end{equation}
By~\eqref{eq:BoundDerivative-Moments} and~\eqref{eq:example-second-order-m}, the characteristic function  $\chi_\rho(z)$ is twice differentiable near the origin. Thus, using the Taylor expansion of $\chi_\rho(\cdot)$  at the origin and  $\chi_{\rho_G}(z) = e^{-2 |z|^2}$ we find that
\[
\chi_\rho(z)= 1 - 2 |z|^2 +  o(|z|^2) = 1 - 2 |z|^2 (1 + o(1)), \qquad \quad \text{ as }  |z| \to 0.
\]
As a result, as $|z|\to 0$  we have
\begin{align*}
2 |z|^2 |\eta(z)| &= \ln\chi_{\rho}(z) +2|z|^2\\
& = \chi_{\rho}(z)-1 + \mathcal O(|\chi_{\rho}(z)-1|^2)+2|z|^2\\
& = \chi_\rho(z) - 1 + 2 |z|^2 + \mathcal{O}(|z|^4).
\end{align*}
Hence, as a consequence of~\eqref{eq:AuxExa1} and the above expression, we have
\begin{equation}\label{eq:AuxExa2}
     \int_{| z| \leq \epsilon} \Big( \chi_\rho(z) - 1 + 2 |z|^2 \Big) \dd^2 z = \mathcal O (\epsilon^5),  \qquad \quad \text{ as } \epsilon\to 0.
\end{equation}
We note that, based on~\eqref{eq:CharExample}, $\chi_\rho(z)$ is a real-valued function that depends only on $|z|$. Thus, we can write
\begin{align*}
    \int_{| z| \leq \epsilon} \chi_\rho(z) \dd^2 z &= 2 \pi \int_{0}^\epsilon r\chi_\rho(r)  \thinspace \dd r \\
    &= 2 \pi \int_{0}^\epsilon \int_{\mathbb C} rW_\rho(u)  e^{-r(u - \bar u)}  \thinspace \dd^2 u\thinspace \dd r \\
    &= 2\pi \int_{0}^\epsilon \int_{-\infty}^{+\infty}\int_{-\infty}^{+\infty} r  W_\rho(y+ix)  \cos(2 r x ) \thinspace \dd y \thinspace \dd x \thinspace \dd r,
\end{align*}
where $W_\rho(\cdot)$ is the Wigner function of $\rho$, and the second line follows from~\eqref{eq:InvWignerFunc}. On the other hand, using~\eqref{eq:DefExam} and~\eqref{eq:Gaussian-characteristic-Wigner}, we have
\[
W_\rho (y + ix ) = \int_{1/2}^{+\infty}  \frac{1}{2 \pi s^2}e^{-\frac{1}{2 s^2} (x^2 + y^2)} w(s)  \dd s.
\]
Thus, we can write
\[
\int_{| z| \leq \epsilon} \chi_\rho(z) \dd^2 z = 2\pi \int_{0}^\epsilon \int_{-\infty}^{+\infty} r  p(x)  \cos(2 r x ) \dd x \thinspace \dd r, \qquad \quad p(x) = \int_{1/2}^{+\infty} w(s) \frac{1}{\sqrt{2 \pi} s} e^{-\frac{x^2}{2 s^2}} \dd s.
\]
Changing the order of the integrals and computing the inner one, we find that
\[
\int_{| z| \leq \epsilon} \chi_\rho(z) \dd^2 z = 2\pi \int_{-\infty}^{+\infty} p(x)  \bigg( \frac{\epsilon}{2x} \sin(2 \epsilon x) + \frac{1}{4 x^2} \big( \cos(2 \epsilon x) - 1 \big)\bigg) \dd x.
\]
Therefore,~\eqref{eq:AuxExa2} implies
\begin{equation}
    \int_{-\infty}^{+\infty} p(x)  \bigg(  \frac{\epsilon}{x} \sin(2 \epsilon x) + \frac{1}{2 x^2} \big( \cos(2 \epsilon x) - 1 \big) - \epsilon ^2 + \epsilon^4\bigg) \dd x = \mathcal{O}(\epsilon^5), \qquad \quad \text{ as } \epsilon \to 0,
\end{equation}
which, using the fact that $\int p(x) x^2 \dd x = \int w(s) s^2 \dd s = 1$, can be written as 
\begin{equation}\label{eq:AuxExa3}
    \int_{-\infty}^{+\infty} p(x)  \bigg(  \frac{\sin(2 \epsilon x)}{2 \epsilon x} + \frac{\cos(2 \epsilon x) - 1}{(2 \epsilon x)^2} - \frac 12 + \frac{(\epsilon x)^2}{2}\bigg) \dd x = \mathcal{O}(\epsilon^3),  \qquad \quad \text{ as } \epsilon \to 0.
\end{equation}
It is not hard to see that the function 
$$h(t) =  \frac{\sin(t)}{t} + \frac{\cos(t) - 1}{t^2} - \frac 12 + \frac{t^2}{8},$$ 
is non-negative for any $t \in \mathbb R$, and there exists a constant $c >0$ such that $h(t) \geq \frac{1}{16} t^2$ for $|t| \geq c$. This is stated in detail in Lemma~\ref{lem:app-function-h(t)} in Appendix~\ref{App:Example}. Hence,~\eqref{eq:AuxExa3} yields
\[
\int_{|\epsilon x| \geq c} (\epsilon x)^2 p(x)   \dd x = \mathcal{O}(\epsilon^3),  \qquad \quad \text{ as } \epsilon \to 0
\]
or equivalently 
\[
\int_{|\epsilon x| \geq 1} x^2 p(x) \dd x \leq C \epsilon, \qquad \quad \forall \epsilon >0,
\]
for some constant $C>0$. The above equation implies that the moment of order $3 - \delta$ of $p(x)$ is finite for any $0 < \delta \leq 1$. To prove this we write
\begin{align*}
    \int_{-\infty}^{+\infty} |x|^{3-\delta} p(x) \dd x &= \int_{|x|\leq 1} |x|^{3-\delta} p(x) \dd x + \sum_{k=0}^\infty \int_{2^k \leq |x| \leq 2^{k+1}} |x|^{3-\delta} p(x) \dd x \\
    & \leq 1 + \sum_{k=0}^\infty 2^{(k+1)(1 - \delta)}\int_{|x| \geq 2^k} |x|^{2} p(x) \dd x \\
    & \leq 1 + C\sum_{k=0}^\infty  2^{(k+1)(1 - \delta)} \cdot 2^{-k} \\
    &< +\infty.
\end{align*}
Now using the fact that $p(x) = \int w(s) \frac{1}{\sqrt{2 \pi} s} e^{-\frac{x^2}{2 s^2}} \dd s$, for some constant $A_{\delta}>0$ we have
\begin{align*}
    \int_{-\infty}^{+\infty} |x|^{3-\delta} p(x) \dd x &=  \int_{-\infty}^{+\infty} |x|^{3-\delta} \int_{1/2}^{+\infty} w(s) \frac{1}{\sqrt{2 \pi} s} e^{-\frac{x^2}{2 s^2}} \dd s \dd x \\
    & = \int_{-1/2}^{+\infty} w(s)  \int_{-\infty}^{+\infty} |x|^{3-\delta} \frac{1}{\sqrt{2 \pi} s} e^{-\frac{x^2}{2 s^2}} \dd x \dd s \\
    & = A_{\delta} \int_{1/2}^{+\infty} w(s) |s|^{3-\delta} \dd s. 
\end{align*}
Therefore, for any $0 < \delta \leq 1$, we have $\int w(s) |s|^{3-\delta} \dd s < +\infty$, which contradicts our initial assumption~\eqref{eq:unboundedMoment}. As a conclusion, 
$$\sup_n \sqrt n\big\| \rho^{\boxplus n} - \rho_G\big\|_1=+\infty.$$

\subsection{Proof of part (ii) of Theorem~\ref{thm:examples} }

We start with the following lemma that is a quantum counterpart of~\cite[Proposition 7.1]{bobkov2013rate} and plays a crucial role in our argument. It provides a fine estimation for $\rho^{\boxplus n}$ in terms of $\rho_G$ and thermal states.

\begin{restatable}{lemma}{primelemma}
\label{lemrestatable}
  Let $w(s)$ be a probability distribution supported on $[1/2, +\infty)$ satisfying $\int s^2 w(s) \dd s = 1$, that has a finite moment of order $(4 - \delta)$ for some $0<\delta <1$. Let $\rho$ the single-mode bosonic quantum state defined in~\eqref{eq:DefExam}. Then, we have
\begin{align}\label{eq:ExaEntropy1}
    \bigg\|\rho^{\boxplus n} -\bigg( \rho_G + n \int_{1/2}^{+\infty} \big(\tau_{s_n} - \tau_1\big) w(s) \dd s \bigg)\bigg\|= \mathcal O\Big(\frac{1}{n^{2-\delta}}\Big),
\end{align}
where 
$$s_n = \sqrt{1+ \frac{s^2-1}{n}},$$ 
and $\|.\|$ is used to indicate the operator norm.
\end{restatable}

The proof of this lemma is analogous to its classical counterpart and follows by analyzing the characteristic function of $\rho^{\boxplus n}$ and estimating it in terms of those of thermal states. The detailed proof is presented in Appendix~\ref{App:Example}.

Similarly to the previous example, for a fix parameter $0<\theta<1$, we let 
\[
    w(s) = \begin{cases}
\frac{a(4-\theta)}{s^{5 - \theta}} \quad~~~ s\geq 1 \\
\nu(s) \quad~~~ 1 > s \geq 1/2 \\
0  \quad \quad~~~~ \text{otherwise},
    \end{cases}
\]
where the constant $a > 0$ and the bounded function $\nu: [1/2, 1) \to \mathbb{R}^+$ are chosen in such a way that $w(s)$ is a true probability density and  $\int s^2 w(s) \dd s = 1$ is satisfied. It is straightforward to verify that $w(s)$ has a finite moment of order $\kappa$ whenever $\kappa<4-\theta$, while its moment of order $(4 - \theta)$ is infinite. We aim to show that, for the single-mode bosonic quantum state $\rho$ defined as the classical mixture in~\eqref{eq:DefExam} with the above probability distribution $w(s)$, we have
\[
\lim_{n\to +\infty }n  D(\rho^{\boxplus n} \| \rho_G)= + \infty.
\]

Starting with the fact that thermal states are diagonal in the Fock basis, $\rho^{\boxplus n}$ is also diagonal in the Fock basis. Then, using~\eqref{eq:ExaEntropy1}, for a fixed $\delta$ satisfying  $\theta<\delta<1$ and some constant $C=C_{\delta}>0$ we can write
\begin{align}\label{eq:ExaEntropy10}
&\bigg| \bra{k} \rho^{\boxplus n} \ket{k}  -\bigg( \bra{k} \rho_G \ket{k} + n \int_{1/2}^{+\infty} \big(\bra k\tau_{s_n}\ket k - \bra k\tau_1\ket k\big) w(s) \dd s \bigg)\bigg|\nonumber\\
&=\bigg| \bra{k} \rho^{\boxplus n} \ket{k} -\bigg( \frac{2 \times 3^k}{5^{k+1}} + n \int_{1/2}^{+\infty} \bigg( \frac{2}{4 s_n^2 + 1}\Big (\frac{4 s_n^2-1}{4 s_n^2-1}\Big)^k -  \frac{2 \times 3^k}{5^{k+1}}\bigg) w(s) \dd s\bigg)\bigg|\nonumber\\
& \leq \frac{C}{n^{2-\delta}},
\end{align}
for any Fock state $\ket{k}$. Recalling the definition of the quantum relative entropy for diagonal states, we have
\begin{align}
    D\big(\rho^{\boxplus n}\big \| \rho_G\big) &= \sum_{k=0}^\infty \bra{k} \rho^{\boxplus n} \ket{k} \ln \frac{\bra{k} \rho^{\boxplus n} \ket{k}}{\bra{k} \rho_G\ket{k}} \nonumber \\
    &= \underbrace{\sum_{k\leq 10  \ln n} \bra{k} \rho^{\boxplus n} \ket{k} \ln \frac{\bra{k} \rho^{\boxplus n} \ket{k}}{\bra{k} \rho_G\ket{k}}}_{I_n} + \underbrace{\sum_{k > 10  \ln n} \bra{k} \rho^{\boxplus n} \ket{k} \ln \frac{\bra{k} \rho^{\boxplus n} \ket{k}}{\bra{k} \rho_G\ket{k}}}_{J_n}. \label{eq:ExaEntSplit}
\end{align}
In the following we analyze the two terms $I_n, J_n$, separately. 

\bigskip
\begin{claim} \label{claim:In} 
For sufficiently large $n$, we have
\begin{align}
    I_n \geq  -\frac 14  n \ln n  \int_{1/2}^{n^{1/2+\gamma}} \bigg(\Big( \frac{4 s_n^2-1}{4 s_n^2+1}\Big)^{10  \ln n +1 }  - \Big( \frac{4 s_n^2-1}{4 s_n^2+1}\Big)^{2 n^{3\gamma} } \bigg) w(s) \dd s - \frac{2a}{n},\label{eq:FirstTermExaEnt}
\end{align}
where $\gamma = \frac{\theta}{2 (4-\theta)}$. 
\end{claim}
\bigskip

To prove this claim, using the inequality $x \ln x \geq x-1$ and~\eqref{eq:ExaEntropy10} we can write 
\begin{align}
    I_n &\geq \sum_{k\leq 10  \ln n} \Big( \bra{k} \rho^{\boxplus n} \ket{k} - \bra{k} \rho_G \ket{k}\Big) \nonumber \\
    &\geq \sum_{k\leq 10  \ln n} n \int_{1/2}^{+\infty} \Big( \bra{k} \tau_{s_n} \ket{k} - \bra{k} \tau_1 \ket{k}\Big) w(s) \dd s - C \frac{\ln n}{n^{2-\delta}}. \label{eq:ExaEnt1}
\end{align}
 We note that, for any integer $K$, we have 
\begin{align}\label{eq:sum-tau-s-K}
    \sum_{k\leq K} \bra{k} \tau_{s_n} \ket{k} = 1 - \Big( \frac{4 s_n^2 - 1}{4 s_n^2 +1}\Big)^{K+1}.
\end{align}
We also observe that if $s \leq 1$, then $s_n = \sqrt{1 + \frac{s^2-1}{n}} \leq 1$ and $\frac{4 s_n^2 - 1}{4 s_n^2 + 1} \leq  \frac{3}{5}$. This implies that $\sum_{k\leq K} \bra{k} \tau_{s_n} \ket{k} \geq \sum_{k\leq K} \bra{k} \tau_1 \ket{k}$ if $s\leq 1$. Thus, using~\eqref{eq:ExaEnt1}, we get
\begin{equation}\label{eq:ExaEnt2}
    I_n  \geq \sum_{k\leq 10  \ln n} n \int_1^\infty \Big( \bra{k} \tau_{s_n} \ket{k} - \bra{k} \tau_1 \ket{k}\Big) w(s) \dd s - C \frac{\ln n}{n^{2-\delta}}.
\end{equation}
Also, for $\gamma = \frac{\theta}{2 (4-\theta)}$ given in the statement of the claim, we can write
\begin{align}
    &\bigg| \sum_{k\leq 10  \ln n}  n \int_{n^{1/2 + \gamma}}^\infty \Big( \bra{k} \tau_{s_n} \ket{k} - \bra{k} \tau_1 \ket{k}\Big) w(s) \dd s \bigg| \nonumber \\
    &\leq 
     \sum_{k\leq 10  \ln n} n \int_{n^{1/2 + \gamma}}^\infty \Big| \bra{k} \tau_{s_n} \ket{k} - \bra{k} \tau_1 \ket{k}\Big| w(s) \dd s \nonumber \\
     &\leq 2  n  \int_{n^{1/2 + \gamma}}^\infty w(s) \dd s \nonumber\\
     & = \frac{2a}{n},  \label{eq:ExaEnt3}
\end{align}
where in the last line we compute the integral using the explicit form of $w(s)$ given above.
Hence, using~\eqref{eq:ExaEnt2} and~\eqref{eq:ExaEnt3}, for sufficiently large $n$, we have
\begin{align*}
    I_n  &\geq \sum_{k\leq 10  \ln n} n \int_1^{n^{1/2+\gamma}} \Big( \bra{k} \tau_{s_n} \ket{k} - \bra{k} \tau_1 \ket{k}\Big) w(s) \dd s - \frac{2a}{n} \nonumber \\
    & = -n \int_1^{n^{1/2+\gamma}} \bigg( \Big( \frac{4 s_n^2-1}{4 s_n^2+1}\Big)^{10  \ln n +1 } - \Big( \frac 34\Big)^{10  \ln n +1 }\bigg) w(s) \dd s - \frac{2a}{n} \nonumber \\
    &> -n \int_1^{n^{1/2+\gamma}}  \Big( \frac{4 s_n^2-1}{4 s_n^2+1}\Big)^{10  \ln n +1 } w(s) \dd s - \frac{2a}{n}\nonumber \\
        &> -n \int_{1/2}^{n^{1/2+\gamma}}  \Big( \frac{4 s_n^2-1}{4 s_n^2+1}\Big)^{10  \ln n +1 } w(s) \dd s - \frac{2a}{n}.
\end{align*}
Finally, we note that for $s \leq n^{1/2+ \gamma}$ and sufficiently large $n$, we have $s_n^2=1+\frac{s^2-1}{n}< 2n^{2\gamma}$ and  
\begin{align*}
\Big( \frac{4 s_n^2-1}{4 s_n^2+1}\Big)^{2 n^{3\gamma} } &= \Big( 1- \frac{2}{4 s_n^2+1}\Big)^{2 n^{3\gamma} } \\
& \leq \Big( \frac{4 s_n^2-1}{4 s_n^2+1}\Big)^{10\ln n+1 } \Big( 1-\frac{2}{8 n^{2\gamma}+1}\Big)^{ n^{3\gamma} }\\
& \leq \Big( \frac{4 s_n^2-1}{4 s_n^2+1}\Big)^{10\ln n+1 }  e^{-2\frac{n^{3\gamma}}{8n^{2\gamma}+1}}\\
&\leq \frac 12 \Big( \frac{4 s_n^2-1}{4 s_n^2+1}\Big)^{10\ln n+1 }.
\end{align*}
Using this in the previous bound, we obtain the desired bound on $I_n$.

\bigskip
\begin{claim} \label{claim:Jn} 
For sufficiently large $n$, we have
\begin{align}
    J_n \geq  \frac 12  n \ln n  \int_{1/2}^{n^{1/2+\gamma}} \bigg( \Big( \frac{4 s_n^2-1}{4 s_n^2+1}\Big)^{10  \ln n +1 }  - \Big( \frac{4 s_n^2-1}{4 s_n^2+1}\Big)^{2 n^{3\gamma} } \bigg) w(s) \dd s, \label{eq:SecTermExamEnt}
\end{align}
where $\gamma = \frac{\theta}{2 (4-\theta)}$. 
\end{claim}
\bigskip

To prove this, we first observe that for $s \geq 1$, the map
\[
k \mapsto \frac{\bra{k} \tau_{s_n} \ket{k}}{\bra{k} \tau_1 \ket{k}} = \frac{\frac{2}{4s_n^2+1}  \big(\frac{4s_n^2-1}{4s_n^2+1}\big)^k}{\frac 25 \big(\frac 35\big)^k},
\]
is non-decreasing. Thus, for $s \geq \sqrt{3n+1}$, which ensures $s_n \geq 4$, and $k > 10  \ln n$, it holds that
\[
\frac{\bra{k} \tau_{s_n} \ket{k}}{\bra{k} \tau_1 \ket{k}}  \geq \frac{\frac{2}{4s_n^2+1}  \big(\frac{4s_n^2-1}{4s_n^2+1}\big)^{10\ln n}}{\frac 25 \big(\frac 35\big)^{10\ln n}}\geq \frac{5}{4s^2+1}   \Big(\frac{5\times 15}{3\times 17}\Big)^{10\ln n} \geq \frac{1}{s^2}  n^{10/3}.
\]
Thus, for $k > 10  \ln n$ and sufficiently large $n$, we have
\begin{align}
    g_n(k) := \int_{1/2}^{+\infty} \frac{\bra{k} \tau_{s_n} \ket{k}}{\bra{k} \tau_1 \ket{k}} w(s) \dd s &\geq \int_{\sqrt{3n+1}}^{+\infty}\frac{\bra{k} \tau_{s_n} \ket{k}}{\bra{k} \tau_1 \ket{k}} w(s) \dd s \nonumber\\
    &\geq n^{10/3} \int_{\sqrt{3n+1}}^{+\infty} \frac{1}{s^2} w(s) \dd s\nonumber\\
    &= \frac{a(4-\theta)}{6-\theta}   \frac{n^{10/3}}{(3n+1)^{\frac{6-\theta}{2}}}. \label{eq:gn(k)-bound-2}
\end{align}
This, in particular, implies that for sufficiently large $n$ and $k > 10  \ln n$, we have $g_n(k) \geq 3$. Hence, using~\eqref{eq:ExaEntropy10}, we find that
\[
\frac{\bra{k} \tau_{s_n} \ket{k}}{\bra{k} \tau_1 \ket{k}} = 1 + n (g_n(k) -1) -\frac{C}{n^{2-\delta}} \geq n,
\]
for sufficiently large $n$.
This gives a lower bound on $J_n$ as we have
\begin{align}
    J_n & =  \sum_{k > 10 \ln n} \bra{k} \rho^{\boxplus n} \ket{k} \ln \frac{\bra{k} \rho^{\boxplus n} \ket{k}}{\bra{k} \rho_G \ket{k}} \nonumber \\
    &\geq \ln n  \sum_{k > 10 \ln n} \bra{k} \rho^{\boxplus n} \ket{k},\label{eq:ExaEnt4}
\end{align}
for sufficiently large $n$. Also, once again using~\eqref{eq:ExaEntropy10}, for $k > 10  \ln n$ and sufficiently large $n$ we have
\begin{align}
    \bra{k} \rho^{\boxplus n} \ket{k} &\geq - (n-1) \frac 25\Big(\frac 35\Big)^k + n \int_{1/2}^{+\infty} \bra{k} \tau_{s_n} \ket{k} w(s) \dd s -\frac{C}{n^{2-\delta}} \nonumber \\
    &\geq - n  \frac 25  n^{10  \ln\big(\frac 35 \big)} + n  \int_{1/2}^{+\infty} \bra{k} \tau_{s_n} \ket{k} w(s) \dd s  -  \frac {C}{n^{2-\delta}}  \nonumber \\
    &\geq n  \int_{1/2}^{+\infty} \bra{k} \tau_{s_n} \ket{k} w(s) \dd s -  \frac {2 C}{n^{2-\delta}} \nonumber \\
    &\geq \frac 12  n  \int_{1/2}^{+\infty} \bra{k} \tau_{s_n} \ket{k} w(s) \dd s,\label{eq:ExaEnt5}
\end{align}
where the last line follows from~\eqref{eq:gn(k)-bound-2}.  
Then, putting~\eqref{eq:ExaEnt4} and~\eqref{eq:ExaEnt5} together, yields
\begin{align*}
    J_n 
    &\geq \frac 12  n \ln n \sum_{ k > 10  \ln n} \int_{1/2}^{+\infty} \bra{k} \tau_{s_n} \ket{k} w(s) \dd s \\
    &\geq \frac 12  n \ln n \sum_{10  \ln n<k< 2n^{3\gamma}} \int_{1/2}^{+\infty} \bra{k} \tau_{s_n} \ket{k} w(s) \dd s \\
    &= \frac 12  n \ln n  \int_{1/2}^{+\infty} \bigg( \Big( \frac{4 s_n^2-1}{4 s_n^2+1}\Big)^{10  \ln n +1 }  - \Big( \frac{4 s_n^2-1}{4 s_n^2+1}\Big)^{2 n^{3\gamma} } \bigg) w(s) \dd s\\
    &= \frac 12  n \ln n  \int_{1/2}^{n^{1/2+\gamma}} \bigg( \Big( \frac{4 s_n^2-1}{4 s_n^2+1}\Big)^{10  \ln n +1 }  - \Big( \frac{4 s_n^2-1}{4 s_n^2+1}\Big)^{2 n^{3\gamma} } \bigg) w(s) \dd s ,
\end{align*}
where the penultimate line follows from~\eqref{eq:sum-tau-s-K}. This gives the desired bound on $J_n$.

\bigskip

Now we can put the  bounds~\eqref{eq:FirstTermExaEnt} and~\eqref{eq:SecTermExamEnt} on $I_n$ and $J_n$ together to conclude that for sufficiently large $n$ we have
\begin{align*}
    D(\rho^{\boxplus n} \| \rho_G) = I_n + J_n &\geq -\frac 14  n \ln n  \int_{1/2}^{n^{1/2+\gamma}} \bigg(\Big( \frac{4 s_n^2-1}{4 s_n^2+1}\Big)^{10  \ln n +1 }  - \Big( \frac{4 s_n^2-1}{4 s_n^2+1}\Big)^{2 n^{3\gamma} } \bigg) w(s) \dd s-\frac{2a}{n} \\
    &\quad + \frac 12  n \ln n  \int_{1/2}^{n^{1/2+\gamma}} \bigg( \Big( \frac{4 s_n^2-1}{4 s_n^2+1}\Big)^{10  \ln n +1 }  - \Big( \frac{4 s_n^2-1}{4 s_n^2+1}\Big)^{2 n^{3\gamma} } \bigg) w(s) \dd s  \\
    &= \frac 14  n \ln n  \int_{1/2}^{n^{1/2+\gamma}} \bigg( \Big( \frac{4 s_n^2-1}{4 s_n^2+1}\Big)^{10  \ln n +1 }  - \Big( \frac{4 s_n^2-1}{4 s_n^2+1}\Big)^{2 n^{3\gamma} } \bigg) w(s) \dd s  - \frac{2a}{n} \\
   & \geq \frac 14 n \ln n  \int_{\sqrt{n\ln n}}^{\sqrt{2n\ln n}} \bigg( \Big( \frac{4 s_n^2-1}{4 s_n^2+1}\Big)^{10  \ln n +1 }  - \Big( \frac{4 s_n^2-1}{4 s_n^2+1}\Big)^{2 n^{3\gamma} } \bigg) w(s) \dd s  - \frac{2a}{n}.
\end{align*}
We observe that for sufficiently large $n$, and for $s \geq \sqrt{n \ln n}$, we have 
\begin{align*}
\Big( \frac{4 s_n^2-1}{4 s_n^2+1}\Big)^{10  \ln n +1 }  = \bigg( 1-\frac{2}{4 \big(1+\frac{s^2-1}{n}\big)+1}\bigg)^{10  \ln n +1 }
 \geq \Big( 1-\frac{1}{2\ln n}\Big)^{11  \ln n }\geq \Big(1-\frac{1}{2}\Big)^{11/2}>2^{-6}. 
\end{align*}
On the other hand, if $s\leq \sqrt{2n\ln n}$, for sufficiently large $n$ we have
\begin{align*}
\Big( \frac{4 s_n^2-1}{4 s_n^2+1}\Big)^{2n^{3\gamma} }  = \bigg( 1-\frac{2}{4 \big(1+\frac{s^2-1}{n}\big)+1}\bigg)^{2n^{3\gamma} } 
 \leq \Big( 1-\frac{1}{5\ln n}\Big)^{2n^{3\gamma} }\leq e^{-\frac{2n^{3\gamma}}{5\ln n}}<2^{-12}. 
\end{align*}
Using these in the previous bound, we arrive at 
\begin{align*}
    D(\rho^{\boxplus n} \| \rho_G)     &\geq 2^{-9} n \ln n  \int_{\sqrt{n \ln n}}^{\sqrt{2 n \ln n}} w(s) \dd s -  \frac{2a}{n} \\
    & = 2^{-9}\big(1-2^{-(4-\theta)/2}\big)a(n\ln n)^{1- \frac 12(4-\theta)} -  \frac{2a}{n}\\
    & = 2^{-9}\big(1-2^{-(2-\theta/2)}\big)a(n\ln n)^{-(1-\theta/2)} -  \frac{2a}{n}.
\end{align*}
Therefore, since $\theta>0$, we have $\lim_{n\to +\infty }n D(\rho^{\boxplus n} \| \rho_G) = +\infty$.

%***********************************************************************
\section{Final remarks} \label{secCo}

Throughout this paper, we established the optimal convergence rates for the quantum CLT in terms of trace distance and relative entropy under minimal assumptions. In Theorem~\ref{MainTheoremTrace}, we demonstrated that for any centered \( m \)-mode quantum state \( \rho \) with a finite third-order moment,
\[
{\big\| \rho^{\boxplus n} -  \rho_G \big\|}_1 = \mathcal{O} \left( \frac{1}{\sqrt{n}} \right) \quad \text{as } n \rightarrow \infty.
\]

Additionally, concerning relative entropy, we proved in Theorem~\ref{MainTheorem} that 
\[
D\big(\rho^{\boxplus n} \big\| \rho_{G}\big) = \mathcal{O}\left( \frac{1}{n} \right) \quad \text{as } n \rightarrow \infty,
\]
for any centered \( m \)-mode quantum state \( \rho \) with a finite moment of order \( 4 + \delta \), where \( \delta > 0 \) in the multi-mode case and \( \delta = 0 \) if \( m = 1 \).

It is worth noting that in~\cite{beigi2023towards}, it has been shown that these convergence rates cannot be improved, even assuming the finiteness of all moments of the quantum state $\rho$. This is why we refer to them as optimal. Also, as shown in Theorem~\ref{thm:examples}, we prove our results under essentially minimal assumptions that match the best know results in the classical case.

In the course of our arguments, we developed the notion of an Edgeworth-type expansion for quantum states as a tool that can be useful in the study of quantum CLT under various metrics. Beyond its role as a technical tool, this expansion provides valuable insights into the quantum CLT.

Theorem~\ref{boundEntropy-Char} is another technical tool that we developed along our arguments. This theorem provides a method to establish an upper bound on the relative entropy between an arbitrary quantum state and a Gaussian state in terms of the Hilbert--Schmidt distance. We expect this theorem to be of interest across various subjects, particularly when approximating the distance of an arbitrary quantum state from a Gaussian one.

Although we establish optimal convergence rates for the quantum CLT under essentially minimal moment assumptions, several directions remain open. First, while the rates in Theorems \ref{MainTheoremTrace} and \ref{MainTheorem} are optimal with respect to $n$, the number of convoluted states, they are likely not optimal with respect to the number of modes $m$ or the moments of the initial quantum state $\rho$. An important direction for future work is therefore to derive convergence rates that are simultaneously optimal in $n$, $m$, and the relevant moment parameters.

Another natural direction is to investigate convergence rates with respect to other distance measures. In the classical setting, convergence in Rényi relative entropies has been studied in \cite{bobkov2019renyi}. It would be interesting to determine whether quantum analogues of these results hold in the setting of quantum central limit theorems.

\paragraph{Acknowledgements.} 
SB is supported by Iran National Science Foundation (INSF) under project No.\ 4031370. MMG is supported by the National Research Foundation, Singapore and A*STAR under its Quantum Engineering Programme NRF2021-QEP2-02-P05.

%**********************************************************************
{\small
\bibliographystyle{abbrvurl}
\bibliography{CLTBIB}

\begin{thebibliography}{10}

\bibitem{accardi1994quantum}
L.~Accardi, L.~Yg, et~al.
\newblock Quantum central limit theorems for weakly dependent maps (i).
\newblock {\em Acta Mathematica Hungarica}, 63(2):183--212, 1994.

\bibitem{araki_inequality_1990}
H.~Araki.
\newblock On an inequality of {Lieb} and {Thirring}.
\newblock {\em Letters in Mathematical Physics}, 19(2):167--170, Feb. 1990.
\newblock \href {https://doi.org/10.1007/BF01045887}
  {\path{doi:10.1007/BF01045887}}.

\bibitem{arous2013central}
G.~B. Arous, K.~Kirkpatrick, and B.~Schlein.
\newblock A central limit theorem in many-body quantum dynamics.
\newblock {\em Communications in Mathematical Physics}, 321(2):371--417, 2013.
\newblock \href {https://doi.org/10.1007/s00220-013-1722-1}
  {\path{doi:10.1007/s00220-013-1722-1}}.

\bibitem{CRLimitTheorem}
S.~Becker, N.~Datta, L.~Lami, and C.~Rouz{\'e}.
\newblock Convergence rates for the quantum central limit theorem.
\newblock {\em Communications in Mathematical Physics}, 383:223--279, 2021.
\newblock \href {https://doi.org/10.1007/s00220-021-03988-1}
  {\path{doi:10.1007/s00220-021-03988-1}}.

\bibitem{beigi2023towards}
S.~Beigi and H.~Mehrabi.
\newblock Towards optimal convergence rates for the quantum central limit
  theorem.
\newblock {\em Annales Henri Poincar{\'e}}, 2025.
\newblock \href {https://doi.org/10.1007/s00023-025-01609-4}
  {\path{doi:10.1007/s00023-025-01609-4}}.

\bibitem{berry_accuracy_1941}
A.~C. Berry.
\newblock The accuracy of the {Gaussian} approximation to the sum of
  independent variates.
\newblock {\em Transactions of the American Mathematical Society},
  49(1):122--136, 1941.
\newblock \href {https://doi.org/10.1090/S0002-9947-1941-0003498-3}
  {\path{doi:10.1090/S0002-9947-1941-0003498-3}}.

\bibitem{bhatia_matrix_1997}
R.~Bhatia.
\newblock {\em Matrix analysis}.
\newblock Number 169 in Graduate texts in mathematics. Springer, New York,
  1997.
\newblock \href {https://doi.org/10.1007/978-1-4612-0653-8}
  {\path{doi:10.1007/978-1-4612-0653-8}}.

\bibitem{BhRa}
R.~N. Bhattacharya and R.~R. Rao.
\newblock {\em Normal approximation and asymptotic expansions}.
\newblock SIAM, 2010.
\newblock \href {https://doi.org/10.1137/1.9780898719895}
  {\path{doi:10.1137/1.9780898719895}}.

\bibitem{bittel2024optimal}
L.~Bittel, A.~A. Mele, J.~Eisert, and L.~Leone.
\newblock Optimal trace-distance bounds for free-fermionic states: Testing and
  improved tomography.
\newblock {\em PRX Quantum}, 6(3):030341, 2025.
\newblock \href {https://doi.org/10.1103/pzx6-nkfb}
  {\path{doi:10.1103/pzx6-nkfb}}.

\bibitem{bittel2025optimal}
L.~Bittel, F.~A. Mele, A.~A. Mele, S.~Tirone, and L.~Lami.
\newblock Optimal estimates of trace distance between bosonic gaussian states
  and applications to learning.
\newblock {\em Quantum}, 9:1769, 2025.
\newblock \href {https://doi.org/10.22331/q-2025-06-12-1769}
  {\path{doi:10.22331/q-2025-06-12-1769}}.

\bibitem{bobkov2019renyi}
S.~G. Bobkov, G.~Chistyakov, and F.~G{\"o}tze.
\newblock R{\'e}nyi divergence and the central limit theorem.
\newblock {\em Annals of Probability}, 47(1):270--323, 2019.
\newblock \href {https://doi.org/10.1214/18-AOP1261}
  {\path{doi:10.1214/18-AOP1261}}.

\bibitem{bobkov2013rate}
S.~G. Bobkov, G.~P. Chistyakov, and F.~G{\"o}tze.
\newblock Rate of convergence and edgeworth-type expansion in the entropic
  central limit theorem.
\newblock {\em The Annals of Probability}, pages 2479--2512, 2013.
\newblock \href {https://doi.org/10.1214/12-AOP780}
  {\path{doi:10.1214/12-AOP780}}.

\bibitem{BCG}
S.~G. Bobkov, G.~P. Chistyakov, and F.~G{\"o}tze.
\newblock Berry--esseen bounds in the entropic central limit theorem.
\newblock {\em Probability Theory and Related Fields}, 159(3-4):435--478, 2014.
\newblock \href {https://doi.org/10.1007/s00440-013-0510-3}
  {\path{doi:10.1007/s00440-013-0510-3}}.

\bibitem{campbell2013continuous}
E.~T. Campbell, M.~G. Genoni, and J.~Eisert.
\newblock Continuous-variable entanglement distillation and noncommutative
  central limit theorems.
\newblock {\em Physical Review A—Atomic, Molecular, and Optical Physics},
  87(4):042330, 2013.
\newblock \href {https://doi.org/10.1103/PhysRevA.87.042330}
  {\path{doi:10.1103/PhysRevA.87.042330}}.

\bibitem{cramer2010quantum}
M.~Cramer and J.~Eisert.
\newblock A quantum central limit theorem for non-equilibrium systems: exact
  local relaxation of correlated states.
\newblock {\em New Journal of Physics}, 12(5):055020, 2010.
\newblock \href {https://doi.org/10.1088/1367-2630/12/5/055020}
  {\path{doi:10.1088/1367-2630/12/5/055020}}.

\bibitem{CH}
C.~D. Cushen and R.~L. Hudson.
\newblock A quantum-mechanical central limit theorem.
\newblock {\em Journal of Applied Probability}, 8(3):454--469, 1971.
\newblock \href {https://doi.org/10.2307/3212170} {\path{doi:10.2307/3212170}}.

\bibitem{derezinski1985boson}
J.~Derezi{\'n}ski.
\newblock Boson free fields as a limit of fields of a more general type.
\newblock {\em Reports on Mathematical Physics}, 21(3):405--417, 1985.

\bibitem{durrett2019}
R.~Durrett.
\newblock {\em Probability: Theory and Examples}.
\newblock Cambridge University Press, Cambridge, 5th edition, 2019.
\newblock \href {https://doi.org/10.1017/9781108591034}
  {\path{doi:10.1017/9781108591034}}.

\bibitem{Feller}
W.~Feller.
\newblock {\em An introduction to probability theory and its applications,
  volume 2}, volume~81.
\newblock John Wiley \& Sons, 1991.

\bibitem{genoni2008quantifying}
M.~G. Genoni, M.~G. Paris, and K.~Banaszek.
\newblock Quantifying the non-gaussian character of a quantum state by quantum
  relative entropy.
\newblock {\em Physical Review A—Atomic, Molecular, and Optical Physics},
  78(6):060303, 2008.
\newblock \href {https://doi.org/10.1103/PhysRevA.78.060303}
  {\path{doi:10.1103/PhysRevA.78.060303}}.

\bibitem{giri1978algebraic}
N.~Giri and W.~von Waldenfels.
\newblock An algebraic version of the central limit theorem.
\newblock {\em Zeitschrift f{\"u}r Wahrscheinlichkeitstheorie und verwandte
  Gebiete}, 42(2):129--134, 1978.
\newblock \href {https://doi.org/10.1007/BF00536048}
  {\path{doi:10.1007/BF00536048}}.

\bibitem{goderis1989non}
D.~Goderis, A.~Verbeure, and P.~Vets.
\newblock Non-commutative central limits.
\newblock {\em Probability Theory and Related Fields}, 82:527--544, 1989.
\newblock \href {https://doi.org/10.1007/BF00341282}
  {\path{doi:10.1007/BF00341282}}.

\bibitem{goderis1989central}
D.~Goderis and P.~Vets.
\newblock Central limit theorem for mixing quantum systems and the ccr-algebra
  of fluctuations.
\newblock {\em Communications in Mathematical Physics}, 122:249--265, 1989.
\newblock \href {https://doi.org/10.1007/BF01257415}
  {\path{doi:10.1007/BF01257415}}.

\bibitem{hayashi2006quantum}
M.~Hayashi.
\newblock {\em Quantum information theory}.
\newblock Springer, 2006.
\newblock \href {https://doi.org/10.1007/978-3-662-49725-8}
  {\path{doi:10.1007/978-3-662-49725-8}}.

\bibitem{hayashi2009quantum}
M.~Hayashi.
\newblock Quantum estimation and the quantum central limit theorem.
\newblock {\em American Mathematical Society Translations Series}, 2(277):95,
  2009.

\bibitem{hepp1973superradiant}
K.~Hepp and E.~H. Lieb.
\newblock On the superradiant phase transition for molecules in a quantized
  radiation field: the dicke maser model.
\newblock {\em Annals of Physics}, 76(2):360--404, 1973.
\newblock \href {https://doi.org/10.1016/0003-4916(73)90039-0}
  {\path{doi:10.1016/0003-4916(73)90039-0}}.

\bibitem{hepp1973phase}
K.~Hepp and E.~H. Lieb.
\newblock Phase transitions in reservoir-driven open systems with applications
  to lasers and superconductors.
\newblock In {\em Condensed Matter Physics and Exactly Soluble Models: Selecta
  of Elliott H. Lieb}, pages 145--175. Springer, 1973.
\newblock \href {https://doi.org/10.1007/978-3-662-06390-3_13}
  {\path{doi:10.1007/978-3-662-06390-3_13}}.

\bibitem{ibragimov1975independent}
I.~Ibragimov.
\newblock Independent and stationary sequences of random variables.
\newblock {\em Wolters, Noordhoff Pub.}, 1975.

\bibitem{jakvsic2009central}
V.~Jak{\v{s}}i{\'c}, Y.~Pautrat, and C.-A. Pillet.
\newblock Central limit theorem for locally interacting fermi gas.
\newblock {\em Communications in Mathematical Physics}, 285:175--217, 2009.
\newblock \href {https://doi.org/10.1007/s00220-008-0610-6}
  {\path{doi:10.1007/s00220-008-0610-6}}.

\bibitem{jakvsic2010quantum}
V.~Jak{\v{s}}i{\'c}, Y.~Pautrat, and C.-A. Pillet.
\newblock A quantum central limit theorem for sums of independent identically
  distributed random variables.
\newblock {\em Journal of Mathematical Physics}, 51(1):015208, 2010.
\newblock \href {https://doi.org/10.1063/1.3285287}
  {\path{doi:10.1063/1.3285287}}.

\bibitem{KS14}
R.~K\"onig and G.~Smith.
\newblock The entropy power inequality for quantum systems.
\newblock {\em IEEE Transactions on Information Theory}, 60(3):1536--1548,
  2014.
\newblock \href {https://doi.org/10.1109/TIT.2014.2298436}
  {\path{doi:10.1109/TIT.2014.2298436}}.

\bibitem{lami_all_2018}
L.~Lami, K.~K. Sabapathy, and A.~Winter.
\newblock All phase-space linear bosonic channels are approximately {Gaussian}
  dilatable.
\newblock {\em New Journal of Physics}, 20(11):113012, Nov. 2018.
\newblock URL:
  \url{https://iopscience.iop.org/article/10.1088/1367-2630/aae738}, \href
  {https://doi.org/10.1088/1367-2630/aae738}
  {\path{doi:10.1088/1367-2630/aae738}}.

\bibitem{lenczewski1995quantum}
R.~Lenczewski.
\newblock Quantum central limit theorems.
\newblock In {\em Symmetries in Science VIII}, pages 299--314. Springer, 1995.
\newblock \href {https://doi.org/10.1007/978-1-4615-1915-7_22}
  {\path{doi:10.1007/978-1-4615-1915-7_22}}.

\bibitem{lindeberg1922}
J.~W. Lindeberg.
\newblock Eine neue herleitung des exponentialgesetzes in der
  wahrscheinlichkeitsrechnung.
\newblock {\em Mathematische Zeitschrift}, 15:211--225, 1922.
\newblock \href {https://doi.org/10.1007/BF01495306}
  {\path{doi:10.1007/BF01495306}}.

\bibitem{marian2013relative}
P.~Marian and T.~A. Marian.
\newblock Relative entropy is an exact measure of non-gaussianity.
\newblock {\em Physical Review A--Atomic, Molecular, and Optucal Physics},
  88(1):012322, 2013.
\newblock \href {https://doi.org/10.1103/PhysRevA.88.012322}
  {\path{doi:10.1103/PhysRevA.88.012322}}.

\bibitem{matsui2002bosonic}
T.~Matsui.
\newblock Bosonic central limit theorem for the one-dimensional xy model.
\newblock {\em Reviews in Mathematical Physics}, 14(07n08):675--700, 2002.
\newblock \href {https://doi.org/10.1142/S0129055X02001272}
  {\path{doi:10.1142/S0129055X02001272}}.

\bibitem{mele2024learning}
F.~A. Mele, A.~A. Mele, L.~Bittel, J.~Eisert, V.~Giovannetti, L.~Lami,
  L.~Leone, and S.~F. Oliviero.
\newblock Learning quantum states of continuous-variable systems.
\newblock {\em Nature Physics}, 21:2002–2008, 2025.
\newblock \href {https://doi.org/10.1038/s41567-025-03086-2}
  {\path{doi:10.1038/s41567-025-03086-2}}.

\bibitem{michoel2004central}
T.~Michoel and B.~Nachtergaele.
\newblock Central limit theorems for the large-spin asymptotics of quantum
  spins.
\newblock {\em Probability Theory and Related Fields}, 4(130):493--517, 2004.
\newblock \href {https://doi.org/10.1007/s00440-004-0364-9}
  {\path{doi:10.1007/s00440-004-0364-9}}.

\bibitem{Serafini}
A.~Serafini.
\newblock {\em Quantum continuous variables: a primer of theoretical methods}.
\newblock CRC press, 2017.
\newblock \href {https://doi.org/10.1201/9781315118727}
  {\path{doi:10.1201/9781315118727}}.

\bibitem{MaSi}
S.~K. Sirazhdinov and M.~Mamatov.
\newblock On convergence in the mean for densities.
\newblock {\em Theory of Probability \& Its Applications}, 7(4):424--428, 1962.
\newblock \href {https://doi.org/10.1137/1107039} {\path{doi:10.1137/1107039}}.

\bibitem{streater1987entropy}
R.~F. Streater.
\newblock Entropy and the central limit theorem in quantum mechanics.
\newblock {\em Journal of Physics A: Mathematical and General}, 20(13):4321,
  1987.
\newblock \href {https://doi.org/10.1088/0305-4470/20/13/033}
  {\path{doi:10.1088/0305-4470/20/13/033}}.

\bibitem{voiculescu1992free}
D.~V. Voiculescu, K.~J. Dykema, and A.~Nica.
\newblock {\em Free random variables}, volume~1.
\newblock American Mathematical Soc., 1992.

\bibitem{winter_tight_2016}
A.~Winter.
\newblock Tight {Uniform} {Continuity} {Bounds} for {Quantum} {Entropies}:
  {Conditional} {Entropy}, {Relative} {Entropy} {Distance} and {Energy}
  {Constraints}.
\newblock {\em Communications in Mathematical Physics}, 347(1):291--313, Oct.
  2016.
\newblock \href {https://doi.org/10.1007/s00220-016-2609-8}
  {\path{doi:10.1007/s00220-016-2609-8}}.

\end{thebibliography}
}

%************************************************
\appendix

\section{Proof of~\eqref{eq:prop1-third-term}}\label{app:f-g-estimate}

In this appendix we show that 
\begin{align}\label{eq:bound-g-m-t}
\sum_{\beta\cdot k> t} (k_1+1) e^{-\beta\cdot k}\leq \frac{2^m (\beta_1+1)^2\cdots (\beta_m+1)^2}{\nu_\beta^2} (t+1)^m e^{-t},
\end{align}
where $\nu_\beta= \prod_{j=1}^m (1-e^{-\beta_j})$.
To this end, we first define
$$f_{m, t} = \sum_{\beta\cdot k> t} e^{-\beta\cdot k},$$
and by induction on $m$ show that 
\begin{align}\label{eq:bound-f-m-t}
f_{m, t} \leq \frac{2^m(\beta_1+1)\cdots (\beta_m+1)}{\nu_\beta} (t+1)^{m-1}e^{-t}.
\end{align}
The base of induction $m=1$ holds since 
$$f_{1, t} = \sum_{ k> t/\beta_1} e^{-\beta_1k_1}\leq \frac{1}{1-e^{-\beta_1}}e^{-t}\leq \frac{2(\beta_1+1)}{1-e^{-\beta_1}}e^{-t}.$$
For the induction step, we compute
\begin{align*}
f_{m, t} &= \sum_{k_1=0}^{\lfloor t/\beta_1\rfloor} e^{-\beta k_1} \sum_{\beta_2k_2+\cdots+\beta_mk_{m}\geq t-\beta_1k_1} e^{-(\beta_2k_2+\cdots +\beta_mk_{m})}\\
&\quad+ \sum_{k_1> \lfloor t/\beta_1\rfloor} e^{-\beta_1 k_1} \sum_{k_2, \dots, k_{m}=0}^\infty e^{-(\beta_2k_2+\cdots+\beta_mk_{m})}\\ 
&\leq \frac{2^{m-1}(\beta_2+1)\cdots (\beta_m+1)}{\prod_{j=2}^m (1-e^{-\beta_j})}\sum_{k_1=0}^{\lfloor t/\beta_1\rfloor} e^{-\beta k_1} (t-\beta_1k_1+1)^{m-2}e^{-(t-\beta_1k_1)} \\
&\quad + \frac{1}{\prod_{j=2}^m (1-e^{-\beta_j})}\sum_{k_1> \lfloor t/\beta_1\rfloor} e^{-\beta_1 k_1} \\ 
&\leq \frac{2^{m-1}(\beta_2+1)\cdots (\beta_m+1)}{\prod_{j=2}^m (1-e^{-\beta_j})} e^{-t}\sum_{k_1=0}^{\lfloor t/\beta_1\rfloor} (t+1)^{m-2}  + \frac{1}{\nu_\beta}e^{-t} \\ 
&\leq \frac{2^{m-1}(\beta_2+1)\cdots (\beta_m+1)}{\prod_{j=2}^m (1-e^{-\beta_j})} e^{-t} (t/\beta_1+1) (t+1)^{m-2}  + \frac{1}{\nu_\beta}e^{-t}.
\end{align*}
Then, the desired bound~\eqref{eq:bound-f-m-t} is derived once we note that $(t+\beta_1)\leq (\beta_1+1)(t+1)$ and $\frac 1{\beta_1}\leq \frac{1}{1-e^{-\beta_1}}$.

We now apply the same ideas to establish~\eqref{eq:bound-g-m-t}:
\begin{align*}
\sum_{\beta\cdot k> t} (k_1+1) e^{-\beta\cdot k} &= \sum_{k_1=0}^{\lfloor t/\beta_1\rfloor} (k_1+1) e^{-\beta k_1} \sum_{\beta_2k_2+\cdots+\beta_mk_{m}\geq t-\beta_1k_1} e^{-(\beta_2k_2+\cdots +\beta_mk_{m})}\\
&\quad+ \sum_{k_1> \lfloor t/\beta_1\rfloor} (k_1+1)e^{-\beta_1 k_1} \sum_{k_2, \dots, k_{m}=0}^\infty e^{-(\beta_2k_2+\cdots+\beta_mk_{m})}\\ 
&\leq \frac{2^{m-1}(\beta_2+1)\cdots (\beta_m+1)}{\prod_{j=2}^m (1-e^{-\beta_j})}\sum_{k_1=0}^{\lfloor t/\beta_1\rfloor} (k_1+1)e^{-\beta k_1} (t-\beta_1k_1+1)^{m-2}e^{-(t-\beta_1k_1)} \\
&\quad + \frac{1}{\prod_{j=2}^m (1-e^{-\beta_j})}\sum_{k_1> \lfloor t/\beta_1\rfloor} (k_1+1)e^{-\beta_1 k_1} \\ 
&\leq \frac{2^{m-1}(\beta_2+1)\cdots (\beta_m+1)}{\prod_{j=2}^m (1-e^{-\beta_j})} e^{-t}(t/\beta_1+1)^2 (t+1)^{m-2}\\
&\quad + \frac{1}{\prod_{j=2}^m (1-e^{-\beta_j})} e^{-t}  \sum_{k_1=0}^\infty (k_1+\lceil t/\beta_1\rceil+1)e^{-\beta_1 k_1} \\ 
&\leq \frac{2^{m-1}(\beta_1+1)^2(\beta_2+1)\cdots (\beta_m+1)}{\nu_\beta} (t+1)^{m}e^{-t} \\
&\quad + \frac{1}{\prod_{j=2}^m (1-e^{-\beta_j})} e^{-t}  \Big(   (t/\beta_1+1) \frac{1}{1-e^{-\beta_1}}  + \frac{1}{(1-e^{-\beta_1})^2}  \Big).
 \end{align*}
This implies~\eqref{eq:bound-g-m-t}.

%******************
\section{Lemmata needed in Section~\ref{secMinAssumption}}\label{App:Example}

This appendix is devoted to some details required for the proof of Theorem~\ref{thm:examples} in Section~\ref{secMinAssumption}. 

\begin{lemma}\label{lem:app-function-h(t)}
The function
$$h(t) =  \frac{\sin(t)}{t} + \frac{\cos(t) - 1}{t^2} - \frac 12 + \frac{t^2}{8},$$
is non-negative for any $t \in \mathbb R$. Moreover, there exists a constant $c >0$ such that $h(t) \geq \frac{1}{16} t^2$ if $|t| \geq c$. Here, $h(0)$ is defined by $h(0)=\lim_{t\to 0} h(t)$.
\end{lemma}

\begin{proof}
Employing the Taylor expansions of the sine and cosine functions, we have
\begin{align*}
h(t) & = \Big(1 - \frac{1}{3!}t^2 + \frac{1}{5!}t^4-\frac{1}{7!}t^6+\cdots\Big) + \Big(-\frac 12 + \frac{1}{4!}t^2 - \frac{1}{6!}t^4+\frac{1}{8!}t^6-\cdots\Big)  - \frac 12 +\frac 18 t^2\\
& = \Big(\frac{1}{5!} - \frac{1}{6!}\Big)t^4 - \Big(\frac{1}{7!} - \frac{1}{8!}\Big)t^6 + \Big(\frac{1}{9!} - \frac{1}{10!}\Big)t^8-\Big(\frac{1}{11!} - \frac{1}{12!}\Big)t^{10} +\cdots\\
& = \frac{5}{6!}t^4 - \frac{7}{8!}t^6 + \frac{9}{10!}t^8-\frac{11}{12!}t^{10}+\cdots\\
%& = \frac{1}{6!} \Big(5 - \frac{1}{8}t^2\Big)t^4 + \frac{1}{10!} \Big(9 - \frac{1}{12}t^2\Big)t^8+\cdots\\
& = \sum_{k=1}^{+\infty} \frac{1}{(4k+2)!} \Big((4k+1) - \frac{1}{4(k+1)}t^2\Big)t^{4k}.
\end{align*}
Therefore, $h(t)\geq 0$ if $t^2\leq 40$. On the other hand, if $t^2>40$, then
\begin{align*}
h(t)& \geq -\frac{1}{\sqrt{40}} -\frac{2}{40} - \frac 12 + \frac 18t^2 > -1 + \frac 18t^2 > \frac{1}{16}t^2.
\end{align*}
\end{proof}

Next, we present the proof of Lemma~\ref{lemrestatable}, copied here for reader's convenience.  This lemma is used in the proof of part (ii) of Theorem~\ref{thm:examples}. 

\primelemma*

\begin{proof}
Following the notation and arguments in Subsection~\ref{secMinAssumption-sub1} we write
    \begin{align}\label{eq:app-char-rho-w-2}
	\chi_\rho(z) =\exp(-2 |z|^2(1 + \eta(z)))= \int_{1/2}^{+\infty} \omega(s) e^{-2 s^2 |z|^2} \dd s,
    \end{align}
which gives
    \[
        e^{-2 |z|^2 \eta(z)} - 1 = \int w(s) \Big( e^{-2 (s^2-1) |z|^2} - 1\Big) \dd s.
    \] 
Also, recall that
    \begin{equation}\label{eq:ExampleCharBound}
        e^{-2 |z|^2} \leq \chi_\rho(z) \leq e^{-\frac{|z|^2}{2}}.
    \end{equation}

    We observe that if $s \geq   1$, we have $e^{-2 (s^2-1) |z|^2} - 1 \leq 0$. Thus, for $|z| \leq 1/\sqrt 2$, we can write
    \begin{equation}\label{eq:AppExa1}
        e^{-2 |z|^2 \eta(z)} - 1 \leq \int_{s \leq \frac{1}{\sqrt 2 |z|}} w(s) \Big( e^{-2 (s^2-1) |z|^2} - 1\Big) \dd s.
    \end{equation}
    On the other hand, if $1/2\leq s \leq \frac{1}{\sqrt 2 |z|}$, we have $ -2|z|^2 \leq 2 (s^2-1) |z|^2  \leq 1-2|z|^2$ and then $\big|-2 (s^2-1) |z|^2\big| \leq 1$. Thus, using the inequality $e^x -1\leq  x+ x^2$ for $|x| \leq 1$,  for $|z| \leq 1/\sqrt{2}$, and~\eqref{eq:ExampleCharBound}, yields
    \allowdisplaybreaks\begin{align}
    0&\leq e^{2 |z|^2} \chi_\rho(z) - 1\nonumber\\
    & =e^{-2 |z|^2 \eta(z)} - 1\nonumber \\
    &\leq \int_{s \leq \frac{1}{\sqrt 2 |z|}} w(s) \Big( -2 (s^2-1) |z|^2 + 4 (s^2-1)^2 |z|^4\Big) \dd s \nonumber\\
    & = 2 |z|^2  \int_{s \geq \frac{1}{\sqrt 2 |z|}} (s^2-1) w(s) \dd s + 4 |z|^4 \int_{s \leq \frac{1}{\sqrt 2 |z|}} (s^2-1)^2w(s) \dd s \nonumber\\
    & \leq 2 |z|^2  \int_{s \geq \frac{1}{\sqrt 2 |z|}} s^2 w(s) \dd s + 4 |z|^4 \bigg( \int_{s \leq \frac{1}{\sqrt 2 |z|}} s^4 w(s) \dd s + \int_{s \geq \frac{1}{\sqrt 2 |z|}} s^2 w(s) \dd s\bigg) \nonumber\\
    &\leq 4 |z|^2 \int_{s \geq \frac{1}{\sqrt 2 |z|}} s^2 w(s) \dd s + 4 |z|^4  \int_{s \leq \frac{1}{\sqrt 2 |z|}} s^4 w(s) \dd s \nonumber\\
    &= 2 \int_{s \geq \frac{1}{\sqrt 2 |z|}} (\sqrt 2 |z| s)^2 w(s) \dd s +   \int_{s \leq \frac{1}{\sqrt 2 |z|}} (\sqrt 2 |z| s)^4 w(s) \dd s \nonumber\\
    &\leq 2 \int_{s \geq \frac{1}{\sqrt 2 |z|}} (\sqrt 2 |z| s)^{4-\delta} w(s) \dd s +   \int_{s \leq \frac{1}{\sqrt 2 |z|}} (\sqrt 2 |z| s)^{4-\delta} w(s) \dd s \nonumber\\
    &\leq 10  M_{4-\delta}  |z|^{4-\delta}. \label{eq:AppUpperBound1}
    \end{align}
    Here, the second line follows from $\int s^2 w(s) \dd s = \int w(s) \dd s = 1$ which implies $\int_{s \leq \frac{1}{\sqrt 2 |z|}} (s^2-1)w(s)\dd s = -\int_{s \geq \frac{1}{\sqrt 2 |z|}} (s^2-1)w(s)\dd s$. Applying the same trick and ignoring some negative terms, we obtain the third line. In the fourth line, we use $4 |z|^4 \leq 2|z|^2$ for $|z| \leq 1/\sqrt{2}$. Also, in the last line, $M_{4-\delta}$ is the moment of order $4-\delta$ of $w(\cdot)$ which by assumption is finite. 
    
Now, to derive the representation in~\eqref{eq:ExaEntropy1}, we write
\begin{align}
    \rho^{\boxplus n} &= \frac 1\pi \int_{\mathbb C} \chi_{\rho^{\boxplus n}}(z) D_{-z} \dd^2 z \nonumber\\
    &= \frac 1\pi \int_{\mathbb C} \chi_\rho\Big(\frac{z}{\sqrt n}\Big)^n D_{-z} \dd^2 z \nonumber \\
    &= \underbrace{\frac 1\pi \int_{|z| \leq C \sqrt {\ln n}} \bigg( 1 + \Big( e^{2 \big|\frac{z}{\sqrt n}\big|^2}\chi_\rho\Big(\frac{z}{\sqrt n}\Big) - 1\Big)\bigg)^n e^{-2 |z|^2} D_{-z} \dd^2 z}_{A_n} + \underbrace{\frac 1\pi \int_{|z| > C \sqrt {\ln n}} \chi_\rho\Big(\frac{z}{\sqrt n}\Big)^n D_{-z} \dd^2 z}_{B_n},\label{eq:SplitApp1}
\end{align}
where $C>0$ is a sufficiently large constant to be determined. For the second term in~\eqref{eq:SplitApp1}, using~\eqref{eq:ExampleCharBound} we have 
\begin{equation}\label{eq:AppSecondTerm}
   \|B_n\|=\bigg\|\frac 1\pi \int_{|z| > C \sqrt {\ln n}} \chi_\rho\Big(\frac{z}{\sqrt n}\Big)^n D_{-z} \dd^2 z\bigg\| \leq \frac 1\pi \int_{|z| > C \sqrt {\ln n}} e^{-\frac{|z|^2}{2}} \dd^2 z \leq 4 n^{-\frac 14C^2}.
\end{equation}
For the first term in~\eqref{eq:SplitApp1}, we observe that for sufficiently large $n$, we have $|z/\sqrt{n}| \leq 1/\sqrt 2$ since $|z| \leq C \sqrt{\ln n}$. Thus, using~\eqref{eq:AppUpperBound1}, for the region $|z| \leq C \sqrt{\ln n}$ and sufficiently large $n$, we find that 
\begin{equation}\label{eq:ezchi-1/n-bound}
    0\leq e^{2\big|\frac{z}{\sqrt n}\big|^2}\chi_\rho\Big(\frac{z}{\sqrt n}\Big) - 1 \leq 10 M_{4-\delta}  \frac{|z|^{4-\delta}}{n^{\frac{4-\delta}{2}}} \leq \frac 1n,
\end{equation}
where in the last step we once again use $z \leq C \sqrt{\ln n}$. Therefore, using the inequality $0 \leq (1+x)^n - (1 + n x) \leq 2 (n x)^2$, for $0 \leq x \leq \frac 1n$, we arrive at
\begin{align*}
    0 &\leq \bigg( 1 + \Big( e^{2 \big|\frac{z}{\sqrt n}\big|^2}\chi_\rho\Big(\frac{z}{\sqrt n}\Big) - 1\Big)\bigg)^n  - \bigg( 1 + n  \Big( e^{2 \big|\frac{z}{\sqrt n}\big|^2}\chi_\rho\Big(\frac{z}{\sqrt n}\Big) - 1\Big)\bigg)  \\
    &\leq 2 n^2 \bigg( e^{2 \big|\frac{z}{\sqrt n}\big|^2}\chi_\rho\Big(\frac{z}{\sqrt n}\Big) - 1\bigg)^2  \\
    &\leq 200 M_{4-\delta}^2  \frac{|z|^{2 (4-\delta)}}{n^{2-\delta}},
\end{align*}
where in the last line we once again use~\eqref{eq:ezchi-1/n-bound}. 
Now by the above inequality, and employ the triangle's inequality, we derive
\begin{align}
   &\bigg\| A_n- \frac 1\pi \int_{|z| \leq C \sqrt {\ln n}} \Big( 1 + n  \bigg( e^{2 \big|\frac{z}{\sqrt n}\big|^2}\chi_\rho\Big(\frac{z}{\sqrt n}\Big) - 1\Big)\bigg) e^{-2 |z|^2} D_{-z} \dd^2 z\bigg\| \nonumber \\
    & \leq  \frac {200M_{4-\delta}}{\pi n^{2-\delta}}   \int_{z \leq C \sqrt {\ln n}}  |z|^{2 (4-\delta)} e^{-2 |z|^2}  \dd^2 z  \nonumber \\
    &= \mathcal{O}\Big(\frac{1}{n^{2-\delta}}\Big).\label{eq:AppFirstTerm}
\end{align}
We also can once again use~\eqref{eq:ExampleCharBound} to bound 
\begin{align}
    \bigg \| \frac 1\pi &\int_{|z| > C \sqrt {\ln n}} \bigg( 1 + n  \Big( e^{2 \big|\frac{z}{\sqrt n}\big|^2}\chi_\rho\Big(\frac{z}{\sqrt n}\Big) - 1\Big)\bigg) e^{-2 |z|^2} D_{-z} \dd^2 z \bigg \| \nonumber \\
    &\leq \frac 1\pi  \int_{|z| > C \sqrt {\ln n}} \Big| 1 + n  \Big( e^{2\big |\frac{z}{\sqrt n}\big|^2}\chi_\rho\Big(\frac{z}{\sqrt n}\Big) - 1\Big)\Big| e^{-2 |z|^2} \dd^2 z \nonumber \\
    &\leq \frac 1\pi  \int_{|z| > C \sqrt {\ln n}}  \Big(1 + n  \Big( e^{2 \big|\frac{z}{\sqrt n}\big|^2} e^{-\frac 12 \big|\frac{z}{\sqrt n}\big|^2} - 1\Big) \Big) e^{-2 |z|^2} \dd^2 z \nonumber \\
    &= \mathcal{O}(n^{-C/4+1}).\label{eq:AppThirdTerm}
\end{align}
Therefore, putting ~\eqref{eq:AppSecondTerm},~\eqref{eq:AppFirstTerm} and~\eqref{eq:AppThirdTerm} in~\eqref{eq:SplitApp1}, for $C=20$,  we obtain
\begin{align*}
     \bigg \| \rho^{\boxplus n}- \frac 1\pi &\int_{\mathbb C} \bigg( 1 + n  \Big( e^{2 \big|\frac{z}{\sqrt n}\big|^2}\chi_\rho\Big(\frac{z}{\sqrt n}\Big) - 1\Big)\bigg) e^{-2 |z|^2} D_{-z} \dd^2 z \bigg \| =\mathcal{O}\Big(\frac{1}{n^{2-\delta}}\Big). 
    \end{align*}
 On the other hand, employing~\eqref{eq:app-char-rho-w-2} we can verify that  the function  
\begin{align*}
\bigg( 1 + n  \Big( e^{2 \big|\frac{z}{\sqrt n}\big|^2}\chi_\rho\Big(\frac{z}{\sqrt n}\Big) - 1\Big)\bigg) e^{-2 |z|^2} & = e^{-2 |z|^2} + n \Big( e^{2 \big|\frac{z}{\sqrt n}\big|^2-2|z|^2}\chi_\rho\Big(\frac{z}{\sqrt n}\Big) - e^{-2 |z|^2}\Big)\\
& = e^{-2 |z|^2} + n \int_{1/2}^{+\infty} w(s)\Big( e^{-2s_n^2|z^2|} - e^{-2 |z|^2}\Big) \dd s,
\end{align*}
equals the characteristic function of 
     \begin{align*}
    \rho_G +  n \int (\tau_{s_n} - \tau_1) w(s) \dd s,
    \end{align*}
where $s_n = \sqrt{1+ \frac{s^2-1}{n}}$. Using this in the previous bound, gives the desired result.  
\end{proof}
%*********************************************************

%**********************************************************

\end{document}